\documentclass[11pt]{article}

\usepackage{amsmath}
\usepackage{epsfig, amssymb, amsfonts}
\usepackage{amsthm}
\usepackage{amstext}
\usepackage{amsopn}
\usepackage{enumerate}
\usepackage{mathrsfs}
\usepackage{subfigure}
\usepackage{graphicx}
\usepackage{color}
\usepackage[left=2.3cm,top=2cm,right=2.3cm,bottom=2cm, nohead,foot=1cm]{geometry}
\numberwithin{equation}{section}

\newtheorem{theorem}{Theorem}[section]
\newtheorem{lemma}[theorem]{Lemma}
\newtheorem{assumptions}[theorem]{List}

\newtheorem{convention}[theorem]{Convention}
\newtheorem{proposition}[theorem]{Proposition}

\newcommand{\R}{{\mathbb R}}

\newcommand{\p}{{p}}

\newcommand{\mfp}{{\mathfrak{p}}}
\newcommand{\mfq}{{\mathfrak{q}}}
\newcommand{\mcJ}{{\mathcal{J}}}
\newcommand{\tuP}{{\Psi}}

\newcommand{\calN}{{\mathcal{N}}}

\newcommand{\N}{{\mathbb N}}

\DeclareMathOperator{\supp}{supp}

\DeclareMathOperator*{\esssup}{ess\,sup}

\title{A Brownian particle in a microscopic periodic potential}

\author{\textbf{Jeremy Thane Clark}\footnote{jtclark@math.msu.edu} \\  Department of Mathematics, Michigan State University  \\ East Lansing, MI 48824, USA \and  \textbf{Lo\"ic Dubois}\footnote{ldubois@mappi.helsinki.fi} \\ Department of Mathematics, University of Helsinki \\ Helsinki 00014, Finland     }

\begin{document}
\maketitle

\begin{abstract}

We study a model for a massive test particle in a microscopic periodic potential and interacting with a reservoir of light particles.   In the regime considered, the fluctuations in the test particle's momentum resulting from  collisions typically outweigh the shifts in momentum generated by the periodic force, and so  the force is effectively a perturbative contribution.  The mathematical starting point is an idealized reduced dynamics for the test particle given by a linear Boltzmann equation.  In the limit that the mass ratio of a  single reservoir particle to the test particle tends to zero, we show that there is convergence to the Ornstein-Uhlenbeck process under the standard normalizations for the test particle variables.  Our analysis is primarily directed towards bounding the perturbative effect of the periodic potential on the particle's momentum.

\end{abstract}

\section{Introduction}

The Ornstein-Uhlenbeck process offers a homogenized picture for the motion of a massive particle interacting with a gas of lightweight particles at fixed temperature~\cite{Uhlen}.  In this description, the spatial degrees of freedom are driven ballistically by momentum variables which are themselves governed by a diffusion equation that includes a drift term corresponding to the drag felt by the massive particle as it accumulates speed and has more frequent collisions with the gas.  Under diffusive rescaling, the spatial variables converge in law to a Brownian motion.   This result follows by an elementary analysis of the closed formulas available for the Ornstein-Uhlenbeck process~\cite{Nelson}.  The Brownian motion description for the test particle transport is effectively ``more macroscopic" than the Ornstein-Uhlenbeck model  since the fluctuations  in the particle's momentum are integrated into infinitesimal spatial ``jumps" for the Brownian particle.      
  
In the other direction, we may consider derivations of the Ornstein-Uhlenbeck process from models that are ``more microscopic".   These relatively microscopic descriptions may merely be  more complicated stochastic models for the test particle such as a  linear Boltzmann equation, or more fundamentally, a reduced dynamics for the test particle beginning from a full microscopic model that includes the evolution of the degrees of freedom for the gas.   The stochastic model in the former case should be regarded as an intermediary picture between the Ornstein-Uhlenbeck and  the Hamiltonian dynamics arising in some limit; see~\cite{Spohn} for a discussion of the low density limit. In the Boltzmann models, the test particle undergoes a Markovian dynamics, whereas for the  Hamiltonian model including the gas, the randomness is only in the initial configuration, and the resulting dynamics for the test particle given by integrating out the gas is non-Markovian.  In the other direction, the contrast between the Ornstein-Uhlenbeck and the Boltzmann-type dynamics is that the momentum in the Boltzmann case makes discrete jumps, which are individually small in the Brownian limit, corresponding to collisions with gas particles rather than evolving with continuous trajectories according to a Langevin equation as in the Ornstein-Uhlenbeck case.  We refer to the book~\cite{Nelson} for a discussion of these various levels of description for a Brownian particle.  

Rigorous mathematical derivations of the Ornstein-Uhlenbeck process were achieved in~\cite{Hennion,Brunn} from  stochastic models giving an effective description of the test particle as it receives collisions from particles in a background gas.  For models that begin with a full  mechanical Hamiltonian model including the test particle and the gas, derivations of the Ornstein-Uhlenbeck process from the reduced dynamics of the test particle were obtained in~\cite{Holley,Durr,Szasz}. 

In this article we consider the Brownian regime for a stochastic model in which a one-dimensional test particle makes jumps in momentum, interpreted as collisions with a background gas, and is acted upon by a force from an external, spatially periodic potential field.  With the presence of the field, the momentum process is no longer Markovian since it drifts at a rate depending on the particle's position.  The momentum of the particle has two contributions: the total displacement in momentum generated by the field, which is given by a time integral of the force, and the sum of the momentum jumps from collisions.  As a result of  the specific scaling regime considered, which includes the period length of the potential, the force field typically makes a smaller-scale contribution to the test particle's momentum than the fluctuations in momentum due to the jumps identified with ``collisions".   The vanishing of the  force contribution is an averaged effect driven by  the frequent rate at which the test particle is typically passing through the period cells of the potential field.  The Brownian limit of the model to first-order thus yields the same Ornstein-Uhlenbeck process  as if the force were set to zero.  Our analysis is focused on obtaining a sharp upper bound for the influence of the external potential on the momentum of the particle, and our techniques improve those applied to a related model in~\cite{Old}.  Ultimately, the main contributions to the total drift in momentum due to the forcing are made during ``rare" time periods at which the test particle's momentum  returns to ``small" values.  The results of this article are extended in~\cite{Further} to prove that the integral of the force, or net displacement in momentum due to the potential, converges in law to a fractional diffusion whose rate depends on the amount time that the limiting Ornstein-Uhlenbeck process spends at zero momentum, i.e., the local time at zero.  

Our model is a linear Boltzmann dynamics for a one-dimensional particle making elastic collisions with the gas and including a spatially periodic potential.   The jump rate kernel is the one-dimensional case of the formula appearing in~\cite[Ch. 8.6]{Spohn}, which corresponds to a hard-rod interaction between the test particle and a single reservoir particle. However, since the model is one-dimensional, it cannot be derived from a mechanical microscopic dynamics in the Boltzmann-Grad limit.  We thus regard our model as phenomenological, and we argue that the resulting behavior that we find is qualitatively the same as what should be expected in an analogous three-dimensional model for a Brownian particle in a one-dimensional periodic potential.  

We think of our model as corresponding to an experimental situation for a large atom or molecule in a periodic standing-wave light field and interacting with a dilute background gas.  A periodic optical force  on an atom can be produced experimentally by counter-propagating lasers; see, for instance,~\cite{McClelland} or the reviews~\cite{Adams,Morsch}.  A classical treatment of the atom is reasonable in the regime where the potential is effectively weak because the test particle is typically not constrained by the potential and  the coherent quantum effects for the test particle will be suppressed by interactions with the gas.

\subsection{Model and results}

We will consider a one-dimensional particle of mass $M$ interacting with a gas of particles with mass $m$ for $\frac{m}{M}=\lambda\ll 1$ in the presence of a force $-\frac{dV}{dx}(\frac{x}{\lambda})$ for some smooth $V:\R\rightarrow \R^{+}$ with period $a>0$.  We take the phase space density $\tuP_{t,\lambda}(x,p)$ at time $t\in \R^{+}$ to obey a linear Boltzmann equation
\begin{align}\label{TheModel}
\frac{d}{dt}\tuP_{t,\lambda}(x,\,p)=&-\frac{\lambda }{m}p\frac{\partial}{\partial x}\tuP_{t,\lambda}(x,p)+\frac{dV}{dx}\big(\frac{x}{\lambda}\big)\frac{\partial}{\partial p}\tuP_{t,\lambda}(x,p) \nonumber    \\ &+\int_{\R}dp^{\prime}\big(\mcJ_{\lambda}(p^{\prime},p)\tuP_{t,\lambda}(x,p^{\prime})-\mcJ_{\lambda}(p,p^{\prime})\tuP_{t,\lambda}(x,p)     \big),   
\end{align}
where $\mcJ_{\lambda}(p^{\prime},p)$ is a kernel describing the rate of kicks in momentum $p^{\prime}\rightarrow p$ for the massive particle due to collisions with reservoir particles.  Since we are considering an ideal gas, the rates $\mcJ_{\lambda}(p^{\prime},p)$ will be determined by the interaction potential between the test particle and a reservoir particle, the temperature $\beta^{-1}$, the ratio $\lambda=\frac{m}{M}$, and the spatial density $\eta$.  We will take the rates $\mcJ_{\lambda}(p^{\prime},p)$ to correspond to a hard-rod interaction (or alternatively ``hard-point" since the length of the objects does not appear for the one-dimensional linear Boltzmann equation), which has the form 
\begin{align}\label{JumpRates}
  \mcJ_{\lambda}(p^{\prime},p):=  \frac{\eta(1+\lambda)}{2m}\big|p^{\prime}-p\big|\frac{e^{-\frac{\beta}{2m}\big(\frac{1-\lambda}{2}p^{\prime} -\frac{1+\lambda}{2}p    \big)^{2}      } }{(2\pi\frac{m}{\beta})^{\frac{1}{2}} }.   
  \end{align}
 The jump rates $\mcJ_{\lambda}$ are the explicit form of those in equation (8.118) from~\cite{Spohn} for the dimension one case, written in momentum variables rather than velocities.

 We will denote the stochastic process whose probability density evolves according to~(\ref{TheModel}) by $(X_{t},P_{t})$.   Let us also define the process
$$D_{t}:= \int_{0}^{t}dr\frac{dV}{dx}(\frac{X_{r}}{\lambda} ).     $$ 
The process $-D_{t}$ is the cumulative drift in the particle's momentum due to the periodic force field, and hence the momentum at time $t\in \R^{+}$ has the form \begin{align}\label{ContToMom}
P_{t}=P_{0}-D_{t}+J_{t},
\end{align}
 where $J_{t}$ is the sum of all the momentum jumps due to collisions with the gas over the time interval $[0,t]$.

Let $(\mfq_{t},\mfp_{t})\in \R^{2}$ be a process satisfying the Langevin equations
\begin{eqnarray}\label{TheLimit}
d\frak{q}_{t}&=&\frac{1}{m}\frak{p}_{t}dt, \nonumber \\
d\frak{p}_{t}&=&-\gamma \frak{p}_{t}dt+\big(\frac{2m\gamma }{\beta }\big)^{\frac{1}{2}}d\mathbf{B}_{t},
\end{eqnarray}
where $\gamma= 8 \eta \big(\frac{  2   }{\pi m \beta   })^{\frac{1}{2}} $ and $\mathbf{B}_{t}$ is a standard Brownian motion.

The technical assumptions for our main results are the following: 

\begin{assumptions}\label{Assumptions}\text{ }

\begin{enumerate}

\item The potential  $V(x)$ is non-negative, has period $a>0$, and is continuously differentiable.

\item The probability measure $\mu$ on $\R^{2}$  for the initial  location in phase space $(X_{0},P_{0})$ has finite moments. 

\end{enumerate}

\end{assumptions}

The following theorems are  the main results of this article.  Theorem~\ref{ThmMain}  states that as $\lambda\searrow 0$ the momentum process $P_{\frac{t}{\lambda} }$ rescaled by a factor $\lambda^{\frac{1}{2}}$ converges to an Ornstein-Uhlenbeck process.  Theorem~\ref{LemNullDrift} bounds the cumulative drift from the periodic force, although only the weaker limit result~(\ref{Sulfur}) is required for the proof of Thm.~\ref{ThmMain}.  The estimates developed to prove~(\ref{Nickel}) are extended in~\cite{Further} to prove that the process $(\lambda^{\frac{1}{4}}D_{\frac{t}{\lambda} },t\in[0,T])$ converges in law to a time-fractional diffusion as $\lambda\searrow 0$.  In particular, the exponent $\iota=\frac{1}{4}$ is the smallest possible such that the expectation of  $|\lambda^{\iota}D_{\frac{t}{\lambda} }|$ is uniformly bounded for small $\lambda>0$ and $\limsup_{\lambda\rightarrow 0}   \mathbb{E}^{(\lambda)}\big[\sup_{0\leq t\leq T}\big| \lambda^{\frac{1}{4}}D_{\frac{t}{\lambda} }\big| \big]>0$.

\begin{theorem}\label{LemNullDrift}
 There exists a $C>0$ such that for all $\lambda<1$, 
\begin{align}\label{Nickel}
 \mathbb{E}^{(\lambda)}\Big[\sup_{0\leq t\leq T}\big| \lambda^{\frac{1}{4}}D_{\frac{t}{\lambda} }\big| \Big]\leq C. 
 \end{align}
In particular, there is convergence in probability as $\lambda\searrow 0$ given by 
\begin{align}\label{Sulfur}
\sup_{0\leq t\leq T}\big| \lambda^{\frac{1}{2}}D_{\frac{t}{\lambda} }\big| \Longrightarrow 0.
\end{align}

\end{theorem}

 \begin{theorem}\label{ThmMain}  In the limit $\lambda\searrow 0$, there is convergence in law of the process $\lambda^{\frac{1}{2}}P_{\frac{t}{\lambda} }$ to the Ornstein-Uhlenbeck process $\frak{p}_{t}$ over the interval $t\in [0,\,T]$.
 The convergence is with respect to the uniform metric.  
\end{theorem}

Since the position process $X_{t}=X_{0}+\frac{\lambda}{m}\int_{0}^{t}drP_{r}  $ is driven by the momentum process, it follows from Thm.~\ref{ThmMain} that $\lambda^{\frac{1}{2}}X_{\frac{t}{\lambda}}$ converges in law as $\lambda\searrow 0$ to the process $\frak{q}_{t}$ defined in~(\ref{TheLimit}).

\subsection{ Further discussion}\label{SubSecDis}

This article concerns the dynamics of a Brownian particle that feels a force from a one-dimensional periodic  potential.  We focus on a regime in which the potential is ``microscopic".  By ``microscopic", we mean that the potential has an amplitude $\sup_{x,x'}|V(x)-V(x')|$ that is much smaller than the typical kinetic energy $\frac{M}{\beta}=\lambda^{-1}\frac{m}{\beta}$ of the test particle at equilibrium with the heat bath, and that the period $a$ is small enough so that the typical rate at which the particle passes through the period cells $(a^{2}M\beta)^{-\frac{1}{2}}$ is much faster than the rate of energy relaxation  $\approx\lambda \gamma$ for the test particle.  

For our mathematical analysis, the force $F(x)=-\frac{dV}{dx}\big(\frac{x}{\lambda}\big)$ is taken to have a period $a\lambda$ which scales proportionally to the mass ratio $\lambda=\frac{m}{M}$.  This is not essential to these results, and only the broad features described above are critical. The same can be said about the amplitude of the potential.

Theorem~\ref{ThmMain} states that to first approximation under Brownian rescaling,  the momentum is an Ornstein-Uhlenbeck process with no dependence on the potential.  This classical treatment of the particle allows for comparisons with quantum models.  A similar model for a one-dimensional quantum particle was studied in~\cite{Quantum} for which the potential is a periodic $\delta$-potential.  In that case, the singular potential makes a first-order change to the dynamics characterized by spatial subdiffusion caused by quantum  reflections even though the periodic potential is ``microscopic" in a similar sense as described above.   See~\cite{Birkl,Friedman,Kunze} for examples of experimental investigations of quantum reflections of atoms from potentials generated through laser light.   Analogous quantum models with smoother potentials will behave more like their classical counterparts.

A three-dimensional linear Boltzmann dynamics for a particle in a gas of hard spheres and under the influence of a one-dimensional periodic potential will have the same limit result up to the constants as in Thm.~\ref{ThmMain} for the degree of freedom in the direction of the potential. Although the momentum for a single spatial degree of freedom is not Markovian in the linear Boltzmann description, it becomes ``more Markovian" in the Brownian limit as is seen in the limiting three-dimensional Ornstein-Uhlenbeck process.  The rates~(\ref{JumpRates}) can then be replaced by the effective rates that emerge for a single degree of freedom in the three dimensional case, which have the same qualitative features for our purposes.

\subsubsection{Features of the model}

By rescaling the spatial coordinate for the particle by a factor of $\lambda^{-1}$, the master equation~(\ref{TheModel}) becomes
\begin{align}\label{ReTheModel}
\frac{d}{dt}\tuP_{t,\lambda}(x,\,p)=\mathcal{L}_{\lambda}^{*}(\tuP_{t,\lambda})(x,p)= &  -\frac{p }{m}\frac{\partial}{\partial x}\tuP_{t,\lambda}(x,p)+\frac{dV}{dx}\big(x\big)\frac{\partial}{\partial p}\tuP_{t,\lambda}(x,p)    \nonumber \\ &+\int_{\R}dp^{\prime}\big(\mcJ_{\lambda}(p^{\prime},p)\tuP_{t,\lambda}(x,p^{\prime})-\mcJ_{\lambda}(p,p^{\prime})\tuP_{t,\lambda}(x,p)     \big),     
\end{align}
where the generator $\mathcal{L}_{\lambda}^{*}$ is defined by the second equality.  Notice that $\lambda>0$ does not appear in the deterministic terms on the right side of~(\ref{ReTheModel}).    We thus effectively have a particle with Hamiltonian $H(x,p)=\frac{1}{2m}p^{2}+V(x)$ and a $\lambda$-dependent noise.  Note that under the new spatial metric, the velocity of the test  particle is $\frac{p}{m}$ rather than $\frac{p}{M}=\lambda\frac{p}{m}$.   For the purpose of Thm.~\ref{ThmMain} and this article generally, it is sufficient to consider the spatial degree of freedom to be a unit torus $\mathbb{T}=[0,1)$ so that the total state space is $\Sigma:=\mathbb{T}\times \R $.  The equilibrium state for the dynamics on $\Sigma$ is given by the Maxwell-Boltzmann distribution
\begin{align}\label{MaxBolt}
\Psi_{\infty,\lambda}(x,p):= \frac{ e^{-\beta\lambda H(x,p) } }{N(\lambda)   }
\end{align} 
for some normalization $N(\lambda)$. 

After  the spatial stretching,  the drift process in momentum $D_{t}$  has the form
\begin{align}\label{NormDrift}
D_{t}= \int_{0}^{t}dr g(X_{r},P_{r})   
 \end{align}
 for $g:\Sigma\rightarrow \R^{+}$ given by $g(x,p)= \frac{dV}{dx}(x)$.  It is thus an integral functional of an exponentially ergodic Markov process on $\Sigma$; see Appx.~\ref{AppendErgod} for a discussion of the exponential ergodicity of the dynamics.   Nonetheless, a central limit theorem for $\lambda^{\frac{1}{4}}D_{\frac{t}{\lambda}}$  does not follow from the limit  theory for integral functionals of ergodic Markov processes~\cite{Landim} because the relaxation to the state~(\ref{MaxBolt}) only occurs on the time scale $\lambda^{-1}\gg 1$.  Indeed, there must be many collisions with reservoir particles before  there is memory loss for the heavy particle.   As will be explained in Sects.~\ref{SecRoughPict}-\ref{SecProofStrat}, the analysis of $\lambda^{\frac{1}{4}}D_{\frac{t}{\lambda}}$ is more related  to the  limit theory for martingales whose bracket processes are additive functionals of a null-recurrent Markov process~\cite{Hopfner}.  This is due to the fact that the fluctuations in $D_{t}$ accumulate mainly during time intervals in which $|P_{t}|$ is much smaller than the typical momentum size $\big(\frac{m}{\lambda\beta}\big)^{\frac{1}{2}}\gg \big(\frac{m}{\beta}\big)^{\frac{1}{2}}$ for the equilibrium state $\Psi_{\infty,\lambda}$.

\subsubsection{Rough picture of the behavior in the Brownian regime $\lambda\ll 1$} \label{SecRoughPict}

Since the equilibrium state of the dynamics is given by the Maxwell-Boltzmann distribution~(\ref{MaxBolt}), the typical energy for the particle when $\lambda\ll 1$ will be on the order $\lambda^{-1}$.  Moreover,  the potential $V(x)$ is bounded, so most of the energy will be in the kinetic component $\frac{1}{2m}p^{2}$ corresponding to momenta $p$ with absolute value on the order of $ \big(\frac{m}{\lambda\beta}\big)^{\frac{1}{2}}\gg \big(\frac{m}{\beta}\big)^{\frac{1}{2}} $.  The jump rates $\mathcal{J}_{\lambda}(p,p')$ for $|p|=\mathit{O}( \lambda^{-\frac{1}{2}})$ are approximately 
\begin{align}\label{RateApprox}  
\mathcal{J}_{\lambda}(p,p')= j(p-p')+\lambda  \frac{\beta}{4m}\big((p)^2-(p')^{2}\big) j(p-p')+\mathit{O}(\lambda), 
\end{align}
where the idealized rates $j(p)$ have the form
\begin{align}\label{LambdaZero}
 j(p) :=   \frac{\eta}{2m}|p|\frac{e^{-\frac{\beta}{2m}p^{2}    } }{(2\pi\frac{m}{\beta})^{\frac{1}{2}} }.    
 \end{align}
 The second term on the right side of~(\ref{RateApprox}) is $\mathit{O}(\lambda^{\frac{1}{2}})$ by our assumption $|p|=\mathit{O}( \lambda^{-\frac{1}{2}})$.  The physical meaning behind the approximation~(\ref{RateApprox}) is that the gas reservoir particles are typically moving at speeds on the order $(m\beta)^{-\frac{1}{2}}$ which is greater than the typical speed of the test particle  $(M\beta)^{-\frac{1}{2}}= \lambda^{\frac{1}{2}}(m\beta)^{-\frac{1}{2}}\ll (m\beta)^{-\frac{1}{2}} $ \footnotemark 
.  The statistics for the momentum transfers from the gas thus do not depend strongly on the momentum of the test particle, and have approximately a convolution form as in the zeroth-order term in~(\ref{RateApprox}).   The zeroth-order approximation in~(\ref{RateApprox}) suggests that the collision component $J_{t}$ of the momentum~(\ref{ContToMom}) is typically behaving as an unbiased random walk with increments having density $j(v)$.  Based on this reasoning,   $\lambda^{\frac{1}{2}}J_{\frac{t}{\lambda}}$ should converge   to a Brownian motion with diffusion constant $\frac{2m\gamma}{\beta}$ as $\lambda\searrow 0$ by the central limit theorem.  However, the first-order term in~(\ref{RateApprox}) generates a drift for  $\lambda^{\frac{1}{2}}J_{\frac{t}{\lambda}}$ that is retained as $\lambda\searrow 0$ and converges to a limit by a law of large numbers.   This can be seen in the friction term appearing in the Langevin equation~(\ref{TheLimit}).  

\footnotetext{These velocities refer to the original length scale, before stretching by a factor of $\lambda^{-1}$.}

According to the heuristics above,  $J_{\frac{t}{\lambda}}$ should  typically be found on the scale $\lambda^{-\frac{1}{2}}$ when  $\lambda \ll 1$ and $t\in [0,T]$, and we will now argue that $D_{\frac{t}{\lambda}}$ should typically be $\mathit{O}(\lambda^{-\frac{1}{4}}) $.  We can parse the integral for $D_{t}$ according to the collision times $t_{n}$ as
$$ D_{t}=\int_{t_{\calN_{t}}}^{t}   dr\frac{dV}{dx}(X_{r})   + \sum_{n=1}^{\calN_{t}}\int_{t_{n-1}}^{t_{n}}dr\frac{dV}{dx}(X_{r}),  $$
where $t_{0}=0$ and  $\calN_{t}$ is the number of collisions up to time $t$.   Between collisions from the gas, the particle evolves deterministically according to the Hamiltonian $H(x,p)=\frac{1}{2}p^{2}+V(x)$, and Newton's equations give
\begin{align}\label{MomentumIncrement}
 P_{t_{n}^{-}}-P_{t_{n-1} }=-\int_{ t_{n-1} }^{t_{n }  }dr\frac{dV}{dx}(X_{r}).   
 \end{align}
If $H(X_{t_{n-1}},P_{t_{n-1}})>2\sup_{x}V(x)$, then the momentum will not change signs over the interval $[t_{n-1},t_{n})$, and
$$ \big|P_{t_{n}^{-}}-P_{t_{n-1} }\big|= \Big||P_{t_{n-1} }|-\sqrt{ P_{t_{n-1}}^{2}+2V(X_{t_{n-1} })-2V(X_{t_{n} })}\Big|\leq \frac{2\sup_{x}V(x)}{|P_{t_{n-1} }|},$$ 
which follows from the conservation of energy and the quadratic formula.  Thus, when $|P_{t_{n-1}}|$ is on the typical order $\propto \lambda^{-\frac{1}{2}}$,  the increment~(\ref{MomentumIncrement}) of the momentum drift is $\mathit{O}(\lambda^{\frac{1}{2}})\ll 1$.  

 In fact there is another critical feature of an ergodic nature that makes the contributions~(\ref{MomentumIncrement}) to $D_{t}$ even smaller when $|P_{t_{n-1} }|\gg 1$.  There is an ergodicity on the spatial torus relating to the fact that when the momentum is high, then the particle revolves quickly around the torus, and its location at the time of the next collision is close to uniform over $\mathbb{T}$.  This idea can be used to show that the mean for $\int_{ t_{n-1} }^{t_{n }  }dr\frac{dV}{dx}(X_{r})$ is $\mathit{O}(|P_{t_{n-1} }|^{-2})$ when given the information known up to time $t_{n-2}$.  In other words, besides the increments~(\ref{MomentumIncrement}) just being small when $|P_{t_{n-1}}|\gg 1$,  $D_{t}$ is also behaving like a martingale since the increments are close to  being uncorrelated with mean zero.   Thus, there is a central limit theorem-like cancellation among the terms.  
This  motivates that the contribution to $D_{\frac{t}{\lambda}}$ from time intervals where $|P_{r}|\propto \lambda^{-\frac{1}{2}}$ is $\mathit{o}(\lambda^{-\frac{1}{2}})$ since, for fixed small $\epsilon>0$,
\begin{align*}
\mathbb{E}\Big[\Big( &\sum_{n=1}^{\calN_{\frac{t}{\lambda} }}\chi\big( |P_{t_{n-1}}|    \geq \epsilon \lambda^{-\frac{1}{2}}\big) \int_{t_{n-1}}^{t_{n}}dr\frac{dV}{dx}(X_{r}) \Big)^{2}    \Big]\\ &=\mathit{O}\Big(\mathbb{E}\Big[\sum_{n=1}^{\calN_{\frac{t}{\lambda} }}\chi\big( |P_{t_{n-1}}|    \geq \epsilon \lambda^{-\frac{1}{2}}\big) \Big( \int_{t_{n-1}}^{t_{n}}dr\frac{dV}{dx}(X_{r}) \Big)^{2}    \Big]\Big)
\\ &\leq \epsilon^{-2}\mathbb{E}\big[\calN_{\frac{t}{\lambda}}\big]\mathit{O}(\lambda)=\mathit{O}(1),
\end{align*}
because the collisions occur with a frequency on the order of one per unit time $\mathbb{E}\big[\calN_{\frac{t}{\lambda}}\big]=\mathit{O}(\frac{t}{\lambda})$.  These contributions disappear for the normalized expression $\lambda^{\frac{1}{4}}D_{\frac{t}{\lambda}}$.  For technical reasons, our analysis of these facts is actually performed with a different set of artificially introduced stopping times rather than the collision times;  see Sect.~\ref{SecMomentumDrift}.   
 
 The above arguments motivate that $D_{\frac{t}{\lambda}}$ spends the greater portion of the time interval $t\in [0,T]$ behaving as a constant, or, said differently, its larger fluctuations are typically concentrated on a small fraction of the interval $[0,T]$.   Let us consider the order of the contributions to $D_{t}$ that are likely to occur for the periods of time when $P_{r}$ returns to the region around the origin, that is, $|P_{r}|=\mathit{O}(1)$.  If $P_{r}$ is behaving roughly as a random walk for $t\in[0,\frac{T}{\lambda}]$ with some very weak friction, then we expect that $P_{r}$ spends on the order of $\lambda^{-\frac{1}{2}}$ time in the vicinity of the origin.  If there are central limit theorem-like cancellations between the increments $\int_{ t_{n-1} }^{t_{n }  }dr\frac{dV}{dx}(X_{r})$ in those time periods, then  $D_{\frac{t}{\lambda}}$ should be expected to be on the scale $\lambda^{-\frac{1}{4}}$.

\subsubsection{Techniques and strategy of the proof }\label{SecProofStrat}

The main difficulty in showing that $\lambda^{\frac{1}{2}}P_{\frac{t}{\lambda}}$ converges in law to the Ornstein-Uhlenbeck process $\frak{p}_{t}$ is to show that the component $ D_{\frac{t}{\lambda}}$ of the momentum is typically $\mathit{o}(\lambda^{-\frac{1}{2}})$ for $t\in [0,T]$.  As indicated by the heuristics of Sect.~\ref{SecRoughPict}, we should expect, in fact, that typically $\sup_{0\leq t\leq T}| D_{\frac{t}{\lambda}}|$ is $\mathit{O}(\lambda^{-\frac{1}{4}})$.  

One of the main ingredients in our analysis is a splitting technique that consists in introducing an artificial ``atom" into the state space  by embedding the original process as a component of a process with an enlarged state space. In principle, the benefit for having  an extended state space with an atom is that the trajectories for the process $S_{t}$ can be decomposed into a series of i.i.d. parts, i.e., \textit{life cycles}, corresponding to time intervals  $[R_{n},R_{n+1})$ where $R_{n}$ are the return times to the atom.  This would allow the integral functional $D_{t}$ to be written as a pair of boundary terms plus a sum of i.i.d. random variables with a random number of terms.  For Markov chains such a technique for embedding an atom was developed independently in~\cite{Nummelin} and~\cite{Athreya} and is referred to as  \textit{Nummelin splitting} or merely \textit{splitting}.  When it comes to splitting a Markov process, there are different schemes available.  In~\cite{Hopfner} there is a sequence of split processes constructed which contain marginal processes that are arbitrarily close to the original process.  The construction in~\cite{Loch} involves a larger state space $\Sigma\times [0,1 ]\times \Sigma$ although an exact copy of the original process is embedded as a marginal. The idea that splitting constructions could be used as a tool to prove certain limit theorems for Markov  processes was suggested in an unpublished paper~\cite{TouatiUnpub}.

We use a truncated version of the split process introduced in~\cite{Loch}.  The split process is not Markovian itself, but contains an embedded chain (the split resolvent chain) which is Markovian.  The life cycles for the process are not completely independent in this construction because  there are correlations between successive life cycles.  The details of the construction are explained in Sect.~\ref{SecNum}.  The original process $S_{t}$, which lives in $\Sigma=\mathbb{T}\times \R$,   is embedded as a component of $\tilde{S}_{t}=(S_{t},Z_{t}) \in \tilde{\Sigma}=\Sigma\times \{0,1\}$.  The process $D_{t}$ can be written as four boundary terms plus a martingale
\begin{align}\label{FirstLC}
D_{t}&= \big(\text{Sum of boundary terms})+\tilde{M}_{t} \\ \nonumber
\tilde{M}_{t}&=  \sum_{n=1}^{ \tilde{N}_{t} }\Big(\int_{R_{n}}^{R_{n+1}}dr\frac{dV}{dx}(X_{r})-\big(\frak{R}^{(\lambda)}\,\frac{dV}{dx}\big)( S_{R_{n}} ) +\big(\frak{R}^{(\lambda)}\,\frac{dV}{dx}\big)( S_{R_{n+1}} ) \Big) ,
\end{align}
where $\tilde{N}_{t}$ is the number of returns to the atom $\Sigma\times 1$ to have occurred before time $t$, and  $\frak{R}^{(\lambda)}:L^{\infty}(\Sigma)\rightarrow  L^{\infty}(\Sigma)$ is the reduced resolvent of the backwards generator $\mathcal{L}_{\lambda}$.    The boundary terms are  
$$  \int_{0}^{R_{1}}dr\frac{dV}{dx}(X_{r})-\int_{t}^{R_{\tilde{N}_{t}+1} }dr\frac{dV}{dx}(X_{r})+\big( \frak{R}^{(\lambda)}\,\frac{dV}{dx}\big)( S_{R_{1}} ) -  \big(\frak{R}^{(\lambda)}\,\frac{dV}{dx}\big)( S_{R_{\tilde{N}_{t}+1}} )  . $$
The interjection of the telescoping terms $ \big(\frak{R}^{(\lambda)}\frac{dV}{dx}\big)( S_{R_{n}} )$ removes the correlations  between successive life cycles.  The fact that the increments  of $\tilde{M}_{t}$ have mean zero with respect to the information known up to time $R_{n}$ is a consequence of the splitting construction and the fact that the observable $\frac{dV}{dx}$ has mean zero in the equilibrium state $\Psi_{\infty,\lambda}$.  The process  $\tilde{M}_{t}$ is a martingale with respect to its own filtration, and this opens the possibility of applying Doob's maximal inequality to bound the fluctuations of $\tilde{M}_{t}$.

The martingale $\tilde{M}_{t}$ is a variant of the martingale $\tilde{M}'$ below that is usually employed when studying limit theorems for integral functionals of  Markov processes:
\begin{align}\label{AltMart}
\tilde{M}_{t}'=  \big(\frak{R}^{(\lambda)}\,\frac{dV}{dx}\big)(S_{t})- \big(\frak{R}^{(\lambda)}\,\frac{dV}{dx}\big)(S_{0}) +D_{t}.     
\end{align}
The martingale $\tilde{M}'$ has predictable quadratic variation
$$\langle \tilde{M}'\rangle_{t}=\int_{0}^{t}dr\int_{\R}dp'\Big(  \big(\frak{R}^{(\lambda)}\,\frac{dV}{dx}\big)(X_{r},p')- \big(\frak{R}^{(\lambda)}\,\frac{dV}{dx}\big)(X_{r},P_{r})      \Big)^{2}\mathcal{J}_{\lambda }(P_{r},p'). $$
Using the martingale~(\ref{AltMart}) would require showing some decay that is uniform in $\lambda<1$ for increments $| \big(\frak{R}^{(\lambda)}\,\frac{dV}{dx}\big)(x,p')- \big(\frak{R}^{(\lambda)}\,\frac{dV}{dx}\big)(x,p)    |$ when $|p|,|p'| $ are large and $|p-p'|=\mathit{O}(1)$.  However, it is not clear to us how to obtain the necessary bounds on the resolvent, and the methods here are designed to  exploit the time-averaging of the oscillatory process $\frac{dV}{dx}(X_{r})$  as suggested by the heuristics in Sect.~\ref{SecRoughPict}.  Our technique is based on having bounds for a generalized resolvent $U^{(\lambda)}_{h}:L^{\infty}(\Sigma)\rightarrow L^{\infty}(\Sigma)$ of the form
$$ \big(U^{(\lambda)}_{h}g\big)(s):=\mathbb{E}^{(\lambda)}_{s}\Big[ \int_{0}^{\infty}dt e^{-\int_{0}^{t}dr h(S_{r})}g(S_{t} )    \Big],    $$
where $h$ is a non-negative function with compact support, and the function $g \equiv g_{\lambda}$ essentially has the form
$$g_{\lambda}(s)=\Big|\mathbb{E}_{s}^{(\lambda)}\Big[\int_{0}^{\infty}dt\,t\,e^{-t}\frac{dV}{dx}(X_{t})    \Big]\Big|.    $$
The operator $U^{(\lambda)}_{h}$ arises in the study of recurrence for Markov processes and has been referred to as the \textit{state-modulated resolvent}~\cite{Meyn}.   Analysis of $U^{(\lambda)}_{h}$ for our dynamics is contained in~\cite{Resolvent}.

\subsubsection{The unit conventions and organization of the article} 

Throughout the remainder of the article, we will remove units by setting $\beta=a=m=1$,  and picking $\eta$ such that $ \gamma=\frac{1}{2}$; recall that $\gamma$ is defined below~(\ref{TheLimit}).  We assume List~\ref{Assumptions} in all theorems, lemmas, etc. unless otherwise stated.\vspace{.5cm} 

Most of the analysis, Sects. 2- 5, is concerned with the proof of Thm.~\ref{LemNullDrift}.  The proof of Thm.~\ref{ThmMain} given Thm.~\ref{LemNullDrift} is relatively nontechnical.  The contents of later sections are roughly characterized by the following:

\begin{itemize}
\item Section~\ref{SecNum} presents the splitting structure that allows us to decompose the dynamics into a series of life cycles as sketched in Sect.~\ref{SecProofStrat}.     

\item  Section~\ref{SecReturns} is directed towards gaining  control over the frequency and duration of life cycles in the limit $\lambda\searrow 0$.   

\item Section~\ref{SecSumFun} demonstrates how to bound the fluctuations of the  integral functional $\int_{0}^{t}dr\frac{dV}{dx}(X_{r})$ over the time period of a single life cycle.

\item  Sections~\ref{SecDrift} and~\ref{BrownProof} contain the proofs respectively for Thms.~\ref{LemNullDrift} and~\ref{ThmMain}.

\item Various proofs are placed in Sect.~\ref{SecMiscProof} to avoid diverting the reader from the main points in earlier sections.

\end{itemize}

\section{Nummelin splitting}\label{SecNum}
   
The split process that we define here is  a truncated version of that in~\cite{Loch}. In the context of a larger probability space, the drift in momentum $D_{t}=\int_{0}^{t}dr\frac{dV}{dx}(X_{r})$ may be viewed as a  martingale plus a few small ``boundary" terms.  This allows us to apply martingale techniques.   For those familiar with the terminology related to Nummelin splitting, we outline the extension of the process as follows:  We introduce a resolvent chain embedded in the original process, we split the chain using Nummelin's technique, and  we extend the resolvent chain to a non-Markovian process which contains an embedded version of the original process.      

We will begin with a generic discussion of the splitting structure by assuming that we have a function $h:\Sigma\rightarrow [0,1]$ and a probability measure $\nu$ on $\Sigma$ satisfying the inequality~(\ref{NummelinCrit}).  The specific $h$ and $\nu$ that we use in this article are defined below in Conv.~\ref{NumConv}. 

Let $(e_m)$ be a sequence of mean one exponential random variables that are independent of each other and of the process $(X_t,\,P_t)$, and let $\tau_n := \sum_{m=1}^n e_m$ with the convention $\tau_0 = 0$. The $\tau_{n}$ will be referred to as the \textit{partition times}.  Define $\mathbf{N}_{t}$ to be the number of non-zero $\tau_{n}$ less than $t$, and the Markov chain $\sigma_{n}:=(X_{\tau_{n}},P_{\tau_{n}})\in \Sigma$,  
 which is referred to as the \textit{resolvent chain}.  The resolvent chain has the same invariant probability density as the original process.  Let $\mathcal{T}$ be the  transition kernel for the chain, acting on functions from the left and measures from the right.  Recall that for a Markov chain, an {\em atom} is a nonempty set $\alpha$ such that the probability transitions starting from a point  $s\in\alpha$ are independent of $s$. An atom is said to be {\em recurrent} if, when starting from a point in the atom, the probability of returning to the atom in the future is one.  In general, an Harris recurrent Markov chain with invariant measure $\mu$ does not necessarily have a recurrent atom $\alpha $ with positive weight $\mu(\alpha)>0$. The splitting technique that we outline presently, originally due to Nummelin, allows us to create a recurrent atom for an Harris recurrent Markov chain though a minorization condition~(\ref{NummelinCrit}). The idea is to extend the state space  $\Sigma$ to   $\tilde{\Sigma}:=\Sigma\times \{0,1\}$  in order to construct a chain $(\tilde{\sigma}_{n})\in \tilde{\Sigma}$ with a recurrent atom and having the statistics for $(\sigma_{n})$ embedded in the first component of $(\tilde{\sigma}_{n})$.  Let $\nu$ be a probability measure on $\Sigma$ and $h:\Sigma\rightarrow [0,1) $ be such that       
\begin{align}\label{NummelinCrit}
 \mathcal{T}(s_{1},ds_{2})\geq h(s_{1})\nu(ds_{2}).   
 \end{align}
  We have the following transition rates from the state $(s_{1},z_{1})\in \tilde{\Sigma}$ to the infinitesimal region $(ds_{2}, z_{2})$:
\begin{align*}
\tilde{\mathcal{T}}( s_{1},z_{1}; ds_{2},z_{2})=\left\{  \begin{array}{ccc} \frac{1-h(s_{2}) }{1-h(s_{1})}  \big( \mathcal{T}-  h\otimes \nu \big) (s_{1},ds_{2})            & \hspace{.5cm} & z_{1}=z_{2}=0 , \\  \frac{h(s_{2}) }{1-h(s_{1})}  \big( \mathcal{T}-  h \otimes \nu\big) (s_{1},ds_{2})         &  \hspace{.5cm}& z_{1}=1-z_{2}=0 ,   \\ \big(1- h(s_{2})\big)  \nu (ds_{2})             & \hspace{.5cm} & z_{1}=1- z_{2}=1 , \\ h(s_{2}) \nu(ds_{2})    & \hspace{.5cm}& z_{1}=z_{2}=1. 
    \end{array} \right.  
\end{align*}
Given a measure $\mu$ on $\Sigma$, we refer to its \textit{splitting} $\tilde{\mu}$ as the measure on $\tilde{\Sigma}$ given by 
\begin{align}\label{SplitMeasure}
 \tilde{\mu}(ds,z)=  \chi(z=0)\big(1-h(s)\big)\mu(ds)+\chi(z=1)h(s)\mu(ds). 
 \end{align}
In particular, the split chain is taken to have initial distribution given by the splitting of the initial distribution for the original (pre-split) chain. The set $\Sigma\times 1$ is an atom since the transition measure from $(s_1,1)$ is independent of $s_1$. Moreover, it is a recurrent atom because our original process is exponentially ergodic to $\Psi_{\infty,\lambda}$ (see Appx.~\ref{AppendErgod}), and,  as a consequence, the split chain is exponentially ergodic with respect to the invariant state $\tilde{\Psi}_{\infty,\lambda}$ (see Part (2) of Prop.~\ref{BasicsOfNum}) which has $\tilde{\Psi}_{\infty,\lambda}(\Sigma\times 1)=\Psi_{\infty,\lambda}(h)>0$.  Notice that the conditional probability that $z_2=1$ given $s_1, z_1, s_2$  is determined by a coin with heads-probability $h(s_2)$.

  Using the law for the split chain $(\tilde{\sigma}_{n})$, we may construct a split process $(\tilde{S}_{t})\in \tilde{\Sigma}$ and a sequence of times $\tilde{\tau}_{n}$ with the recipe below.  We refer to pages 1302 and 1306 of~\cite{Hopfner} for more discussion on the construction.  The $\tilde{\tau}_{n}$ should be thought of as the partition times $\tau_{n}$ embedded in the split statistics, although we temporarily denote them differently to emphasize their axiomatic role in the construction of the split process.  Let $\tilde{\tau}_{n}$ and $\tilde{S}_{t}=(S_{t},Z_{t})$ be such that 
\begin{enumerate}
\item  $0=\tilde{\tau}_{0}$, $\tilde{\tau}_{n}\leq \tilde{\tau}_{n+1}$, and $\tilde{\tau}_{n}\rightarrow \infty$ almost surely.  

\item The chain $(\tilde{S}_{\tilde{\tau}_{n}})$ has the same law as $(\tilde{\sigma}_{n})$.  

\item  For $t\in [\tilde{\tau}_{n},\tilde{\tau}_{n+1})$, then $Z_{t}=Z_{\tilde{\tau}_{n}}$.

\item  Conditioned on the information known up to time $\tilde{\tau}_{n}$ for $\tilde{S}_{t}$, $t\in [0,\tilde{\tau}_{n}]$ and $\tilde{\tau}_{m}$, $m\leq n$, and also the value $\tilde{S}_{\tilde{\tau}_{n+1} }$, the law for the trajectories $S_{t}$, $t\in[\tilde{\tau}_{n},\tilde{\tau}_{n+1}]$ (which refers also to the length $\tilde{\tau}_{n+1}-\tilde{\tau}_{n}$) agrees with the law for the original process conditioned on knowing the values $S_{\tilde{\tau}_{n}}$ and $S_{\tilde{\tau}_{n+1}}$.   
\end{enumerate}
The marginal distribution for the first component $S_{t}$ agrees with the original process and the times $\tilde{\tau}_{n}$ are independent mean one exponential random variables that are independent of $S_{t}$. Of course, the times $\tilde{\tau}_{n}$ are not independent of the process $\tilde{S}_{t}$, and we note that the increment $ \tilde{\tau}_{n+1}-\tilde{\tau}_{n}$ is not necessarily exponential when conditioned on the state $\tilde{S}_{\tilde{\tau}_{n} }$.  The process $\tilde{S}_{t}$ is not Markovian although, as emphasized in~\cite{Loch}, the process $(S_{t},Z_{t},S_{\tau(t)})\in \Sigma\times \{0,1\}\times \Sigma$ is Markovian, where $\tau(t)$ is the first partition time $\tilde{\tau}_{n}$ following time $t$.  Importantly, the strong Markov property for  $\tilde{S}_{t}$ does hold for the times   $\tilde{\tau}_{n}$; see~\cite[Remark 2.5]{Hopfner}.   We now drop the tilde from $\tilde{\tau}_{n}$, and use $(\tilde{\sigma}_{n})$ to denote the sequence $(\tilde{S}_{\tau_{n}})$.  We refer to the statistics of the split process by $\tilde{\mathbb{E}}^{(\lambda)}$ and $\tilde{\mathbb{P}}^{(\lambda)}$ for expectations and probabilities, respectively.      

Now that we have defined the split process $\tilde{S}_{t}$, we can proceed to define the ``life cycles".  Let $R_{m}'$ be the value $\tau_{\tilde{n}_{m}}$ for $\tilde{n}_{m}=\textup{min} \{ n\in \mathbb{N}\,\big|\,\sum_{k=0}^{n}\chi(Z_{\tau_k}= 1)  =m  \}   $.  In other words, $R_{m}'$ is the $m$th partition time to visit the atom set $\Sigma\times 1$, and we use the convention that $R_{0}'=0$.  Define $R_{m}$, $m\geq 1$ to be the partition time following $R_{m}'$.  The $m$th life cycle is the time interval $[R_{m},R_{m+1})$.  Intuitively, it may at first seem more natural to define $S_{R_{m}'}$ as the beginning of the life cycle.  However, the distribution for $R_{1}'$ will depend on the initial distribution $\tilde{S}_{0}$.   It is better to consider the beginning of the life cycle to be the partition time $R_{m}$ following $R_{m}'$, which has distribution  $\tilde{\nu}$ with respect to information known up to time  $R_{m}'$.   Although the conditional distribution for $\tilde{S}_{R_{m}}$ is independent  of the value $\tilde{S}_{R_{m}'}\in \Sigma\times 1$, successive live cycles $[R_{n-1},R_{n})$, $[R_{n},R_{n+1})$ are obviously not independent since, for instance, there is almost sure convergence $\lim_{t\nearrow R_{n} } S_{t }=S_{R_{n} }$.  Let $d\mathbf{N}_{t}$ be the counting measure on $\R^{+}$ such that  $\int_{(t_{1},t_{2}]} d\mathbf{N}_{r}=\mathbf{N}_{t_{2}}-\mathbf{N}_{t_{1}}$ for $0\leq t_{1}<t_{2}$, i.e., the number of partition times over the interval $(t_{1},t_{2}]$.  The following proposition lists some independence properties that follow closely from the construction of the split process.   The measure $\nu$  in the statement of  Prop.~\ref{IndependenceProp} can be regarded as a generic normalized measure satisfying~(\ref{NummelinCrit}) for  some $h:\Sigma\rightarrow [0,1]$  although we will choose it to be of the specific form in Conv.~\ref{NumConv} later in the text.

\begin{proposition}\label{IndependenceProp}\text{  }

\begin{enumerate}

\item The distribution for $\tilde{S}_{R_{n}}$ is $\tilde{\nu}$ when conditioned on all information known up to time $R_{n}'$: $\tilde{\mathcal{F}}_{R_{n}'}$.

\item  The sequence of trajectories $\big( S_{t},\, d\mathbf{N}_{t} : \, t\in [R_{n},R_{n+1}'] \big) $ are i.i.d. for $n\geq 1$, and $\big( S_{t},\, d\mathbf{N}_{t} : \, t\in [R_{n},R_{n+1}'] \big) $ is independent of $\big( \tilde{S}_{t},\,d\mathbf{N}_{t}: \, t\notin (R_{n}',R_{n+1}) \big) $.

\item The trajectory $\big( \tilde{S}_{t},\,d\mathbf{N}_{t}: \, t\in [R_{n},R_{n+1}] \big) $ is independent of  $\big( \tilde{S}_{t},\, d\mathbf{N}_{t}: \, t\notin (R_{n}',R_{n+2}) \big) $. In particular, $\big( \tilde{S}_{t},\, d\mathbf{N}_{t}: \, t\in [R_{n},R_{n+1}] \big) $ is independent of $\big( \tilde{S}_{t},\, d\mathbf{N}_{t}: \, t\in [R_{m},R_{m+1}] \big) $ for $|n-m|\geq 2$.

\end{enumerate}

\end{proposition}

\begin{proof}
Statement (1),  which is given in~\cite[Prop. 2.13]{Loch}, follows immediately from the construction.  Statements (2) and (3) follow from Part (1), the strong Markov property at the times $R_{n}$, and the independence of the partition times from the past~\cite[Prop. 2.6]{Loch}.  For instance, $\tilde{S}_{R_{n}}$ has distribution $\tilde{\nu}$ independently of $\big( S_{t},\, d\mathbf{N}_{t} : \, t\in [0,R_{n}'] \big) $   by Part (1).   By the strong Markov property for $\tilde{S}_{t}$  at the time $R_{n}$, the trajectory  $\big( S_{t} : \, t\in [R_{n},R_{n+1}'] \big) $ is independent of $\big( \tilde{S}_{t},\, d\mathbf{N}_{t} : \, t\in [0,R_{n}'] \big) $ when given the state $\tilde{S}_{R_{n}}$ and has the  same law as  $\big( S_{t}: \, t\in [0 ,R_{1}'] \big) $ when $\tilde{S}_{0}$ has distribution $\tilde{\nu}$.   The partition times $\tau_{m}$ over the interval $[R_{n},R_{n}']$, encoded by $\int_{R_{n}}^{t} d\mathbf{N}_{r}$ for $t\in  [R_{n},R_{n+1}'] $, are independent of $\big( S_{t},\, d\mathbf{N}_{t} : \, t\in [0,R_{n}'] \big) $ by~\cite[Prop. 2.6]{Loch} (and also independent of the process $S_{t}$ for all $t\in \R^{+}$).

\end{proof}

Unfortunately, notation multiplies when the splitting structure is invoked. 
 For easier reference, we list the following frequently used symbols: 
\begin{eqnarray*}
&\tilde{S}_{t}=(S_{t},Z_{t})       & \text{State of the split process at time $t$      }\\
&  \tau_{m}\in \R^{+}     &   \text{$m$th partition time} \\
&\tilde{\sigma}_{m}= \tilde{S}_{\tau_{m}}     &    \text{$m$th state of the split chain  }    \\
&(\sigma_{m} , \zeta_{m}) =\tilde{\sigma}_{m}   &    \text{$\sigma_{m}$ and $\zeta_{m}$ are the state and binary components, respectively, of $\tilde{\sigma}_{m}$  }\\
 &\mathbf{N}_{t}\in \mathbb{N} \text{\,}    & \text{Number of partition times $\tau_{m}$, $m\geq 1$ to occur up to time $t$ }\\
& R_{m}' \in \R^{+}         &  \text{$m$th partition time visiting the set $\Sigma\times 1$  }\\
& R_{m} \in \R^{+}         &  \text{Partition time succeeding $R_{m}'$ and the beginning of the $m$th life cycle}\\
&  \tilde{N}_{t} \in \mathbb{N}         &    \text{Number of returns to the atom up to time $t$} \\
  & \tilde{n}_{m} \in \mathbb{N}          &  \text{Number of partition times in the interval $(0,R_{m }]$  }\\
    & \mu\rightarrow \tilde{\mu}         &  \text{The splitting of a measure $\mu$ on $\Sigma$ as defined in~(\ref{SplitMeasure})}   
  \\
&\mathcal{F}_{t}    &    \text{Information up to time $t$ for the original process $S_{r}$ and the $\tau_{m}$ } \\ 
&\tilde{\mathcal{F}}_{t}    &    \text{Information up to time $t$ for the split process  $\tilde{S}_{r}$ and the $\tau_{m}$ }\\
&\tilde{\mathcal{F}}_{t}'    &    \text{Information for  $\tilde{S}_{t}$ and the $\tau_{m}$  before time $R_{n+1}$, where $R_{n}'\leq t<R_{n+1}'$}, \\ &   &   \text{plus  knowledge of the time  $ R_{n+1}$ itself}
\end{eqnarray*}
 If $\mathbf{t}$ is a partition time, e.g.,  $\mathbf{t}=\tau_{n}$ or  $\mathbf{t}=R_{n}$,  the $\sigma$-algebra  $\tilde{\mathcal{F}}_{\mathbf{t}^{-}}$ will refer to all information before time $\mathbf{t}$ plus the information that  $\mathbf{t}$ is a partition time.  

We will henceforth attach the subscript $\lambda$ to the transition map $ \mathcal{T}$ to emphasize the dependence of the dynamics on this parameter.   There is some flexibility in the choice of $\nu$ and $h$ in the criterion~(\ref{NummelinCrit}), although choosing them to be independent of $\lambda>0$ adds a little extra constraint.    By Part (1) of Prop.~\ref{BasicsOfNum}, we can select a pair $\nu$, $h$  that is independent of $\lambda$, and where both are functions of the energy.  We will use the symbol $\nu$ for both the measure and the corresponding density.  
\begin{convention}\label{NumConv}
We take $\nu$ and $h$ of the form
$$ h(s)= \mathbf{u}  \frac{ \chi\big(H(s)\leq l   \big)   }{ U } \hspace{1cm} \text{and} \hspace{1cm}  \nu(ds)=  ds \frac{ \chi\big( H(s)\leq l   \big)   }{U},   $$
where $l:=1+2\sup_{x}V(x)$, $U>0$ is the normalization constant of $\nu$, and $\mathbf{u}\in(0,U)$ is from Part (1) of Prop.~\ref{BasicsOfNum}.  
\end{convention}
The compact support of $h:\Sigma\rightarrow [0,1]$ implies that the  extended state space for the split dynamics is effectively $ \Sigma \times 0 \cup \supp(h) \times 1 \subset \tilde{\Sigma} $ since other states in $ \tilde{\Sigma}=\Sigma\times \{0,1\} $ will not be visited.  Any supremum, minimum, etc. over $\tilde{\Sigma}$ refers to this contracted set.  Parts (2) and (3) of the proposition below are elementary consequences of the splitting structure defined above and the proof is contained in Sect.~\ref{SecNumProofs}.

\begin{proposition}\label{BasicsOfNum} \text{  }

\begin{enumerate}
\item  There is a constant $\mathbf{u} >0$ such that the $h$ and $\nu$ in Conv.~\ref{NumConv} satisfy $\mathcal{T}_{\lambda}(s,ds'  )\geq h(s)\nu(ds')\,$ for all $s,s'\in \Sigma$ and    $\lambda<1$.  Also, the transition measures $\mathcal{T}_{\lambda}(s,ds')$ have densities over the domains $\{s'\in \Sigma\,\big|\, H(s')\neq H(s)\}$, which have the following bound
$$ 
\sup_{\lambda\leq 1 }\esssup_{ \substack{ H(s)>l\\ H(s)\neq H(s')  }     }\frac{\mathcal{T}_{\lambda}(s,ds^{\prime})}{ds^{\prime} }<\infty. $$

\item The invariant state of both the split chain $(\tilde{\sigma}_{n})$  and the split process $(\tilde{S}_{t})$ is the splitting of the invariant state of the original process, i.e.,
$$   \tilde{\Psi}_{\infty,\lambda}(s,0)= \big(1-h(s)\big)\Psi_{\infty,\lambda}(s)\quad \text{and}\quad \tilde{\Psi}_{\infty,\lambda}(s,1)= h(s)\Psi_{\infty,\lambda}(s).  $$
Thus, the ``atom" has measure $  \int_{\Sigma}ds h(s)\Psi_{\infty,\lambda}(s) >0 $.

\item  If $\mathbf{t}$ is a partition time, the distribution for  $\tilde{S}_{\mathbf{t}}$ conditioned on  $\tilde{\mathcal{F}}_{\mathbf{t}^{-}}$ is the splitting of the $\delta$-distribution at $S_{\mathbf{t}}$:  
$$\tilde{\delta}_{ S_{\mathbf{t}}  }(s,z)=\delta(s-S_{\mathbf{t}})\big(  \chi(z=0)\big(1-h(S_{\mathbf{t}})\big)+\chi(z=1)h(S_{\mathbf{t}})\big).   $$
In particular,  $\tilde{\mathbb{P}}^{(\lambda)}\big[ Z_{\mathbf{t} }= 1 \,\big| \,\tilde{\mathcal{F}}_{\mathbf{t}^{-}}\big]=h(S_{\mathbf{t}})$.   The strong Markov property at the time $\mathbf{t}$ and stationarity give us that 
$$  \mathcal{L}\big( (\tilde{S}_{\mathbf{t}+r  })\,\big|\,  \tilde{\mathcal{F}}_{\mathbf{t}^{-}}\big)= \mathcal{L}_{\tilde{\delta}_{ S_{\mathbf{t}}  }}\big(  (\tilde{S}_{r  })  \big), \hspace{1cm} r\in \R^{+},  $$
where $\mathcal{L}_{\mu}$ refers to the law starting from the distribution $\mu$.

\end{enumerate}

\end{proposition}

Besides the nearly independent behavior of the process $\tilde{S}_{t}$ over the intervals $[R_{m},R_{m+1})$, the payoff for introducing the splitting structure includes the closed formulas in Prop.~\ref{BasicsOfNumII}.   Part (2) of Prop.~\ref{BasicsOfNumII} is a special case of ~\cite[Prop. 2.20]{Loch}, which applies also to null-recurrent processes.   For Part (3) and (4) of the proposition below, $\frak{R}^{(\lambda)}$ is the reduced resolvent of the backward generator $\mathcal{L}^{*}_{\lambda}$
$$\frak{R}^{(\lambda)}g=\int_{0}^{\infty}dr e^{r\mathcal{L}^{*}_{\lambda}}(g),$$ 
which operates on $g\in L^{\infty}(\Sigma)$ with $\Psi_{\infty,\lambda}(g)=0 $.  
The reduced resolvent is well-defined since the process $S_{t}$ is exponentially ergodic for any fixed $\lambda>0$.  As $\lambda\searrow 0$ the expression in Part (4) is related to the diffusion constant $\kappa$ appearing in \cite[Thm. 1.1]{Further}.

\begin{proposition}\label{BasicsOfNumII} \text{ }

\begin{enumerate}
\item For $g\in L^{\infty}(\tilde{\Sigma})$, 
$$\tilde{\mathbb{E}}_{\tilde{\nu}}^{ (\lambda)}   \Big[ \sum_{m=0}^{\tilde{n}_{1}}  g(\tilde{\sigma}_{m}) \Big]=\tilde{\mathbb{E}}_{\tilde{\nu}}^{ (\lambda)}   \Big[ \sum_{m=1}^{\tilde{n}_{1}+1}  g(\tilde{\sigma}_{m}) \Big]=\frac{ \int_{\tilde{\Sigma}}d\tilde{s}\tilde{\Psi}^{(\lambda)}_{\infty}(\tilde{s}) g(\tilde{s})   }{\int_{\Sigma}ds\Psi^{(\lambda)}_{\infty}(s)   h(s) }.  $$
In particular, if $g\in L^{\infty}(\Sigma)$ does not depend on the binary variable, then the numerator on the right side  above is equal to $\int_{\tilde{\Sigma}}d\tilde{s}\tilde{\Psi}^{(\lambda)}_{\infty}(\tilde{s}) g(\tilde{s})=\int_{\Sigma}ds \Psi_{\infty,\lambda}(s) g(s)$.

\item For $g\in L^{\infty}(\Sigma)$,
$$\tilde{\mathbb{E}}_{\tilde{\nu}}^{ (\lambda)}   \Big[ \int_{0}^{R_{1}}dr g(S_{r}) \Big]=\frac{ \int_{\Sigma}ds\Psi_{\infty,\lambda}(s) g(s)   }{\int_{\Sigma}ds\Psi_{\infty,\lambda}(s)     h(s) } . $$

\item  For $g\in L^{\infty}(\Sigma)$ with $\Psi_{\infty,\lambda}(g)=0$ and $s_{1},s_{2}\in \Sigma$, 
$$  \tilde{\mathbb{E}}_{\tilde{\delta}_{s_{1}}}^{ (\lambda)}   \Big[ \int_{0}^{R_{1}}dr g(S_{r}) \Big]-  \tilde{\mathbb{E}}_{\tilde{\delta}_{s_{2}}}^{ (\lambda)}   \Big[ \int_{0}^{R_{1}}dr g(S_{r}) \Big]  =\big(\frak{R}^{(\lambda)}g\big)(s_{1})-\big(\frak{R}^{(\lambda)}g\big)(s_{2}),     $$
where $\tilde{\delta}_{s}$ is the splitting of the $\delta$-measure at $s\in \Sigma$.

\item For $g\in L^{\infty}(\Sigma)$ with $\Psi_{\infty,\lambda}(g)=0$, 
 $$\tilde{\mathbb{E}}_{\tilde{\nu}}^{ (\lambda)}   \Big[  \int_{0}^{R_{1}}dr g(S_{r})  \int_{r}^{R_{2}}dr'g(S_{r'})
    \Big]= \frac{  \int_{\Sigma}ds\Psi_{\infty,\lambda}(s) g(s)\big(\frak{R}^{(\lambda)}g\big)(s)  }{ \int_{\Sigma}ds\Psi_{\infty,\lambda}(s)h(s) }.   $$

\end{enumerate}

\end{proposition}

\begin{proof}\text{ }\\
\noindent Part (1):\hspace{.1cm}  This follows as a general fact for split chains when the original chain $\sigma_{n}$ is positive recurrent with normalizable invariant measure $\Psi_{\infty,\lambda}$.   As mentioned in the proof of~\cite[Thm. 3]{Nummelin},  the measure $\beta$ on $\tilde{\Sigma}$ given by 
\begin{align}\label{User}
\beta(g)=   \tilde{\mathbb{E}}_{\tilde{\nu}}^{ (\lambda)}   \Big[ \sum_{m=0}^{\tilde{n}_{1}}  g(\tilde{\sigma}_{m}) \Big], \hspace{1cm}
g\in L^{\infty}(\tilde{\Sigma}),  \end{align}
 satisfies $ \beta (h)=1 $ and is invariant for the split chain dynamics, i.e.,   $\beta \tilde{\mathcal{T}}_{\lambda}=\beta$.  Since the dynamics is positive recurrent with invariant state $\tilde{\Psi}^{(\lambda)}_{\infty}$, these features uniquely determine the above measure by the explicit form
$$\beta(g)=\frac{ \int_{\tilde{\Sigma}}d\tilde{s}\tilde{\Psi}^{(\lambda)}_{\infty}(\tilde{s}) g(\tilde{s})   }{\int_{\Sigma}ds\Psi^{(\lambda)}_{\infty}(s)   h(s) }.  $$

The distribution for $\tilde{\sigma}_{m}$ when $m=0$ and $m=\tilde{n}_{1}+1$ is $\tilde{\nu}$, so the summation of $ g(\tilde{\sigma}_{m})$ over  $[1,\tilde{n}_{1}+1]$ rather than $[0,\tilde{n}_{1}]$  in~(\ref{User}) yields the same result.    

\vspace{.4cm}

\noindent Part (2): \hspace{.1cm} Let $g_{n}^{(\lambda)}:\Sigma^{2}\rightarrow \R$ and $\mathbf{g}^{(\lambda)}:\Sigma\rightarrow \R$ be defined as
\begin{eqnarray*}
 g_{n}^{(\lambda)}(s,s')&:=&\mathbb{E}_{s}^{(\lambda)}\Big[  \Big(\int_{0}^{\tau_{1} }dr g(S_{r})\Big)^{n}\,\Big|\,s'=S_{\tau_{1} } \Big]\\ \mathbf{g}^{(\lambda)}(s)&:=& \mathbb{E}_{s}^{(\lambda)}\Big[\int_{0}^{\tau_{1}}dr g(S_{r})  \Big]  .
 \end{eqnarray*}  
 Also define $\mathbf{\tilde{g}}^{(\lambda)}:\tilde{\Sigma}\rightarrow \R$ analogously to $\mathbf{g}^{(\lambda)}(s)$ with $\mathbb{E}_{s}^{(\lambda)}$ replaced by $\tilde{\mathbb{E}}_{\tilde{s}}^{(\lambda)}$.

We have the following equalities:
\begin{align}
 \tilde{\mathbb{E}}_{\tilde{\nu}}^{ (\lambda)}   \Big[ \int_{0}^{R_{1}}dr g(S_{r}) \Big]&= \tilde{\mathbb{E}}_{\tilde{\nu}}^{ (\lambda)}   \Big[\sum_{n=0}^{\tilde{n}_{1}}\tilde{\mathbb{E}}^{(\lambda)}\Big[  \int_{\tau_{n}}^{\tau_{n+1} }dr g(S_{r})\,\Big|\tilde{\sigma}_{n},\, \tilde{\sigma}_{n+1} \Big] \Big] =\tilde{\mathbb{E}}_{\tilde{\nu}}^{ (\lambda)}   \Big[\sum_{n=0}^{\tilde{n}_{1}} g_{1}^{(\lambda)}(\sigma_{n},\sigma_{n+1})\Big]\nonumber\\ &= \tilde{\mathbb{E}}_{\tilde{\nu}}^{ (\lambda)}   \Big[\sum_{n=0}^{\tilde{n}_{1}}  \tilde{\mathbf{g}}^{(\lambda)}(\tilde{\sigma}_{n})\Big]=\frac{ \int_{\tilde{\Sigma}}d\tilde{s}\tilde{\Psi}_{\infty,\lambda}(\tilde{s}) \tilde{\mathbf{g}}^{(\lambda)}(\tilde{s})  }{ \int_{\Sigma}ds\Psi_{\infty,\lambda}(s) h(s)  },\label{GeeWhiz}
\end{align}
where the second equality holds because the statistics for $S_{r}$ over an interval $(\tau_{n},\tau_{n+1})$ given the values $\tilde{\sigma}_{n}=(\sigma_{n},\zeta_{n}),\,\tilde{\sigma}_{n+1}=(\sigma_{n+1},\zeta_{n+1})$  is independent of  $\zeta_{n}$, $\zeta_{n+1}$ and is the same for the split and the original dynamics.  The fourth  equality is from Part (1).

The numerator of the expression on the right side of~(\ref{GeeWhiz}) can be rewritten as follows:
\begin{align}  \int_{\tilde{\Sigma}}d\tilde{s}\tilde{\Psi}_{\infty,\lambda}(\tilde{s}) \tilde{\mathbf{g}}^{(\lambda)}(\tilde{s})&=  \int_{\Sigma}ds\Psi_{\infty,\lambda}(s) \mathbf{g}^{(\lambda)}(s) = \int_{\Sigma}ds\Psi_{\infty,\lambda}(s) \mathbb{E}_{s}^{(\lambda)}\Big[\int_{0}^{\tau_{1}}dr g(S_{r})  \Big]  \nonumber\\ &= \int_{\Sigma}ds\Psi_{\infty,\lambda}(s) \mathbb{E}_{s}^{(\lambda)}\Big[\int_{0}^{\infty}dr e^{-r}g(S_{r})  \Big]   \nonumber\\ &= \int_{0}^{\infty} dr e^{-r} \big(\int_{\Sigma}ds \Psi_{\infty,\lambda}(s) \mathbb{E}_{s}^{(\lambda)}\big[g(S_{r})    \big] \big)  \nonumber\\ &= \int_{0}^{\infty} dr e^{-r} \big(\int_{\Sigma}ds\Psi_{\infty,\lambda}(s) g(s)    \big)  =\int_{\Sigma}ds\Psi_{\infty,\lambda}(s) g(s)   .     \label{CompGee}
\end{align}   
 The first equality uses that $\tilde{\Psi}_{\infty,\lambda}$ has the split form in Part (2) of Prop.~\ref{BasicsOfNum}, and the third equality holds since $\tau_{1}$ is a mean one exponential independent of $S_{t}$ in the original statistics.  The fourth equality is Fubini, and the fifth is due to the stationarity of $\Psi_{\infty,\lambda}$.  

\vspace{.4cm}

\noindent Part (3):\hspace{.1cm}  The reduced resolvent $\frak{R}^{(\lambda)}$ is the pointwise limit given by
\begin{align*}
\big(\frak{R}^{(\lambda)}g\big)(s)& = \lim_{\gamma\searrow 0    }\mathbb{E}^{ (\lambda)}_{s }\Big[ \int_{0}^{\infty}dr e^{-r\gamma} g(S_{r})  \Big]
=\lim_{\gamma\searrow 0    }\tilde{\mathbb{E}}^{ (\lambda)}_{\tilde{\delta}_{s} }\Big[ \int_{0}^{\infty}dre^{-r\gamma} g(S_{r})  \Big],
\end{align*}
where the second equality embeds the expectation in the split statistics.  
However, for $s_{1},s_{2}\in \Sigma$, 
$$\Big(\tilde{\mathbb{E}}^{ (\lambda)}_{\tilde{\delta}_{s_{1}} }-\tilde{\mathbb{E}}^{ (\lambda)}_{\tilde{\delta}_{s_{2}} }\Big)\Big[ \int_{0}^{\infty}dr e^{-r\gamma} g(S_{r})  \Big]=\Big(\tilde{\mathbb{E}}^{(\lambda)}_{\tilde{\delta}_{s_{1}} }-\tilde{\mathbb{E}}^{ (\lambda)}_{\tilde{\delta}_{s_{2}} }\Big)\Big[ \int_{0}^{R_{1}}dr e^{-r\gamma} g(S_{r})  \Big]$$
since the distribution for the state $\tilde{S}_{R_{1}} $ is $\tilde{\nu}$ regardless of the initial measure.  Using the above equalities, we have that 
\begin{align*}
\big(\frak{R}^{(\lambda)}g\big)(s_{1})-\big(\frak{R}^{(\lambda)}g\big)(s_{2})   &= \lim_{\gamma\searrow 0    }\Big(\tilde{\mathbb{E}}^{(\lambda)}_{\tilde{\delta}_{s_{1}} }-\tilde{\mathbb{E}}^{ (\lambda)}_{\tilde{\delta}_{s_{2}} }\Big)\Big[ \int_{0}^{R_{1}}dr e^{-r\gamma} g(S_{r})  \Big]\\ & =\tilde{\mathbb{E}}^{(\lambda)}_{\tilde{\delta}_{s_{1}}}\Big[ \int_{0}^{R_{1}}dr g(S_{r})  \Big]-\tilde{\mathbb{E}}^{(\lambda)}_{\tilde{\delta}_{s_{1}}}\Big[ \int_{0}^{R_{1}}dr  g(S_{r})  \Big],
\end{align*}
where the limits are well-defined since the process $\tilde{S}_{t}$ is positive-recurrent and hence $\tilde{\mathbb{E}}^{(\lambda)}_{\tilde{s}}[R_{1}]$ is finite for all $\tilde{s}\in \tilde{\Sigma}$.

\vspace{.4cm}

\noindent Part (4): \hspace{.1cm} Notice that  
\begin{align}\label{Torque}
 \tilde{\mathbb{E}}_{\tilde{\nu}}^{ (\lambda)}   \Big[  \int_{0}^{R_{1}}dr g(S_{r})  \int_{r}^{R_{2}}dr'g(S_{r'})
    \Big] &=  \tilde{\mathbb{E}}_{\tilde{\nu}}^{ (\lambda)}   \Big[  \sum_{n=0}^{\tilde{n}_{1}}\tilde{\mathbb{E}}^{(\lambda)}\Big[\int_{\tau_{n}}^{\tau_{n+1} }dr g(S_{r})  \int_{r}^{R_{2}}dr'g(S_{r'})\,\Big|\,\tilde{\mathcal{F}}_{\tau_{n}^{-} }\Big]
    \Big]  \nonumber \\ &=
 \tilde{\mathbb{E}}_{\tilde{\nu}}^{ (\lambda)}   \Big[\sum_{n=0}^{\tilde{n}_{1}}\mathbf{f}^{(\lambda)}(\sigma_{n})   \Big],
\end{align} 
 where $\mathbf{f}^{(\lambda)}:\Sigma\rightarrow \R$ is defined as
$$
\mathbf{f}^{(\lambda)}(s):= \tilde{\mathbb{E}}_{\tilde{\delta}_{s}}^{ (\lambda)}   \Big[  \int_{0}^{\tau_{1}}dr g(S_{r})  \int_{r}^{R_{2}}dr'g(S_{r'})
    \Big].   $$
 The equality~(\ref{Torque})  is a consequence of Part (3) of Prop.~\ref{BasicsOfNum} and uses the strong Markov property at the times $\tau_{n}$ for $n\in[0,\tilde{n}_{1}]$.  The function  $\mathbf{f}^{(\lambda)}(s) $ can be rewritten as 
\begin{align}\label{Haughty}
 \mathbf{f}^{(\lambda)}(s)  & = 
 \tilde{\mathbb{E}}_{\tilde{\delta}_{s}}^{(\lambda)}\Big[\int_{0}^{\tau_{1}}drg(S_{r})\int_{r}^{\tau_{1}}dv g(S_{v})+
\Big(\int_{0}^{\tau_{1}}dr g(S_{r})\Big)\tilde{\mathbb{E}}^{ (\lambda)}\Big[ \int_{\tau_{1} }^{R_{2}}dr g(S_{r})\,\Big|\,\tilde{\mathcal{F}}_{\tau_{1}^{-}} \Big]  \Big] \nonumber \\
&=\tilde{\mathbb{E}}_{\tilde{\delta}_{s}}^{(\lambda)}\Big[\int_{0}^{\tau_{1}}dr g(S_{r})\int_{r}^{\tau_{1}}dv g(S_{v})+
\Big(\int_{0}^{\tau_{1}}dr g(S_{r})\Big)\big(\frak{R}^{(\lambda)}g\big)\big(S_{\tau_{1}}\big)+c \int_{0}^{\tau_{1}}dr g(S_{r}) \Big] \nonumber  \\
&=\mathbb{E}_{s}^{(\lambda)}\Big[\int_{0}^{\tau_{1}}dr g(S_{r})
\big(\frak{R}^{(\lambda)}g\big)\big(S_{r}\big) +c\int_{0}^{\tau_{1}}dr g(S_{r}) \Big],
\end{align}
where $c\in \R$ is the constant such that $ (\frak{R}^{(\lambda)}g)(s)+c=\tilde{\mathbb{E}}^{ (\lambda)}_{\tilde{\delta}_{s}}\big[ \int_{0 }^{R_{1}}dr g(S_{r})\big]$ for all $s\in \Sigma$, which exists by Part (3).  The value for $c$ depends on $g$ and the choice of $\nu$, $h$ defining the Nummelin splitting.  For the second equality, we have used that
\begin{align*}
\tilde{\mathbb{E}}^{ (\lambda)}\Big[ \int_{\tau_{1}}^{R_{2}}dr g(S_{r})\,\Big|\,\tilde{\mathcal{F}}_{\tau_{1}^{-}}  \Big]
=\tilde{\mathbb{E}}^{ (\lambda)}_{\tilde{\delta}_{S_{\tau_{1}}} }\Big[ \int_{0}^{R_{1}}dr g(S_{r})  \Big]
 =  \big(\frak{R}^{(\lambda)}g\big)\big(S_{\tau_{1}}\big)+c,    
\end{align*}
where the first equality is by Part (3) of Prop.~\ref{BasicsOfNum}. The third equality of~(\ref{Haughty}) follows by replacing  $\tilde{\mathbb{E}}_{\tilde{\delta}_{s}}^{(\lambda)}$ with $ \mathbb{E}_{s}^{(\lambda)}$ and  using a nested conditional expectation with respect to $\mathcal{F}_{r}$ for $r\leq \tau_{1}$:
$$  \mathbb{E}^{(\lambda)}\Big[\int_{r}^{\tau_{1}}dv g(S_{v})
+\big(\frak{R}^{(\lambda)}g\big)\big(S_{\tau_{1}}\big)\,\Big|\,\mathcal{F}_{r}  \Big] =\mathbb{E}_{S_{r}}^{(\lambda)}\Big[\int_{0}^{\tau_{1}}dv g(S_{v})
+\big(\frak{R}^{(\lambda)}g\big)\big(S_{\tau_{1}}\big)  \Big]= \big(\frak{R}^{(\lambda)}g\big)\big(S_{r}\big) .   $$

We can then plug our expression~(\ref{Haughty}) for $ \mathbf{f}^{(\lambda)}(s)$ into~(\ref{Torque}) and invert the first two steps of the proof to obtain   the first equality below:
\begin{align}\label{Shock}
 \tilde{\mathbb{E}}_{\tilde{\nu}}^{ (\lambda)}   \Big[  \int_{0}^{R_{1}}dr g(S_{r})  \int_{r}^{R_{2}}dr'g(S_{r'})
    \Big]= &\tilde{\mathbb{E}}_{\tilde{\nu}}^{ (\lambda)}   \Big[  \int_{0}^{R_{1}}dr g(S_{r})  \big(\frak{R}^{(\lambda)}g\big)\big(S_{r}\big)
+c  \int_{0}^{R_{1}}dr g(S_{r})   \Big]\nonumber \\ =& \frac{ \int_{\Sigma}ds\Psi_{\infty,\lambda}(s) \big[\big(\frak{R}^{(\lambda)}g\big)(s)\mathbf{g}^{(\lambda)}(s) +cg(s) \big]  }{ \int_{\Sigma}ds\Psi_{\infty,\lambda}(s) h(s)  } 
  .    
\end{align}
The second equality follows by  Part (2), and  the constant $c$ disappears from the expression since $\Psi_{\infty,\lambda}(g)=0$.

\end{proof}

The following  proposition lists a few martingales related to the number $\tilde{N}_{t}$ of returns to the atom up to time $t\in \R^{+}$. 

\begin{proposition}\label{TrivialMart}
For the split statistics,  $ \tilde{N}_{t} - \sum_{n=1}^{\mathbf{N}_{t}}h(S_{\tau_{n}})$ is a martingale with respect to the filtration $\tilde{\mathcal{F}}_{t}$.   For the original statistics, 
$\sum_{n=1}^{\mathbf{N}_{t}}h(S_{\tau_{n}}) - \int_{0}^{t}dr h(S_{r})$ is a martingale with respect to  $\mathcal{F}_{t}$.  
In particular,
$$\tilde{\mathbb{E}}^{(\lambda)}\big[ \tilde{N}_{t}    \big]=\mathbb{E}^{(\lambda)}\Big[ \int_{0}^{t}dr h(S_{r})    \Big].  $$

\end{proposition}
\begin{proof}

The difference 
$$ \tilde{N}_{t} - \sum_{n=1}^{\mathbf{N}_{t}}h(S_{\tau_{n}})=  \sum_{n=1}^{\mathbf{N}_{t}}\big(\chi(Z_{\tau_{n}}=1)-h(S_{\tau_{n}})\big)$$ 
is a martingale since for $t< \tau_{n}$ the increments satisfy 
 \begin{align}\label{Not}
\mathbb{E}\big[\chi(Z_{\tau_{n}}=1)-h(S_{\tau_{n}})\,\big|\,\tilde{\mathcal{F}}_{t} \big]&=\mathbb{E}\big[  \mathbb{P}^{(\lambda)}[ Z_{\tau_{n}}=1\,|\, \tilde{\mathcal{F}}_{ \tau_{n}^{-}}   ]-h(S_{\tau_{n}}) \,\big|\,\tilde{\mathcal{F}}_{t}   \big]=0,
\end{align}
where the second equality holds because $\mathbb{P}^{(\lambda)}[ Z_{\tau_{n}}=1\,|\, \tilde{\mathcal{F}}_{ \tau_{n}^{-}}   ]=h(S_{\tau_{n}})$ by Part (3) of Prop.~\ref{BasicsOfNum}.
  The difference $\sum_{n=1}^{\mathbf{N}_{t}}h(S_{\tau_{n}}) - \int_{0}^{t}dr h(S_{r})$ is a martingale according to the original  law since the contributions $h(S_{\tau_{n}})$ occur with Poisson rate $1$.  

\end{proof}

\section{The frequency of returns to the atom}\label{SecReturns}
Sections~\ref{SecAtoms} and~\ref{SecFracMom} effectively bound the frequency of returns to the atom from above and below, respectively.  

\subsection{Bounding the number of returns to the atom}\label{SecAtoms}

Recall that $\tilde{N}_{t}$ is defined for the split process as the number of returns to the atom set up to time $t\in \R^{+}$.     We will now focus on bounding the expectation of $\tilde{N}_{t}$ for $t=\frac{T}{\lambda}$ in the limit of small $\lambda$.  By Prop.~\ref{TrivialMart} the expectation of $ \tilde{N}_{t}$ with respect to the split statistics is equal to the expectation of $\int_{0}^{t}dr h(S_{r})$ with respect to the original statistics.  The time integral of the process $ h(S_{t})$ keeps track of the amount of time that    $S_{t}$ loiters in the low momentum region where $h:\Sigma\rightarrow \R^+$ has support and the life cycles regenerate.  However, it is useful to work with a process that serves the same purpose as $\int_{0}^{t}dr h(S_{r})$ but that is easier to handle.  A convenient option is the increasing part of the drift $\mathbf{A}_{t}^{+}$  in the semi-martingale decomposition for $ \mathbf{Q}_{t}:=(2H_{t})^{\frac{1}{2}} $, which increases at a decaying rate away from the low momentum region; see the discussion below and Part (2) of Prop.~\ref{AMinus}. Functions of the energy $H(x,p)=\frac{1}{2}p^{2}+V(x)$ have the advantage of being invariant under the Hamiltonian evolution, which   makes energy related quantities a desirable starting point for gaining some control over the typical behavior of the dynamics.

Define the functions  $\mathcal{E}_{\lambda}:\R\rightarrow \R^{+}$ and  $\mathcal{A}_{\lambda},\mathcal{V}_{\lambda,n},\mathcal{V}_{\lambda,n}^{+},\mathcal{K}_{\lambda,n}:\Sigma \rightarrow \R$  as 
\begin{eqnarray*}
\mathcal{E}_{\lambda}(p)&  = &  \int_{\R}dp^{\prime}\mathcal{J}_{\lambda}(p,p^{\prime}),\\
 \mathcal{A}_{\lambda}(x,p)&=&\int_{\R}dp^{\prime} \Big( 2^{\frac{1}{2}}H^{\frac{1}{2}}(x,p^{\prime})- 2^{\frac{1}{2}}H^{\frac{1}{2}}(x,p)     \Big)\mathcal{J}_{\lambda}(p,p^{\prime}), \\  
\mathcal{V}_{\lambda,n}(x,p)&=& \int_{\R}dp^{\prime} \Big( 2^{\frac{1}{2}}H^{\frac{1}{2}}(x,p^{\prime})- 2^{\frac{1}{2}}H^{\frac{1}{2}}(x,p)  \Big)^{2n}      \mathcal{J}_{\lambda}(p,p^{\prime}) ,   \\
 \mathcal{V}_{\lambda,n}^{+}(x,p)&=&
\int_{\R}dp^{\prime} \Big| H^{\frac{1}{2}}(x,p^{\prime})- H^{\frac{1}{2}}(x,p)\Big|^{n}\chi\big(|p'|>|p|\big)      \mathcal{J}_{\lambda}(p,p^{\prime}) ,  \\ \mathcal{K}_{\lambda,n}(x,p)&=&
\int_{\R}dp^{\prime} \Big| H^{\frac{1}{2}}(x,p^{\prime})- H^{\frac{1}{2}}(x,p) -\frac{\mathcal{A}_{\lambda}(x,p)}{ \mathcal{E}_\lambda(p)  }\Big|^{n}      \mathcal{J}_{\lambda}(p,p^{\prime}).
\end{eqnarray*}
Also define $\mathcal{A}_{\lambda}^{\pm}(s)= \max(\pm \mathcal{A}_{\lambda}(s),0)$ to be the positive and negative parts of  $\mathcal{A}_{\lambda}$.  We will often denote $\mathcal{V}_{\lambda,1} $ as $\mathcal{V}_{\lambda} $.  Let $\mathbf{M}_{t}$ and $\mathbf{A}_{t}$ be the martingale and predictable parts in the semi-martingale decomposition of $\mathbf{Q}_{t}$ in which both are initially zero:
$$ \mathbf{Q}_{t}= \big( 2 H_{t}     \big)^{\frac{1}{2}}=\mathbf{Q}_{0}+\mathbf{M}_{t}+\mathbf{A}_{t}.  $$
   The predictable component has the form  $\mathbf{A}_{t}=\int_{0}^{t}dr\mathcal{A}_{\lambda}(X_r,P_r)   $. By  defining 
     $\mathbf{A}_{t}^{\pm}:=\int_{0}^{t}dr\mathcal{A}_{\lambda}^{\pm}(X_r,P_r)  $ for $\mathcal{A}_{\lambda}^{\pm}(s):=\textup{max}(\pm \mathcal{A}_{\lambda}(s),0)$, the predictable component can be written as the difference $\mathbf{A}_{t}=\mathbf{A}_{t}^{+}-\mathbf{A}_{t}^{-}$.  The martingale $\mathbf{M}_{t}$ has predictable quadratic variation $\langle \mathbf{M}\rangle_{t}=\int_{0}^{t}dr\mathcal{V}_{\lambda}(X_{r},P_{r})  $.  

The following proposition states some basic facts for the functions  $\mathcal{A}_{\lambda}^{\pm} $, $ \mathcal{V}_{\lambda,n} $,  $\mathcal{V}_{\lambda,n}^{+}$, and $\mathcal{K}_{\lambda,n}$.  The proofs of Parts 1-4 of Prop.~\ref{AMinus} are placed in Sect.~\ref{SecEnergyLemProof}, and we do not include the proofs of Parts 5-7 which require similar calculus-based arguments.  The function $\mathcal{D}_\lambda :\R\rightarrow \R$ in Part (1) of Prop.~\ref{AMinus} is the drift rate in momentum due to collisions: $\mathcal{D}_{\lambda}(p)=\int_{\R}dp^{\prime}(p^{\prime}-p) \mcJ_{\lambda}(p,p^{\prime})$.

\begin{proposition}\label{AMinus}\text{There exist $c,C,C_{n}>0$ such that for $\lambda$ small enough the statements below hold.}

\begin{enumerate}

\item  For all $(x,p)\in \Sigma$,  $\mathcal{A}_{\lambda}^{-}(x,p)\leq |\mathcal{D}_\lambda(p)|$.  In particular,   $\mathcal{A}_{\lambda}^{-}(x,p)\leq C(\lambda|p|+\lambda^{2}p^{2})$.

 \item  For all $(x,p)\in \Sigma$, $\mathcal{A}_{\lambda}^{+}(x,p)\leq \frac{C}{1+p^{2}}$.

\item As $\lambda\to 0$, we have
$\int_{\Sigma}ds\mathcal{A}_{\lambda}^{+}(s)= 1+\mathit{O}(\lambda^{\frac{1}{2}})$.

\item For all $(x,p)\in \Sigma$, $\mathcal{K}_{\lambda,n}(x,p)\leq C_n(1+\lambda|p| )$.

\item For all $(x,p)\in \Sigma$, $\mathcal{V}_{\lambda,n}(x,p)\leq C(1+\lambda|p| )^{n+1}$.

\item For all $(x,p)\in \Sigma$,  $\mathcal{V}_{\lambda,n}^{+}(x,p)\leq C_n$.

\item For all $(x,p)\in \Sigma$, $\mathcal{V}_{\lambda}(x,p)\geq c$.

\end{enumerate}

\end{proposition}

Lemma~\ref{FirstEnergyLem} states that the energy process $H_{t}:=H(X_{t},P_{t})$ typically does not go above the scale $\lambda^{-1}$ over the time interval $[0,\frac{T}{\lambda}]$. The proof is based on martingale analysis and the bounds in Prop.~\ref{AMinus} and does not involve the Nummelin splitting structure.  

\begin{lemma}\label{FirstEnergyLem}  For any $n\in \mathbb{N}$, there exists a $C>0$ such that
$$  \mathbb{E}^{(\lambda)}\Big[\sup_{0\leq r\leq \frac{T}{\lambda} } (H_{r})^{\frac{n}{2}} \Big] \leq  C\Big(\frac{ T}{\lambda}\Big)^{\frac{n}{2}}    $$
for all $T>0$ and $\lambda<1$.  

\end{lemma}

\begin{proof}
 We will work with the process $\mathbf{Q}_{t}:=(2 H_{t})^{\frac{1}{2}}$. The reader should think of $\mathbf{Q}_{t}$ as being roughly the absolute value of the momentum $|P_{t}|$.  If $P_{t}$ were a symmetric random walk making steps every unit of time, then the result would follow  by Doob's maximal inequality with $\mathbf{Q}_{t}$ replaced by $|P_{t}|$ (supposing that the tail distribution of the jumps decays sufficiently fast).  The situation for our jump rates should, in principle, be even more accommodating since the jump rates~(\ref{JumpRates}) tend to drag a momentum with large absolute value down to a momentum with smaller absolute value.  However, for the purposes of this lemma, it is useful to discard the term associated with these large downward jumps in the decomposition~(\ref{Ach}) of $\mathbf{Q}_{t}$ because it is less analytically wieldy and it is not helpful on the time scales $\frac{T}{\lambda}$ for $T$ fixed and $\lambda\ll 1$.

  For technical reasons, we partition the time interval $[0,\frac{T}{\lambda}]$ through a sequence of incursion times $\varsigma_m'$ into a region of ``low"  energy.   Let $\varsigma_{0}=\varsigma^{\prime}_{1}=0$, and define the stopping times $\varsigma_{m},\varsigma_{m}^{\prime}$ such that 
\begin{align*}
\varsigma^{\prime}_{m}&= \min \{r\in (\varsigma_{m-1},\infty)\,\big| \,\mathbf{Q}_{r}\leq  \lambda^{-\frac{1}{2} }         \},&
\varsigma_{m}&= \min \{ r\in (\varsigma^{\prime}_{m},\infty)\,\big| \, \mathbf{Q}_{r}\geq 2  \lambda^{-\frac{1}{2} }            \}.
\end{align*}
The intervals $[\varsigma_m',\varsigma_m)$ and $[\varsigma_m,\varsigma_{m+1}')$ are incursions and excursions, respectively.  The above definitions assume that $\mathbf{Q}_{0}\leq \lambda^{-\frac{1}{2}}$, which is reasonable for $\lambda\ll 1$ by the locality assumption on the initial distribution (2) of List~\ref{Assumptions}, but we should take $\varsigma_{1}=\varsigma_{1}'=0$ when $\mathbf{Q}_{0}> \lambda^{-\frac{1}{2}}$.

Trivially, we have the inequality
\begin{align}\label{Incur}
\sup_{0\leq r\leq t}\mathbf{Q}_{r}\leq  2\lambda^{-\frac{1}{2}}+\sup_{0\leq r\leq t}\big(\mathbf{Q}_{r}-\mathbf{Q}_{r^{-}}\big)^+ +\sup_{ \varsigma_{m}\leq t}\sup_{r\in[\varsigma_{m},\,\varsigma_{m+1 }'\wedge t] }\big(\mathbf{Q}_{r}-\mathbf{Q}_{\varsigma_{m}}\big)^+,
\end{align}
where $(y)^{+}:=\max(y,0)$ for $y\in \R$.   The two rightmost terms in~(\ref{Incur}) bound the largest fluctuation of the process $\mathbf{Q}_{t}$ above the line $2\lambda^{-1}$.  In particular,  the middle term on the right side of~(\ref{Incur}) bounds the largest over-jump past the line $2\lambda^{-\frac{1}{2}}$ at the start of the excursions from low energy.

Let $t_{j}$, $j>0$ be the collision times with $t_{0}=0$ and recall that $\calN_{t}$ is the number of collisions up to time $t$.  We can write $\mathbf{Q}_{t}$ as
\begin{align}\label{Ach}
\mathbf{Q}_{t}=\mathbf{Q}_{0}+ \mathbf{m}_{t}+ \mathbf{m}_{t}'+\int_{0}^{t}dr\mathcal{A}^{+}(S_{r})-\sum_{j=1}^{\calN_{t}}\frac{\mathcal{A}_{\lambda}^{-}(S_{t_{j}^{-}})}{\mathcal{E}_{\lambda}(  P_{t_{j}^{-}}  )}, 
\end{align}
where the process $\mathbf{m}_{t}$ is defined as
\begin{align*}
\mathbf{m}_{t}:=\sum_{j=1}^{\calN_{t}}\Delta_{j} \quad \quad \text{for} \quad\quad   \Delta_{j}:= \mathbf{Q}_{t_{j}} -\mathbf{Q}_{t_{j}^{-}} - \frac{\mathcal{A}_{\lambda}(S_{t_{j}^{-}})}{\mathcal{E}_{\lambda}(  P_{t_{j}^{-}}  )},
\end{align*}
 and  $\mathbf{m}_{t}'$ is the difference
$$  \mathbf{m}_{t}':=\sum_{j=1}^{\calN_{t}}\frac{\mathcal{A}_{\lambda}^{+}(S_{t_{j}^{-}}) }{\mathcal{E} _{\lambda}(P_{t_{j}^{-}}) }-\int_{0}^{t}dr\mathcal{A}_{\lambda}^{+}(S_{r}).  $$
The processes  $\mathbf{m}_{t}$ and $\mathbf{m}_{t}'$ are martingales with respect to the filtration $\mathcal{F}_{t}$.   To see that $\mathbf{m}_{t}$ is a martingale, notice that the increments $\Delta_{j}$ have mean zero given the information $\mathcal{F}_{\tau_{j}^{-}}$ and  $\calN_{\tau_{j}}=\calN_{\tau_{j}^{-}}+1$.   The process $\mathbf{m}_{t}'$ is a martingale since  the terms $ \frac{\mathcal{A}_{\lambda}^{+}(S_{r}) }{\mathcal{E} _{\lambda}(P_{r}) }$ in the sum occur with Poisson rate $\mathcal{E} _{\lambda}(P_{r})$.  Moreover, the predictable quadratic variations corresponding to the martingales have the forms
$$ \langle \mathbf{m}\rangle_{t}= \int_{0}^{t}dr\mathcal{K}_{\lambda,2}(S_{r})\hspace{1cm}\text{and} \hspace{1cm}  \langle \mathbf{m}'\rangle_{t}= \int_{0}^{t}dr\frac{\big(\mathcal{A}_{\lambda}^{+}(S_{r})\big)^{2}}{\mathcal{E}_{\lambda}(P_{r})  }.  $$

When $t=\frac{T}{\lambda}$ the third term on the right side of~(\ref{Incur}) is smaller than
\begin{align}
\sup_{ \varsigma_{n}\leq \frac{T}{\lambda} }\sup_{t\in[\varsigma_{n},\,\varsigma_{n+1 }'\wedge \frac{T}{\lambda} ) }\big(\mathbf{Q}_{t}-\mathbf{Q}_{\varsigma_{n}}\big)^+ &\leq  \sup_{ \varsigma_{n}\leq \frac{T}{\lambda} }\sup_{t\in[\varsigma_{n},\,\varsigma_{n+1 }'\wedge \frac{T}{\lambda} ] }\left(\mathbf{m}_{t}+\mathbf{m}_{t}'-\mathbf{m}_{\varsigma_{n}}-\mathbf{m}_{\varsigma_{n}}'+\int_{\varsigma_{n}}^{t}dr\mathcal{A}^{+}(S_{r})\right)^+  \nonumber\\ &\leq 2\sup_{0\leq t\leq \frac{T}{\lambda} }\big|\mathbf{m}_{t}\big| +2\sup_{0\leq t\leq \frac{T}{\lambda} }\big|\mathbf{m}_{t}'\big|+ \int_{0}^{\frac{T}{\lambda}}dr\mathcal{A}^{+}(S_{r})\chi(\mathbf{Q}_{r}\geq \lambda^{-\frac{1}{2}}) \nonumber\\ &\leq 2\sup_{0\leq t\leq \frac{T}{\lambda} }\big|\mathbf{m}_{t}\big| +2\sup_{0\leq t\leq \frac{T}{\lambda} }\big|\mathbf{m}_{t}'\big|+ 2C \label{ReAch}
T.
\end{align}
For the first inequality, we have thrown away the term $-\sum_{j=1}^{\calN_{t}}\frac{\mathcal{A}_{\lambda}^{-}(S_{t_{j}^{-}})}{\mathcal{E}_{\lambda}(  P_{t_{j}^{-}}  )}$ since it is strictly negative.  The second inequality is  the triangle inequality with $\sup_{0\leq s, r\leq t}\big|f_{r}-f_{s} \big|\leq 2\sup_{0\leq r\leq t}\big|f_{r}\big|$ for $f=\mathbf{m},\mathbf{m}'$, and uses the fact that $\mathbf{Q}_{r}\geq \lambda^{-\frac{1}{2}}$ during the excursion intervals $ [\varsigma_{n},\,\varsigma_{n+1 }'\wedge \frac{T}{\lambda} )$.  The third inequality is a consequence of Part (2) of Prop.~\ref{AMinus}, which gives a $C>0$ such that   
$$\mathcal{A}^{+}(X_{r},P_{r})\chi\big(\mathbf{Q}_{r}\geq \lambda^{-\frac{1}{2}}\big)\leq \frac{C}{1+P_{r}^{2}} \chi\big(\mathbf{Q}_{r}\geq \lambda^{-\frac{1}{2}}\big)\leq \frac{C}{1+\frac{1}{2}\lambda^{-1} }\leq 2C\lambda ,  $$
where the second inequality is for $\lambda$ small enough so that $4\sup_{x}V(x)\leq \lambda^{-1}$.

Combining~(\ref{Ach}) for $t=\frac{T}{\lambda}$ with~(\ref{ReAch}) and using the triangle inequality, then 
\begin{align}\label{AchIII}
\mathbb{E}^{(\lambda)}\Big[\sup_{0\leq t\leq \frac{T}{\lambda} } \mathbf{Q}_{t}^{n} \Big]^{\frac{1}{n}} \leq &  2CT+  \mathbb{E}^{(\lambda)}\big[\mathbf{Q}_{0}^{n} \big]^{\frac{1}{n}}+\mathbb{E}^{(\lambda)}\Big[\sup_{0\leq t\leq \frac{T}{\lambda} }\big((\mathbf{Q}_{t}-\mathbf{Q}_{t^{-}})^+\big)^{n} \Big]^{\frac{1}{n}}\nonumber \\ &+2 \mathbb{E}^{(\lambda)}\Big[\sup_{0\leq t\leq \frac{T}{\lambda} }\big|\mathbf{m}_{t}\big|^{n}  \Big]^{\frac{1}{n}} + 2 \mathbb{E}^{(\lambda)}\Big[\sup_{0\leq t\leq \frac{T}{\lambda} }\big|\mathbf{m}_{t}'\big|^{n}  \Big]^{\frac{1}{n}}.
\end{align}
We will now give bounds for each of the terms on the right side above.  The goal is to show that
\begin{align}\label{TheGame}
\mathbb{E}^{(\lambda)}\Big[\sup_{0\leq t\leq \frac{T}{\lambda} } \mathbf{Q}_{t}^{n} \Big]^{\frac{1}{n}}=\mathit{O}\Big(\lambda^{-\frac{1}{2}} +\sum_{m}\Big(\mathbb{E}^{(\lambda)}\Big[\sup_{0\leq t\leq \frac{T}{\lambda} } \mathbf{Q}_{t}^{n} \Big]^{\frac{1}{n}}  \Big)^{\alpha_{m}}   \Big),
\end{align}
where  $\alpha_{m}\geq 2$ and the sum over $m$ includes a finite number of terms not depending on the parameter $\lambda>0$.  The above would imply that $\mathbb{E}^{(\lambda)}\big[\sup_{0\leq t\leq \frac{T}{\lambda} } \mathbf{Q}_{t}^{n} \big]$ is $\mathit{O}(\lambda^{-\frac{n}{2}})$, which is the statement we wish to prove.

 For the second term on the right side of~(\ref{AchIII}), 
$$\mathbb{E}^{(\lambda)}\big[\mathbf{Q}_{0}^{n} \big]=\int_{S}d\mu(x,p)\big(p^{2}+2V(x)\big)^{\frac{n}{2}} <\infty,$$
and the right side is finite by our assumption on the initial measure in List~\ref{Assumptions}.  For the third term on the right side of~(\ref{AchIII}).  Using that $(\sup_{m}a_{m})^{2}\leq \sum_{m} a_{m}^{2}$ and Jensen's inequality, we have the first inequality below
\begin{align*}
\mathbb{E}^{(\lambda)}\Big[\sup_{0\leq t\leq \frac{T}{\lambda} }\big((\mathbf{Q}_{t}-\mathbf{Q}_{t^{-}})^{+}\big)^{n} \Big]^{\frac{1}{n}} &\leq \mathbb{E}^{(\lambda)}\Big[\sum_{j=1}^{\calN_{\frac{T}{\lambda}}} \big( (\mathbf{Q}_{t_{j}}-\mathbf{Q}_{t_{j}^{-}})^+\big)^{2n} \Big]^{\frac{1}{2n}}=\mathbb{E}^{(\lambda)}\Big[\sum_{j=1}^{\calN_{\frac{T}{\lambda}}}\frac{\mathcal{V}_{\lambda,2n}^{+}(X_{t_{j}^{-}},P_{t_{j}^{-} })    }{ \mathcal{E}_{\lambda}(P_{t_{j}^{-} })  } \Big]^{\frac{1}{2n}} \\
&=
\mathbb{E}^{(\lambda)}\Big[\int_{0}^{\frac{T}{\lambda}}dt\mathcal{V}_{\lambda,2n}^{+}(X_{t},P_{t})   \Big]^{\frac{1}{2n}}\leq c^{\frac{1}{2n}}\mathbb{E}^{(\lambda)}\Big[\int_{0}^{\frac{T}{\lambda}}dt\big(1+\lambda \mathbf{Q}_{t} \big)  \Big]^{\frac{1}{2n}}
\\&\leq  c^{\frac{1}{2n}}T^{\frac{1}{2}}\lambda^{-\frac{1}{2n}}+c^{\frac{1}{2n}}T^{\frac{1}{2n}}\mathbb{E}^{(\lambda)}\Big[\sup_{0\leq t\leq \frac{T}{\lambda}} \mathbf{Q}_{t}^{n}  \Big]^{\frac{1}{2n^{2}}}.
\end{align*}
The first equality uses that 
$$\mathbb{E}\Big[\Big((\mathbf{Q}_{t}-\mathbf{Q}_{t^{-}})^+\Big)^{2n}\,\Big|\,\mathcal{F}_{t^{-}},\,\calN_{t}=\calN_{t^{-}}+1\Big]=\frac{\mathcal{V}_{\lambda,2n}^{+}(X_{t},P_{t^{-}})    }{ \mathcal{E}_{\lambda}(P_{t^{-}})  },   $$
and the second equality uses that the times $t_{n}$ occur with Poisson rate $\mathcal{E}_{\lambda}(P_{t})$.  The second inequality is for some $c>0$ by Part (6) of Prop.~\ref{AMinus}.  The right-most term above has the form of~(\ref{TheGame}).

For the fourth term on the right side of~(\ref{AchIII}), we can apply Doob's maximal inequality to get the first inequality below:
\begin{align*}
  \mathbb{E}^{(\lambda)}\Big[\sup_{0\leq t\leq \frac{T}{\lambda} }&\big|\mathbf{m}_{t}\big|^{n}  \Big]^{\frac{1}{n}} \leq \frac{n}{n-1}\mathbb{E}^{(\lambda)}\big[\big|\mathbf{m}_{\frac{T}{\lambda} }\big|^{n} \big]^{\frac{1}{n}}\leq  C' \mathbb{E}^{(\lambda)}\big[\big(\langle \mathbf{m}\rangle_{\frac{T}{\lambda} }\big)^{\frac{n}{2}} \big]^{\frac{1}{n}}+C'\mathbb{E}^{(\lambda)}\Big[ \sum_{j=1}^{\calN_{\frac{T}{\lambda} }}\big|\Delta_{j}   \big|^{n} \Big]^{\frac{1}{n}} \\ &= C' \mathbb{E}^{(\lambda)}\Big[\Big(\int_{0}^{\frac{T}{\lambda} }dt\mathcal{K}_{\lambda,2}(S_{t})\Big)^{\frac{n}{2}} \Big]^{\frac{1}{n}}+C' \mathbb{E}^{(\lambda)}\Big[ \int_{0}^{\frac{T}{\lambda} }dt\mathcal{K}_{\lambda,n}(S_{t}) \Big]^{\frac{1}{n}} \\ &\leq C'' \mathbb{E}^{(\lambda)}\Big[\Big(\int_{0}^{\frac{T}{\lambda} }dt\big(1+\lambda \mathbf{Q}_{t}\big)\Big)^{\frac{n}{2}} \Big]^{\frac{1}{n}}+C''\mathbb{E}^{(\lambda)}\Big[ \int_{0}^{\frac{T}{\lambda} }dt\big(1+\lambda \mathbf{Q}_{t}\big) \Big]^{\frac{1}{n}} \\ &\leq C''(T^{\frac{1}{2}}  \lambda^{-\frac{1}{2}}+T^{\frac{1}{n}}\lambda^{-\frac{1}{n}} )+C'' T^{\frac{1}{2}}\mathbb{E}^{(\lambda)}\Big[\sup_{0\leq t\leq \frac{T}{\lambda}} \mathbf{Q}_{t}^{n}  \Big]^{\frac{1}{2n}}+ C'' T^{\frac{1}{n}}\mathbb{E}^{(\lambda)}\Big[\sup_{0\leq t\leq \frac{T}{\lambda}} \mathbf{Q}_{t}^{n}  \Big]^{\frac{1}{n^{2}}}.
\end{align*}
The second inequality is for some $C'>0$ by Rosenthal's inequality (see e.g. \cite[Lem.~2.1]{Rosenthal}).  The third inequality is Part (4) of Prop.~\ref{AMinus}.  We have combined the constants at each step.  Thus we have  dressed our bound in the form~(\ref{TheGame}).

The last term in~(\ref{AchIII}) is bounded similarly to  $\mathbb{E}^{(\lambda)}\big[\sup_{0\leq t\leq \frac{T}{\lambda} }\big|\mathbf{m}_{t}\big|^{n}  \big]^{\frac{1}{n}}$.

\end{proof}

The following lemma bounds the expected number of returns to the atom up to time $\frac{T}{\lambda}$ for $\lambda\ll 1$.

\begin{lemma}\label{LocalTimeBnd}  There is a $C>0$ such that for all $\lambda<1$,
$$  \tilde{\mathbb{E}}^{(\lambda)}\big[ \lambda^{\frac{1}{2}}\tilde{N}_{\frac{T}{\lambda}}\big]\leq C.         $$

\end{lemma}

\begin{proof}  The main step in this proof is the inequality in~(\ref{Harkonnen}) where we bound the function $h(s)$, which  determines the   probability of life cycle expiration,  by a constant multiple of the function $\mathcal{A}_{\lambda}^{+}(s)$ arising as the increasing  part of the semi-martingale decomposition for $(2H_{t})^{\frac{1}{2}}$.   Once the quantity we must bound is formulated in terms of processes related to the square root of the energy process, we can apply Lem.~\ref{FirstEnergyLem} and results from Prop.~\ref{AMinus} to finish the proof.    By Prop.~\ref{TrivialMart} we have the first equality below:
 \begin{align}\label{Harkonnen}
\tilde{\mathbb{E}}^{(\lambda)}\big[ \lambda^{\frac{1}{2}}\tilde{N}_{\frac{T}{\lambda}}\big] &=\mathbb{E}^{(\lambda)}\Big[ \lambda^{\frac{1}{2}}\int_{0}^{\frac{T}{\lambda}}drh(S_{r})\Big]\nonumber 
 \\ &\leq c\mathbb{E}^{(\lambda)}\Big[ \lambda^{\frac{1}{2}}\int_{0}^{\frac{T}{\lambda}}dr\mathcal{A}_{\lambda}^{+}(S_{r})\Big] 
\nonumber \\ &= c \mathbb{E}^{(\lambda)}\big[ \lambda^{\frac{1}{2}}\mathbf{A}_{\frac{T}{\lambda}}^{+}  \big].
\end{align}
The second equality in~(\ref{Harkonnen}) follows by the definition of the process $\mathbf{A}_{t}^{+}$, and the  inequality  holds since there is $c>0$ such that for small enough $\lambda>0$ and all $s\in \Sigma$
\begin{align}\label{Ginger}
h(s)\leq c\mathcal{A}_{\lambda}^{+}(s).
\end{align}
To show~(\ref{Ginger}), we first observe that for $\lambda=0$   we have $\mathcal{A}_{0}^{+}(x,p)=\mathcal{A}_{0}(x,p)$ since for all $(x,p)\in \Sigma$
\begin{align}\label{Snap}
 \mathcal{A}_{0}(x,p)=&\int_{\R}dvj(v)\Big( 2^{\frac{1}{2}}H^{\frac{1}{2}}(x,p+v)-2^{\frac{1}{2}}H^{\frac{1}{2}}(x,p)      \Big)\nonumber \\=&V(x) \int_{0}^{\infty}dvj(v)\int_{-|v|}^{|v|}dw\frac{ |v|-|w|  }{ 2^{\frac{1}{2}}H^{\frac{3}{2}}(x,p+w)  }\nonumber \\ >& 0 ,
\end{align}
where the jump rate density $j(p-p')=\mathcal{J}_{0}(p,p')$ is defined as in~(\ref{LambdaZero}).  The second equality in~(\ref{Snap}) holds  by a  Taylor formula for $H^{\frac{1}{2}}(x,p+v)$ around $v=0$ with second-order error:
\begin{align*}
2^{\frac{1}{2}}H^{\frac{1}{2}}(x,p+v)-2^{\frac{1}{2}}H^{\frac{1}{2}}(x,p)=& v\frac{  p  }{2^{\frac{1}{2}}H^{\frac{1}{2}}(x,p)   }    + \int_{0}^{v}dw(v-w) \frac{  V(x) }{2^{\frac{1}{2}} H^{\frac{3}{2}}(x,p+w)  } ,
\end{align*}
where the first-order term vanishes from~(\ref{Snap})  because  $j(v)=j(-v)$, and the error term is symmetrized for the same reason.   The formula on the second line of~(\ref{Snap}) assumes $V(x)>0$ and can be replaced by a simpler expression when $V(x)=0$ (in which case $2^{\frac{1}{2}}H^{\frac{1}{2}}(x,p)=|p|$), although the resulting  value is still strictly positive for all $(x,p)$.  For fixed $\lambda>0$, the function $\mathcal{A}_{\lambda}^{+}(x,p)$  is supported in a compact region around the origin of phase space.  As $\lambda\searrow 0$ the bias that tends to drag the test  particle to lower energy due to head-on collisions becomes less pronounced, and  $\supp(\mathcal{A}_{\lambda}^{+})$ grows to all of phase space (as in the $\lambda=0$ case).   For any compact set $\mathcal{K}\subset \Sigma$ we can pick $\lambda>0$ small enough such that $\mathcal{K}\subset \supp(\mathcal{A}_{\lambda}^{+})$, and  there is uniform convergence for all $(x,p)\in \mathcal{K}$ as $\lambda\searrow 0$
\begin{align}\label{Snafu}
 \mathcal{A}_{\lambda}^{+}(x,p)=\mathcal{A}_{\lambda}(x,p)=\int_{\R}dv\mathcal{J}_{\lambda}(p,p+v)\Big( 2^{\frac{1}{2}}H^{\frac{1}{2}}(x,p+v)-2^{\frac{1}{2}}H^{\frac{1}{2}}(x,p)      \Big)\longrightarrow  \mathcal{A}_{0}(x,p)>0.       
\end{align}
This uniform convergence follows easily from the smooth dependence of the rates $\mathcal{J}_{\lambda}(p,p+v)$ on $\lambda$ and the uniform exponential decay of the rates in $v$ for all $p$ in a compact set.    In particular, we can pick a $\delta>0$ and  $\lambda>0$ small enough so that  $\mathcal{A}_{\lambda}^{+}(x,p) >\delta$ for all $(x,p)\in \mathcal{K}$.  
Since $h:\Sigma\rightarrow \R^{+}$ is bounded by one and has compact support, the above remarks imply that~(\ref{Ginger}) holds for some $c>0$, and  we thus have the inequality in~(\ref{Harkonnen}).

By~(\ref{Harkonnen}) it is sufficient to show that  $\mathbb{E}^{(\lambda)}\big[ \lambda^{\frac{1}{2}}\mathbf{A}_{\frac{T}{\lambda}}^{+}  \big]$ is uniformly bounded for $\lambda\ll 1$. Since $\mathbf{A}_{t}^{+}=\mathbf{Q}_{t}-\mathbf{Q}_{0}-\mathbf{M}_{t}+\mathbf{A}_{t}^{-}$,  the triangle inequality gives that
\begin{align}
 \mathbb{E}^{(\lambda)}\big[ \lambda^{\frac{1}{2} }\mathbf{A}_{\frac{T}{\lambda}}^{+}\big] \leq  &  2 \mathbb{E}^{(\lambda)}\big[\sup_{0\leq t\leq \frac{T}{\lambda}}\lambda^{\frac{1}{2} }\mathbf{Q}_{t} \big]+  \mathbb{E}^{(\lambda)}\big[\big|\lambda^{\frac{1}{2}}\mathbf{M}_{\frac{T}{\lambda} }\big| \big]+\mathbb{E}^{(\lambda)}\big[ \lambda^{\frac{1}{2} }\mathbf{A}_{\frac{T}{\lambda}}^{-}\big]   
\nonumber \\ \leq &  2 \mathbb{E}^{(\lambda)}\big[\sup_{0\leq t\leq \frac{T}{\lambda}}\lambda^{\frac{1}{2} }\mathbf{Q}_{t} \big]+   \mathbb{E}^{(\lambda)}\Big[\lambda\int_{0}^{\frac{T}{\lambda} }dr\mathcal{V}_{\lambda}(S_{r})  \Big]^{\frac{1}{2}}+\mathbb{E}^{(\lambda)}\Big[ \lambda^{\frac{1}{2}}\int_{0}^{\frac{T}{\lambda}}dr\mathcal{A}^{-}_{\lambda}(S_{r})\Big] \nonumber \\ \leq &   2 \mathbb{E}^{(\lambda)}\big[\sup_{0\leq t\leq \frac{T}{\lambda}}\lambda^{\frac{1}{2} }\mathbf{Q}_{t} \big]+   C_{1}^{\frac{1}{2}}\mathbb{E}^{(\lambda)}\Big[\lambda\int_{0}^{\frac{T}{\lambda} }dr(1+\lambda \mathbf{Q}_{r})^{2} \Big]^{\frac{1}{2}}\nonumber \\ &+C_{2}\mathbb{E}^{(\lambda)}\Big[ \lambda^{\frac{3}{2}}\int_{0}^{\frac{T}{\lambda}}dr(\mathbf{Q}_{r}+\lambda\mathbf{Q}_{r}^{2}) \Big]. \label{Kaput}
 \end{align}
For the second inequality, the second term employs Jensen's inequality with the square function along with the fact that the martingale $\mathbf{M}_{t}$ has bracket $\langle \mathbf{M}\rangle_{t}=\int_{0}^{t }dr\mathcal{V}_{\lambda}(S_{r})$.   The bound for the second term in the third equality is by Part (5) of Prop.~\ref{AMinus}. The third term in the third inequality is bounded by Part (1) of Prop.~\ref{AMinus} and $|P_{r}|\leq \mathbf{Q}_{r}$.  The right side of~(\ref{Kaput}) is uniformly bounded in $\lambda<1$ since $ \mathbf{Q}_{r} = (2H_{r})^{\frac{1}{2}} $  and $\mathbb{E}^{(\lambda)}\big[\sup_{0\leq r\leq \frac{T}{\lambda}} H_{r}^{\frac{n}{2}}  \big]\leq C_{n}'\lambda^{-\frac{n}{2}}$ for some constants $C'_{n}>0$ by Lem.~\ref{FirstEnergyLem}.

\end{proof}

\subsection{Fractional moments for the duration of life cycles}\label{SecFracMom}

By Appx.~\ref{AppendErgod}   the original dynamics is exponentially ergodic to the equilibrium state $\Psi_{\infty,\lambda}$ for   any fixed $\lambda$.  Thus the split dynamics converges exponentially to the $\tilde{\Psi}_{\infty,\lambda}$, and the time span $R_{1}$ and the number of partition times $\tilde{n}_{1}$ during a single life cycle will have finite expectation.  However, the process $S_{t}=(X_{t},P_{t})$ behaves more and more like a random walk in the $P_{t}$ variable as $\lambda\searrow 0$, so we should expect that
$$\tilde{\mathbb{E}}_{\tilde{\nu}}^{(\lambda)}\big[  R_{1}\big]\longrightarrow \infty \hspace{1cm} \text{and}\hspace{1cm} \tilde{\mathbb{E}}_{\tilde{\nu}}^{(\lambda)}\big[  \tilde{n}_{1}\big]\longrightarrow \infty$$
as $\lambda\searrow 0$ since the time elapsed during a random walk's excursion from a region around the origin has infinite expectation. However, the excursive durations for random walks  do have finite fractional moments for exponents $<\frac{1}{2}$, and Part (2) of  Prop.~\ref{FracMoment} states the analogous property for our process. 

\begin{proposition}\label{FracMoment}
Let $\tilde{n}_{1}$ and $R_{1}$ be defined as above. 
\begin{enumerate}
\item There is a $C>0$ such that for $\lambda<1$,
$$ \tilde{\mathbb{E}}^{(\lambda)}_{\tilde{\nu}}\big[  \tilde{n}_{1}\big]\leq C \lambda^{-\frac{1}{2}}\quad \text{and} \quad \tilde{\mathbb{E}}^{(\lambda)}_{\tilde{\nu}}\big[  R_{1}\big]\leq C \lambda^{-\frac{1}{2}}.   $$

\item Each fractional moment $0<\alpha<\frac{1}{2}$  is uniformly bounded for $\lambda<1$,
$$\sup_{\lambda<1}\tilde{\mathbb{E}}^{(\lambda)}_{\tilde{\nu}}\big[  \tilde{n}_{1}^{\alpha}\big]<\infty  \hspace{.5cm} \text{and} \hspace{.5cm} \sup_{\lambda<1}\tilde{\mathbb{E}}^{(\lambda)}_{\tilde{\nu}}\big[  R_{1}^{\alpha}\big]<\infty . $$
\end{enumerate}

\end{proposition}

  Before beginning with the proof of Prop.~\ref{FracMoment}, we must establish Lemmas~\ref{LemPartTime} and~\ref{LTLB} below.  The following trivial lemma bounds the length of time up to the first partition time $\tau_{1}$ independently of  the initial state $\tilde{s}\in \tilde{S}$.  Although the time intervals between partition times are not exponentially distributed, there is still an exponential bound on their densities.

\begin{lemma}\label{LemPartTime}
For $c=\max(\frac{U}{\mathbf{u}}, \frac{ U}{U-\mathbf{u}})$, the following inequality holds:
$$\sup_{\lambda<1}\sup_{\tilde{s}\in \tilde{\Sigma}}\tilde{\mathbb{E}}^{(\lambda)}_{ \tilde{s} }\big[ \delta_{t}( \tau_{1} )  \big]\leq ce^{-t},$$
where $\tilde{\mathbb{E}}^{(\lambda)}_{ \tilde{s} }\big[ \delta_{t}( \tau_{1} )  \big]$ refers to the density of the random variable $\tau_{1}$ at the value $t\geq 0$.  As a consequence, for any $n\in \mathbb{N}$,  
$$\sup_{\lambda\leq 1}\sup_{\tilde{s}\in\tilde{\Sigma}}\tilde{\mathbb{E}}_{\tilde{s}}^{(\lambda)}\big[ (R_{1}-R_{1}')^{n}   \big]<\infty . $$  
  
\end{lemma}

\begin{proof}
 In the original dynamics, $\tau_{1}$ has a mean one exponential distribution regardless of the initial state. Splitting the distribution starting from $s\in \Sigma $ yields the equality 
$$ e^{-t}=\mathbb{E}^{(\lambda)}_{ s }\big[ \delta_{t}( \tau_{1} )  \big] = \big(1-h(s)\big)\tilde{\mathbb{E}}^{(\lambda)}_{(s,0)  }\big[ \delta_{t}( \tau_{1} )  \big] + h(s)\tilde{\mathbb{E}}^{(\lambda)}_{(s,1)  }\big[ \delta_{t}( \tau_{1} )  \big]  .  $$
For $s\in \Sigma $ with $H(s)>l$, we have $h(s)=0$ and no splitting occurs, and thus $c=1$ is sufficient for the inequality.   For $s\in \Sigma$ with $H(s)\leq l$, then $c= \inf_{H(s)\leq l}\max( \frac{1}{h(s)}, \frac{1}{1-h(s)})=\max(\frac{U}{\mathbf{u}}, \frac{ U}{U-\mathbf{u}} ) $, where $l$, $\mathbf{u}$, and $U$ are defined as in Conv.~\ref{NumConv}.  The bound for the moments of $R_{1}-R_{1}'$ follows because $R_{1}$ is defined as the first partition time after $R_{1}'$ and by the strong Markov property for the chain $\tilde{\sigma}_{n}=\tilde{S}_{\tau_{n}}$ since $R_{1}'=\tau_{\tilde{n}_{1}}$ for the hitting time $\tilde{n}_{1}$

\end{proof}

\begin{lemma}\label{LTLB}
There exist  $c,C>0$ such that for all  $t\in \R^{+}$, $\lambda<1$, and $s$ in a given compact subset of $\Sigma$,
$$ \mathbb{E}^{(\lambda)}_{s}\big[ \mathbf{A}_{t}^{+} \big]\geq -C+c t^{\frac{1}{2}}.  $$

\end{lemma}

\begin{proof}
The positive-valued, increasing process  $ \mathbf{A}_{t}^{+}$ is difficult to analyze directly, so our strategy will be to write it using the other terms in the semi-martingale decomposition of $\mathbf{Q}_{t}$ as we did before at the end of the proof of Lem.~\ref{LocalTimeBnd}:   $ \mathbf{A}_{t}^{+}=  \mathbf{Q}_{t}-\mathbf{Q}_{0}-\mathbf{M}_{t}+\mathbf{A}_{t}^{-} $.  
In fact we can immediately throw away the positive terms $ \mathbf{Q}_{t}$, $\mathbf{A}_{t}^{-}$  in this expression for  $\mathbf{A}_{t}^{+}$ since we are looking for a lower bound; see (\ref{Phat}).   Our analysis will rely on  applications of Prop.~\ref{AMinus} and Lem.~\ref{FirstEnergyLem} to bound the remaining martingale term.

Since $\mathbf{Q}_{t}$ and $\mathbf{A}_{t}^{-}$ are positive and $\mathbf{A}_t^+\geq 0$ is increasing, we have the first inequality below:
\begin{align} \label{Phat}  \mathbf{A}_{t}^{+}= \mathbf{Q}_{t}-\mathbf{Q}_{0}-\mathbf{M}_{t}+\mathbf{A}_{t}^{-} & \geq -\mathbf{Q}_{0} +\sup_{0\leq r\leq t}-\mathbf{M}_{r}\nonumber \\  &\geq -\mathbf{Q}_{0} +\mathbf{M}_{t}^-, 
\end{align}
where $\mathbf{M}_{t}^-:=-\mathbf{M}_{t}\chi(\mathbf{M}_{t}\leq 0)$. Taking the expectation of both side gives 
$$  \mathbb{E}^{(\lambda)}_{s}\big[ t^{-\frac{1}{2}} \mathbf{A}_{t}^{+} \big]\geq  -2^{\frac{1}{2}}t^{-\frac{1}{2}} H^{\frac{1}{2}}(s)+ t^{-\frac{1}{2}}\mathbb{E}^{(\lambda)}_{s}\big[ \mathbf{M}_{t}^-\big] .  $$
Since  $\mathbf{M}_{t}$ has mean zero,  we have the equality below
$$2\mathbb{E}^{(\lambda)}_{s}\big[ \mathbf{M}_{t}^-\big] =\mathbb{E}^{(\lambda)}_{s}\big[ |\mathbf{M}_{t}|\big]\geq \frac{ \mathbb{E}^{(\lambda)}_{s}\big[ |\mathbf{M}_{t}|^{2}\big]^{2}   }{ \mathbb{E}^{(\lambda)}_{s}\big[ |\mathbf{M}_{t}|^{3}\big]    },$$
and the inequality is by Cauchy-Schwarz.  However, 
$$\mathbb{E}^{(\lambda)}_{s}\big[ |\mathbf{M}_{t}|^{2}\big]=  \mathbb{E}^{(\lambda)}_{s}\Big[\int_{0}^{t}dr\mathcal{V}_{\lambda}(S_{r}) \Big]\geq c t ,  $$
where $c>0$ is from Part (7) of Prop.~\ref{AMinus}.  For the first inequality below, we use Rosenthal's inequality  to produce a $C>0$ such that  
\begin{align*}
\mathbb{E}^{(\lambda)}_{s}\big[ |\mathbf{M}_{t}|^{3}\big]& \leq  C\mathbb{E}^{(\lambda)}_{s}\big[ \langle\mathbf{M}_{t}\rangle^{\frac{3}{2}}\big]+C\mathbb{E}^{(\lambda)}_{s}\Big[ \sum_{n=1}^{\calN_{t}}\big|\mathbf{M}_{t_{n}}-\mathbf{M}_{t_{n}^{-}}\big|^{3}\Big] \\ & = C\mathbb{E}^{(\lambda)}_{s}\big[ \langle\mathbf{M}_{t}\rangle^{\frac{3}{2}}\big]+C\mathbb{E}^{(\lambda)}_{s}\Big[ \sum_{n=1}^{\calN_{t}}\big|\mathbf{Q}_{t_{n}}-\mathbf{Q}_{t_{n}^{-}}\big|^{3}\Big] \\ & = C\mathbb{E}^{(\lambda)}_{s}\Big[\Big( \int_{0}^{t}dr\mathcal{V}_{\lambda}(S_{r})   \Big)^{\frac{3}{2}}+ \int_{0}^{t}dr\mathcal{V}_{\lambda,3}(S_{r})\Big]\\ &  \leq C'\mathbb{E}^{(\lambda)}_{s}\Big[\Big( \int_{0}^{t}dr(1+\lambda \mathbf{Q}_{r})^{2} \Big)^{\frac{3}{2}}+ \int_{0}^{t}dr(1+\lambda \mathbf{Q}_{r})^{4}\Big]\leq C''t^{\frac{3}{2}},
\end{align*}
where $t_{n}$ are the collision times and $\calN_{t}$ is the number of collisions up to time $t$.  The  first equality uses that $\mathbf{Q}_{t}$ and $\mathbf{M}_{t}$ differ by a continuous process and thus have the same jumps.  The second inequality is for some $C'>0$ by Part (5)  of Prop.~\ref{AMinus} along with the relation $|p|\leq 2^{\frac{1}{2}}H^{\frac{1}{2}}(x,p)$.  The last inequality is by Lem.~\ref{FirstEnergyLem}, and $C''$ is independent of $\lambda<1$.  

Putting our results together 
 $$ \mathbb{E}^{(\lambda)}_{s}\big[\mathbf{A}_{t}^{+} \big]\geq  -2^{\frac{1}{2}} H^{\frac{1}{2}}(s)+\frac{ \mathbb{E}^{(\lambda)}_{s}\big[ |\mathbf{M}_{t}|^{2}\big]^{2}   }{ \mathbb{E}^{(\lambda)}_{s}\big[ |\mathbf{M}_{t}|^{3}\big]    }\geq -2^{\frac{1}{2}} H^{\frac{1}{2}}(s)+\frac{ c^{2}   }{ C''  }t^{\frac{1}{2}},$$
which proves the lemma.

\end{proof}

\begin{proof}[Proof of Prop.~\ref{FracMoment}]\text{  }\\
\noindent Part (1):\hspace{.1cm}  By Part (1) of Prop.~\ref{BasicsOfNumII} applied to  the constant function $g(s)=1$, we have that
$$ \tilde{\mathbb{E}}^{(\lambda)}_{\tilde{\nu}}\big[  \tilde{n}_{1}+1\big]= \frac{1}{ \int_{\Sigma}ds\Psi_{\infty,\lambda}(s)h(s)    }= \lambda^{-\frac{1}{2}} \frac{(2\pi)^{\frac{1}{2}}}{ \int_{\Sigma}ds h(s)   }+\mathit{O}(1),    $$
where the order equality is for small $\lambda$.  The same equality holds with   $\tilde{\mathbb{E}}^{(\lambda)}_{\tilde{\nu}}\big[  \tilde{n}_{1}+1\big]$ replaced with $\tilde{\mathbb{E}}^{(\lambda)}_{\tilde{\nu}}\big[  R_{1}\big]$ by Part (2) of Prop.~\ref{BasicsOfNumII}.

\vspace{.5cm}

\noindent Part (2): \hspace{.1cm} We can prove the result through the Laplace transform by showing that there is a $C>0$ such that for all $\gamma\in \R^{+}$
$$ \sup_{\lambda<1}\big| \tilde{\mathbb{E}}^{(\lambda)}_{\tilde{\nu}}\big[  e^{-\gamma\tilde{n}_{1} }\big] -1    \big|\leq C\gamma^{\frac{1}{2}}\quad \text{and}\quad  \sup_{\lambda < 1} \big|\tilde{\mathbb{E}}^{(\lambda)}_{\tilde{\nu}}\big[  e^{-\gamma R_{1} }\big] -1    \big| \leq C\gamma^{\frac{1}{2}}.      $$
The proof for $R_{1}$ and $\tilde{n}_{1}$ are similar, and we focus on $R_{1}$.   
Also,  it is sufficient to prove the result with $R_{1}'$ rather than $R_{1}$ since, by Lem.~\ref{LemPartTime}, the random variable $R_{1}-R_{1}'$ has finite expectation.  We will study the following regimes for $\gamma$:
\begin{enumerate}[(i).]
\item $\gamma<\lambda $,

\item $\lambda \leq \gamma $ and $\gamma$ sufficiently small.   
    
\end{enumerate}

 The case (i) can be shown with a simple linearization around $\gamma=0$. 
As a result of Part (1), there exists a $C'>0$ such that 
  $$  \big|\tilde{\mathbb{E}}^{(\lambda)}_{\tilde{\nu}}\big[  e^{-\gamma R_{1}' }\big] -1    \big|\leq \gamma \tilde{\mathbb{E}}^{(\lambda)}_{\tilde{\nu}}\big[   R_{1}' \big] \leq C'\gamma\lambda^{-\frac{1}{2}}.   $$
When $\gamma<\lambda $ the bound on the right side is smaller than $C'\gamma^{\frac{1}{2}}$.
  
 For the regime (ii), we can no longer rely on the first derivative of the Laplace transform because the upper bound is growing as $\mathit{O}(\lambda^{-\frac{1}{2}})$ .   In the analysis below, we will  show that there is a $c>0$ such that for all $ \gamma $ and $\lambda\leq 1$,
\begin{align}\label{Porvo}
\big| \tilde{\mathbb{E}}^{(\lambda)}_{\tilde{\nu}}\big[  e^{-\gamma R_{1}' }\big] -1    \big|\leq \frac{1}{  c\mathbb{E}^{(\lambda)}_{\nu}\big[ \int_{0}^{\gamma^{-1}}dr\mathcal{A}_{\lambda}^{+}(S_{r})      \big]-c^{-1}\gamma^{-\frac{1}{4}}} .
\end{align}
 However, by Lem.~\ref{LTLB} there is $c'>0$ such that all $\gamma^{-1}\leq \lambda^{-1} $ and $s\in \text{Supp}(\nu)$,
\begin{align}\label{Cyclone}
 \mathbb{E}^{(\lambda)}_{s}\big[ \int_{0}^{\gamma^{-1}}dr\mathcal{A}_{\lambda}^{+}(S_{r})      \big]\geq c'\gamma^{-\frac{1}{2}}-c'^{-1}.
\end{align}
Combining~(\ref{Porvo}) and~(\ref{Cyclone}) yields statement (ii).  

In order to show~(\ref{Porvo}), we will show some preliminary bounds.  The difference between $\tilde{\mathbb{E}}^{(\lambda)}_{\tilde{\nu}}\big[  e^{-\gamma R_{1}' }\big]$ and  $1$ is smaller than 
\begin{align}\label{Gizmo}
\big| \tilde{\mathbb{E}}^{(\lambda)}_{\tilde{\nu}}\big[  e^{-\gamma R_{1}' }\big] -1    \big|\leq  \frac{\big| \tilde{\mathbb{E}}^{(\lambda)}_{\tilde{\nu}}\big[  e^{-\gamma R_{1}' }\big] -1    \big|}{\tilde{\mathbb{E}}^{(\lambda)}_{\tilde{\nu}}\big[  e^{-\gamma R_{1}' }\big]    }\leq \Big(\sum_{m=1}^{\infty} \tilde{\mathbb{E}}^{(\lambda)}_{\tilde{\nu}}\big[e^{-\gamma R_{1}' }  \big]^{m}     \Big)^{-1}  \leq \tilde{\mathbb{E}}^{(\lambda)}_{\tilde{\nu}}\Big[ \sum_{m=1}^{\infty}  e^{-\gamma R_{m}' }\Big]^{-1}.   
\end{align}
The third inequality follows since
$$\Big(\tilde{\mathbb{E}}^{(\lambda)}_{\tilde{\nu}}\big[e^{-\gamma R_{1}' }  \big]\Big)^{m}=\tilde{\mathbb{E}}^{(\lambda)}_{\tilde{\nu}}\big[e^{-\gamma \sum_{n=1}^{m} R_{n}'-R_{n-1} }  \big]   \geq \tilde{\mathbb{E}}^{(\lambda)}_{\tilde{\nu}}\big[   e^{-\gamma R_{m}' }\big],$$
where the $R_{n}'-R_{n-1}$ are independent by Part (2) of Prop.~\ref{IndependenceProp} and distributed as $R_{1}'$ when the initial distribution of the split process is $\tilde{\nu}$, and the inequality is from $\sum_{n=1}^{m} R_{n}'-R_{n-1}\leq R_{m}'$.  By~(\ref{Gizmo}) it is sufficient to give a lower bound for $\tilde{\mathbb{E}}^{(\lambda)}_{\tilde{\nu}}\big[ \sum_{m=1}^{\infty}  e^{-\gamma R_{m}' }\big]$.  This term can be rewritten as
\begin{align}\label{Oozie}
\tilde{\mathbb{E}}^{(\lambda)}_{\tilde{\nu}}\Big[ \sum_{m=1}^{\infty}  e^{-\gamma R_{m}' }\Big]= \tilde{\mathbb{E}}^{(\lambda)}_{\tilde{\nu}}\Big[ \sum_{m=0}^{\infty}  e^{-\gamma \tau_{m} }\chi( \zeta_{m}=1  )      \Big]=   \tilde{\mathbb{E}}^{(\lambda)}_{\tilde{\nu}}\Big[ \sum_{m=0}^{\infty}  e^{-\gamma \tau_{m} }h(\sigma_{m} )      \Big],
\end{align}
where $\sigma_{m}=S_{\tau_{m}}$ is the resolvent chain and $\zeta_{m}=Z_{\tau_{m}}$ is the binary component of the split chain.  The first equality in~(\ref{Oozie}) is from the definition of the times $R_{m}'=\tau_{\tilde{n}_{m}}$, and the second equality is by Part (3) of Prop.~\ref{BasicsOfNum}.  The right side above is equal to
\begin{align}
\tilde{\mathbb{E}}^{(\lambda)}_{\tilde{\nu}}\Big[ \sum_{m=0}^{\infty}  e^{-\gamma \tau_{m} }h(\sigma_{m} )      \Big] &=  \mathbb{E}^{(\lambda)}_{\nu}\Big[ \sum_{m=0}^{\infty}  e^{-\gamma \tau_{m} }h(\sigma_{m} )      \Big] =\nu(h)+
  \mathbb{E}^{(\lambda)}_{\nu}\Big[  \int_{0}^{\infty}dr   e^{-\gamma r}h(S_{r} )       \Big]\nonumber \\ &\geq    e^{-1  } \mathbb{E}^{(\lambda)}_{\nu}\Big[  \int_{0}^{\gamma^{-1} }dr  h(S_{r} )       \Big].\label{TheBirds}
 \end{align}
The first equality uses that argument of the expectation is a function of only the times $\tau_{m}$ and the resolvent chain $\sigma_{m}$   in order to revert to the original statistics. 
   The second equality in~(\ref{TheBirds}) holds since the $m=0$ term in the sum is $\mathbb{E}^{(\lambda)}_{\nu} [h(\sigma_{0})]=\nu(h)$ and  
 $$ \sum_{m=1}^{\mathbf{N}_{t}}  e^{-\gamma \tau_{m} }h(S_{\tau_{m}} )  - \int_{0}^{t}dr   e^{-\gamma r}h(S_{r} )   $$
  is a mean zero martingale which converges  to a limiting value as $t\rightarrow \infty$.   The above process is a martingale because the terms $e^{-\gamma \tau_{m} }h(\sigma_{m} )=e^{-\gamma \tau_{m} }h(S_{\tau_{m}} )$ in the sum occur with Poisson rate one.

Thus far we have shown that
\begin{align}\label{Bygones}
\big| \tilde{\mathbb{E}}^{(\lambda)}_{\tilde{\nu}}\big[  e^{-\gamma R_{1} }\big] -1    \big|\leq \frac{1}{  e^{-1}  \mathbb{E}^{(\lambda)}_{\nu}\Big[  \int_{0}^{\gamma^{-1} }dr h(S_{r} )       \Big] }.
\end{align}
Now we find a lower bound for $\mathbb{E}^{(\lambda)}_{\nu}\Big[  \int_{0}^{\gamma^{-1} }dr  h(S_{r} )       \Big]$ in terms of the same expression except with $h$ replaced by $\mathcal{A}_{\lambda}^{+}$.   Define the constant
 $$u_{\lambda}:= \frac{\int_{\Sigma}ds\Psi_{\infty,\lambda}(s)h(s)       }{ \int_{\Sigma}ds\Psi_{\infty,\lambda}(s)\mathcal{A}_{\lambda}^{+}(s)    } .     $$
 By the triangle inequality and going to the split statistics, 
\begin{align}\label{Orgasmatic}
\Big|&\mathbb{E}^{(\lambda)}_{\nu}\Big[  \int_{0}^{\gamma^{-1} }dr  h(S_{r} )       \Big]-u_{\lambda}\mathbb{E}^{(\lambda)}_{\nu}\Big[  \int_{0}^{\gamma^{-1} }dr  \mathcal{A}_{\lambda}^{+}(S_{r} ) \Big]\Big|\nonumber   \\ &\leq  \Big| \tilde{\mathbb{E}}^{(\lambda)}_{\tilde{\nu}}\Big[\sum_{n=0}^{\tilde{N}_{\gamma^{-1}}}\int_{R_{n}}^{R_{n+1}}dr\Big(h(S_{r})-u_{\lambda}\mathcal{A}_{\lambda}^{+}(S_{r})    \Big)\Big] \Big|  + \tilde{\mathbb{E}}^{(\lambda)}_{\tilde{\nu}}\Big[\sup_{0\leq n\leq  \tilde{N}_{\gamma^{-1}} }\int_{R_{n}}^{R_{n+1}}dr\Big(h(S_{r})+u_{\lambda}\mathcal{A}_{\lambda}^{+}(S_{r})    \Big)\Big]\nonumber  \\ &=  \tilde{\mathbb{E}}^{(\lambda)}_{\tilde{\nu}}\big[\tilde{N}_{\gamma^{-1}}+1\big] \left|\frac{ \int_{\Sigma}ds\Psi_{\infty,\lambda}(s)\Big(h(s)-u_{\lambda}\mathcal{A}_{\lambda}^{+}(s)    \Big)    }{  \int_{\Sigma}ds\Psi_{\infty,\lambda}(s)h(s)   } \right| +\mathit{O}\Big(\tilde{\mathbb{E}}^{(\lambda)}_{\tilde{\nu}}\big[\tilde{N}_{\gamma^{-1}}\big]^{\frac{1}{2}}\Big)=\mathit{O}(\gamma^{-\frac{1}{4}}) ,
\end{align}
where $\tilde{N}_{t}$ is the number of life cycles  completed by time $t$.  The inequality covers the leftover interval $[\gamma^{-1},R_{\tilde{N}_{\gamma^{-1}}+1}]$ of the integration. The first term on the third line of~(\ref{Orgasmatic}) is zero by the definition of $u_{\lambda}$.  Also, the first term on the second line  is equal to the first term  on the third line since $\tilde{S}_{R_{n}}$ has distribution $\tilde{\nu}$ when conditioned on $\tilde{\mathcal{F}}_{R_{n}'}$ by Part (1) of Prop.~\ref{IndependenceProp}  and by Part (2) of Prop.~\ref{BasicsOfNumII}.  The second term on the second line of~(\ref{Orgasmatic}) is bounded by a constant multiple of $\tilde{\mathbb{E}}^{(\lambda)}_{\tilde{\nu}}[\tilde{N}_{\gamma^{-1}}]$  by the same reasoning as in~(\ref{Harkonnen}).  Finally, $\tilde{\mathbb{E}}^{(\lambda)}_{\tilde{\nu}}\big[\tilde{N}_{\gamma^{-1}}\big] $ is  $\mathit{O}(\gamma^{-\frac{1}{2}})$ by the same argument as in the proof of Lem.~\ref{LocalTimeBnd}.

The constant $u_{\lambda}$ satisfies 
$$    u_{\lambda}= \frac{\int_{\Sigma}ds e^{-\lambda H(s)   }h(s)       }{ \int_{\Sigma}ds e^{-\lambda H(s)}\mathcal{A}_{\lambda}^{+}(s)    }=\frac{\int_{\Sigma}ds h(s)       }{ \int_{\Sigma}ds \mathcal{A}_{\lambda}^{+}(s)    } +\mathit{O}(\lambda^{\frac{1}{2}})= \mathbf{u}+\mathit{O}(\lambda^{\frac{1}{2}})\geq \frac{\mathbf{u}}{2},   $$
where $\mathbf{u}=\int_{\Sigma}ds h(s)$, and the inequality holds for $\lambda$ small enough. The third equality is by Part (3) of Prop.~\ref{AMinus}.   These observations imply that for small enough $\lambda>0$
$$\mathbb{E}^{(\lambda)}_{\nu}\Big[  \int_{0}^{\gamma^{-1} }dr  h(S_{r} )       \Big]\geq     \frac{\mathbf{u}}{2}\mathbb{E}^{(\lambda)}_{\nu}\Big[  \int_{0}^{\gamma^{-1} }dr  \mathcal{A}_{\lambda}^{+}(S_{r} ) \Big]-\mathit{O}(\gamma^{-\frac{1}{4}}). $$
Plugging this inequality into~(\ref{Bygones}) gives~(\ref{Porvo}).  

\end{proof}

\section{Bounding integral functionals over a life cycle} \label{SecSumFun}

In this section we prove Prop.~\ref{FurtherNum}, which effectively bounds the expected fluctuations for the momentum drift $D_{t}=\int_{0}^{t}dr\frac{dV}{dx}(X_{r})$ over the period of a single life cycle.

\begin{proposition}\label{FurtherNum}    \text{   }

\begin{enumerate}

\item For any $m\in \mathbb{N}$,  there is a $C>0$ such that 
$$    \sup_{\lambda\leq 1}  \tilde{\mathbb{E}}_{\tilde{\nu}}^{ (\lambda)}   \Big[ \sup_{0\leq t\leq R_{1}}\Big(\int_{0}^{t}dr \frac{dV}{dx}(X_{r}) \Big)^{2m}   \Big]< C.  $$

\item  There is a $C>0$ such that for all $(x,p,z)\in \tilde{\Sigma}$,
$$    \sup_{\lambda\leq 1}  \tilde{\mathbb{E}}_{(x,p,z) }^{ (\lambda)}   \Big[ \Big|\int_{0}^{R_{1}}dr \frac{dV}{dx}(X_{r}) \Big| \Big]< C\big(1+\log(1+|p|)\big).  $$

\end{enumerate}

\end{proposition}

 For the task of proving Prop.~\ref{FurtherNum}, we rely on the bounds stated in Thm.~\ref{LifeOperatorTwo} for the generalized resolvent $U^{(\lambda)}:L^{\infty}(\Sigma)\rightarrow L^{\infty}(\Sigma)$  given by 
\begin{align}\label{Rockies}
\hspace{2cm} \big(U^{(\lambda)}g\big)(s):=\mathbb{E}_{s}^{(\lambda)}\Big[\int_{0}^{\infty}dt g(S_{t})e^{-\int_{0}^{t}dr h(S_{r} )}    \Big],\hspace{1cm} s\in \Sigma,
\end{align}
where $h:\Sigma\rightarrow \R^{+}$ is defined as in Conv.~\ref{NumConv}.  The expression~(\ref{Rockies}) would have the form of a standard resolvent if the function $h$ were replaced by a constant. In coarse terms, the generalized resolvent $U^{(\lambda)}$ characterizes the expected duration  that the process $S_{t}$ resides in different regions of phase space before returning to the set $\textup{supp}(h)\subset \Sigma $ when beginning from a point $s\in \Sigma$.  The end time for the random walk to $\textup{supp}(h)$ is better described as a Poisson time with variable rate depending stochastically on $S_{t} $ through $h(S_{t}) $.   Operators of the form~(\ref{Rockies}) were introduced in~\cite{Neveu}, and the following theorem is from~\cite{Resolvent}.
\begin{theorem}\label{LifeOperatorTwo}
There is a $c>0$ such that for any $g\in L^{\infty}(\Sigma)$ with $g \geq 0$ and $|p|\leq \lambda^{-1}$,
\begin{eqnarray*}
\big(U^{(\lambda)}g\big)(x,p) &\leq & c\|g\|_{\infty}+c|p|\sup_{H'>\frac{1}{2}\lambda^{-2}}g(x',p')+c\int_{H'\leq \frac{1}{2} \lambda^{-2} }dp'dx'\big(1+\min(|p'|, |p|)\big)g(x',p'), \\
\| U^{(\lambda)}g\|_{\infty} & \leq & c\lambda^{-1}\sup_{H'>\frac{1}{2}\lambda^{-2}}g(x',p') +c\sup_{H'\leq \frac{1}{2}\lambda^{-2}}\big(U^{(\lambda)}g\big)(x',p'),    
\end{eqnarray*}  
where $H':=H(x',p')$.

\end{theorem}

The analysis in the proof of  Prop.~\ref{FurtherNum} also applies to Prop.~\ref{FurtherNumII}, which is easier because the ``velocity function" $g(x,p)=\frac{dV}{dx}(x)$ of  Prop.~\ref{FurtherNum} does not have explicit decay for $|p|\gg 1$.  The decay for $\frac{dV}{dx}(x)$ at high momentum only occurs as a time-averaged effect, which is exposed in Lem.~\ref{TheCs}.

\begin{proposition}\label{FurtherNumII}    \text{   }
Let $g:\Sigma\rightarrow \R$ satisfy that $ |g(x,p)|\leq \frac{ C}{1+|p|^{2}}  $ for some $C>0$ and all $(x,p)\in \Sigma$.

\begin{enumerate}

\item For any $m\in \mathbb{N}$,  there is a $C>0$ such that 
$$    \sup_{\lambda\leq 1}  \tilde{\mathbb{E}}_{\tilde{\nu}}^{ (\lambda)}   \Big[ \sup_{0\leq t\leq R_{1}}\Big(\int_{0}^{t}dr g(S_{r})  \Big)^{2m}   \Big]< C.  $$

\item  There is a $C>0$ such that for all $(x,p,z)\in \tilde{\Sigma}$,
$$    \sup_{\lambda\leq 1}  \tilde{\mathbb{E}}_{(x,p,z) }^{ (\lambda)}   \Big[ \Big|\int_{0}^{R_{1}}dr g(S_{r})  \Big| \Big]< C\big(1+\log(1+|p|)\big).  $$

\end{enumerate}

\end{proposition}

\subsection{An inequality for summation functionals over a life cycle }

Recall that $\sigma_{n}=S_{\tau_{n}}$  denotes the resolvent chain  and that $\tilde{\delta}_{s}=\chi(z=0) (1-h(s))\delta_{s}+\chi(z=1)h(s)\delta_{s}   $ is the splitting of the $\delta$-distribution at $s\in \Sigma$.   The following lemma states that the generalized resolvent $(U^{(\lambda)}g)(s)$ can be used 
to bound the expression $\tilde{\mathbb{E}}_{\tilde{\delta}_{s}}^{(\lambda)}\big[\sum_{n=1}^{\tilde{n}_{1}}  g(\sigma_{n})  \big]$. 

\begin{lemma}\label{ResolventToResolvent} The following inequality holds for all $g\in L^{\infty}(\Sigma, \R^{+})$, $\lambda>0$, and $s\in \Sigma$: 
$$\tilde{\mathbb{E}}_{\tilde{\delta}_{s}}^{(\lambda)}\Big[\sum_{n=1}^{\tilde{n}_{1}}  g(\sigma_{n})  \Big]\leq \big(U^{(\lambda)}g\big)(s)+\sup_{s\in \textup{Supp}(h) } \big(U^{(\lambda)}g\big)(s) . $$
\end{lemma}

\begin{proof}
 Let the function $h:\Sigma\rightarrow [0,1]$ and the measure $\nu$ on $\Sigma$ be defined as in Conv.~\ref{NumConv}.  Recall that $\sigma_{n}=S_{\tau_{n}}$ denotes the resolvent chain and has transition kernel $\mathcal{T}_{\lambda}$.  
Define the following operators on $L^{\infty}(\R)$:
$$
\mathcal{W}^{(\lambda)}=\sum_{n=0}^{\infty}\big((1-h)\mathcal{T}_{\lambda}\big)^{n}\quad\text{and}\quad \widetilde{\mathcal{W}}^{(\lambda)}=\sum_{n=0}^{\infty}\big(\mathcal{T}_{\lambda}-h\otimes \nu\big)^{n}.
$$
 In terms of the operators $\mathcal{T}_{\lambda}$ and $h\otimes \nu$ on $L^{\infty}(\R)$, we can write the expressions in the statement of Lem.~\ref{ResolventToResolvent} as
\begin{align}\label{Zip}
\tilde{\mathbb{E}}_{\tilde{\delta}_{s}}^{(\lambda)}\Big[\sum_{n=1}^{\tilde{n}_{1}}  g(\sigma_{n})  \Big]=&\big((\mathcal{T}_{\lambda}-h\otimes \nu )\widetilde{\mathcal{W}}^{(\lambda)}g\big)(s),
\\ U^{(\lambda)}=&\mathcal{T}_{\lambda}\mathcal{W}^{(\lambda)} \label{Zap} . 
\end{align}
The above representation for $\tilde{\mathbb{E}}_{\tilde{\delta}_{s}}^{(\lambda)}\big[\sum_{n=1}^{\tilde{n}_{1}}  g(\sigma_{n})  \big]$ can be found in~\cite{Nummelin}.  The representation for $U^{(\lambda)}$  can be understood through the alternative generalized resolvent form $U^{(\lambda)}=\frac{1}{h-\mathcal{L}} $; see~\cite{Meyn} for other representations.
By  the  identity~(\ref{Zap}) and $\mathcal{T}_{\lambda}-h\otimes \nu\leq \mathcal{T}_{\lambda}$, we have that
\begin{align}\label{Apricot}
\tilde{\mathbb{E}}_{\tilde{\delta}_{s}}^{(\lambda)}\Big[\sum_{n=1}^{\tilde{n}_{1}}  g(\sigma_{n})  \Big] &= \big((\mathcal{T}_{\lambda}-h\otimes \nu) \widetilde{\mathcal{W}}^{(\lambda)}g\big)(s)\nonumber \\ & \leq   \big(\mathcal{T}_{\lambda}\mathcal{W}^{(\lambda)}g\big)(s)+\sup_{s\in \Sigma} \Big(  \big((\mathcal{T}_{\lambda}-h\otimes \nu)\widetilde{\mathcal{W}}^{(\lambda)}g\big)(s)-\big((\mathcal{T}_{\lambda}-h\otimes \nu)\mathcal{W}^{(\lambda)}g\big)(s) \Big) .
\end{align}
With the identity $\widetilde{\mathcal{W}}^{(\lambda)}-\mathcal{W}^{(\lambda)}= \widetilde{\mathcal{W}}^{(\lambda)}  (h\mathcal{T}_{\lambda}-h\otimes \nu    )\mathcal{W}^{(\lambda)}$, we obtain the  equality below:
\begin{align}\label{Nectarine}
 (\mathcal{T}_{\lambda}-h\otimes \nu)\widetilde{\mathcal{W}}^{(\lambda)}g-(\mathcal{T}_{\lambda}-h\otimes \nu)\mathcal{W}^{(\lambda)}g &= (\mathcal{T}_{\lambda}-h\otimes \nu)\widetilde{\mathcal{W}}^{(\lambda)}  (h\mathcal{T}_{\lambda}-h\otimes \nu    )\mathcal{W}^{(\lambda)}g\nonumber \\ &\leq (\mathcal{T}_{\lambda}-h\otimes \nu)\widetilde{\mathcal{W}}^{(\lambda)}h\big(1_{\textup{Supp}(h)} \mathcal{T}_{\lambda}\mathcal{W}^{(\lambda)}\big)g\nonumber \\ & \leq \big\|1_{\textup{Supp}(h)} \mathcal{T}_{\lambda}\mathcal{W}^{(\lambda)}g \big\|_{\infty}.
\end{align}
For the first inequality above, we have thrown  away the negative term and used  that  $h=h1_{\textup{Supp}(h)}$.  The second inequality in~(\ref{Nectarine}) uses  that 
for all $s\in \Sigma$
$$ \big((\mathcal{T}_{\lambda}-h\otimes \nu)\widetilde{\mathcal{W}}^{(\lambda)}h\big)(s)\leq \big(\widetilde{\mathcal{W}}^{(\lambda)}h\big)(s)=1  . $$
The equality $  \big(\widetilde{\mathcal{W}}^{(\lambda)}h\big)(s)=1 $ holds due to the recurrence of the dynamics and otherwise we would have $ \big(\widetilde{\mathcal{W}}^{(\lambda)}h\big)(s)\leq 1 $. 
  To prove this  fact, note that the series form of $ \widetilde{W}^{(\lambda)}$ implies that the 
 function $F(s):=  (\widetilde{W}^{(\lambda)}h)(s)$ satisfies
\begin{align*}
F(s)=&  h(s)+\big( (\tau-h\otimes \nu ) F\big)(s)  =  h(s)+ (\tau F)(s) -h(s)  \nu( F) = \big(\tau F\big)(s),
\end{align*}  
where the third equality uses that $\nu (F)=1$, which is a consequence of recurrence and Thm. 3 of [28].   Thus the function $F(s)$ is invariant of the dynamics and must be constant and hence equal to one by the normalization $\nu (F)=1$.

Applying~(\ref{Nectarine}) in~(\ref{Apricot}) along with the identity~(\ref{Zap}), we have that
\begin{align*}
\tilde{\mathbb{E}}_{\tilde{\delta}_{s}}^{(\lambda)}\Big[\sum_{n=1}^{\tilde{n}_{1}}  g(\sigma_{n})  \Big]  \leq   \big(U^{(\lambda)}g\big)(s)+\sup_{s\in \supp(h)}\big(U^{(\lambda)}g\big)(s)   .
\end{align*}

\end{proof}

  The following lemma states that an additive functional of the resolvent chain $\sum_{n}g_{\lambda}(\sigma_{n})$ has arbitrary finite moments when the summation is over a single life cycle and  $g_{\lambda}\geq 0$ has sufficient decay at large momentum.  In other words, not much typically happens over a single life cycle. 

\begin{lemma}\label{Marsupial}
Let $g_{\lambda}:\Sigma\rightarrow \R^{+}$, and suppose that there is a $C>0$ such that 
$$\hspace{1cm}g_{\lambda}(x,p)\leq C\max\big(\frac{1}{1+|p|^{2}},\lambda\big).        $$
for all $\lambda<1$ and $(x,p)\in \Sigma$. Then,  
$$\hspace{2cm} \sup_{\lambda<1}\tilde{\mathbb{E}}_{\tilde{\nu}}^{(\lambda)}\Big[\Big(\sum_{n=0}^{\tilde{n}_{1}}  g_{\lambda}(\sigma_{n}) \Big)^{m}\Big]<\infty,  \hspace{1cm} m\in \mathbb{N}.$$

\end{lemma}

\begin{proof}
  For the case $m=1$ and  $f=g_{\lambda}$, we have the closed expression
\begin{eqnarray} 
\tilde{\mathbb{E}}_{\tilde{\nu}}^{(\lambda)}\Big[\sum_{n=0}^{\tilde{n}_{1}}f(\sigma_{n}) \Big]&=&\frac{ \int_{\Sigma}ds\Psi_{\infty,\lambda}(s)f(s)     }{\int_{\Sigma}ds\Psi_{\infty,\lambda}(s)h(s)     } \\ \label{Smugger} & \leq & \frac{ \int_{|p|\leq \lambda^{-1}}dxdp f(x,p) +\frac{2}{\lambda^{\frac{1}{2}}}\textup{erfc}(\lambda^{-\frac{1}{2} }  ) \sup_{|p|> \lambda^{-1}}f(x, p)  }{\int_{\Sigma}dxdp e^{-\lambda H(x,p)}h(x,p)     },
\end{eqnarray}
where $\textup{erfc}(q)=\int_{q }^{\infty}dp e^{-\frac{p^{2}}{2}}$ is the complementary error function, and  the equality holds by  Part (1) of Prop.~\ref{BasicsOfNumII}.   The right side of~(\ref{Smugger}) is finite by our conditions on $g_{\lambda}$ and since the denominator is approximately $\mathbf{u}=\int_{\Sigma}ds h(s)$ for small $\lambda$, and thus the right side of~(\ref{Smugger}) is bounded away from zero for $\lambda<1$.  

For $m=2$, we write 
\begin{align}
\tilde{\mathbb{E}}_{\tilde{\nu}}^{(\lambda)}\Big[\Big(\sum_{n=0}^{\tilde{n}_{1}}  g_{\lambda}(\sigma_{n})\Big)^{2} \Big] &= \tilde{\mathbb{E}}_{\tilde{\nu}}^{(\lambda)}\Big[\sum_{n=0}^{\tilde{n}_{1}}  g_{\lambda}^{2}(\sigma_{n})+ 2\sum_{ n<m}^{\tilde{n}_{1}}  g_{\lambda}(\sigma_{n}) g_{\lambda}(\sigma_{m})  \Big] \nonumber\\ &=\tilde{\mathbb{E}}_{\tilde{\nu}}^{(\lambda)}\Big[\sum_{n=0}^{\tilde{n}_{1}}  g_{\lambda}^{2}(\sigma_{n})+2g_{\lambda}(\sigma_{n})\tilde{\mathbb{E}}^{(\lambda)}\Big[    \sum_{ m=n+1}^{\tilde{n}_{1}} g_{\lambda}(\sigma_{m})   \,\Big|\,\tilde{\mathcal{F}}_{\tau_{n}^{-}}\Big] \Big] \nonumber\\   &\leq  \tilde{\mathbb{E}}_{\tilde{\nu}}^{(\lambda)}\Big[\sum_{n=0}^{\tilde{n}_{1}}  g_{\lambda}^{2}(\sigma_{n})+2 g_{\lambda}(\sigma_{n})\big(U^{(\lambda)}g_{\lambda}\big)(\sigma_{n}) +2 g_{\lambda}(\sigma_{n})B(g_{\lambda}) \Big] , \label{Laden}
\end{align}
where      $B:L^{\infty}(\Sigma,\R^{+})\rightarrow \R^{+}  $ is defined by
$$  B(g)=\sup_{ s\in \textup{Supp}(h) } \big(U^{(\lambda)}g\big)(s). $$
To see~(\ref{Laden}) recall that $\sigma_{n}:=S_{\tau_{n}}$ and that the $\sigma$-algebra  $\tilde{\mathcal{F}}_{\tau_{n}^{-}}$ contains knowledge of the state  $\sigma_{n}$.   The value $\sup_{\lambda<1} B(g_{\lambda})$ is finite by the bound we assumed for $g_{\lambda}$ and the bound on $U^{(\lambda)}g_{\lambda}$ from Thm.~\ref{LifeOperatorTwo}; see~(\ref{Gus}) below.
  The inequality in~(\ref{Laden}) applies  Part (3) of Prop.~\ref{BasicsOfNum}, to get the  equality below
\begin{align}\label{Hogs}
\tilde{\mathbb{E}}^{(\lambda)}\Big[    \sum_{ m=n+1}^{\tilde{n}_{1}} g_{\lambda}(\sigma_{m})   \,\Big|\,\tilde{\mathcal{F}}_{\tau_{n}^{-}}\Big] = \tilde{\mathbb{E}}_{\tilde{\delta}_{\sigma_{n}} }^{(\lambda)}\Big[    \sum_{ m=1}^{\tilde{n}_{1}} g_{\lambda}(\sigma_{m})   \Big]   \leq
 \big(U^{(\lambda)} g_{\lambda}\big)(\sigma_{n})+B(g_{\lambda}).
\end{align}
  The inequality in~(\ref{Hogs}) follows by an application of Lem.~\ref{ResolventToResolvent}.

We can apply~(\ref{Smugger}) with $f= g_{\lambda}^{2}+2g_{\lambda} U^{(\lambda)}g_{\lambda}+2g_{\lambda}B(g_{\lambda})$  to bound the right side of~(\ref{Laden}).  Clearly the contribution from  $g_{\lambda}^{2}$ is not a problem since  $g_{\lambda}^{2}(x,p)\leq Cg_{\lambda}(x,p)$. For  $g_{\lambda} U^{(\lambda)}(g_{\lambda})$ there are constants such that  
\begin{align*}
 g_{\lambda}(p)\big( U^{(\lambda)}g_{\lambda}\big)(p)& \leq (\textup{const})\frac{1+\log(1+|p|)   }{|p|^{2}   }, &&\text{for all } |p|\leq \lambda^{-1},\,\lambda<1 \\ g_{\lambda}(p) \big(U^{(\lambda)}g_{\lambda}\big)(p)& \leq  (\textup{const})\frac{1+\log(1+|\lambda^{-1}|)   }{|\lambda|^{-2}   },
 &&\text{for all } |p|> \lambda^{-1},\,\lambda<1.
\end{align*}
We have applied Thm.~\ref{LifeOperatorTwo} along with our conditions on $g_{\lambda}$ to get
\begin{align}\label{Gus} \big( U^{(\lambda)} g_{\lambda}\big)(x,p)\leq c+c\int_{0}^{|p|}dp^{\prime}p^{\prime}|g_{\lambda}(p^{\prime})|\leq c+c'\int_{0}^{|p|}dp^{\prime}\frac{p^{\prime}}{1+|p^{\prime}|^{2} }\leq c+c''\log(1+|p|), 
\end{align}
for some $c,c',c''>0$ and all $\lambda<1$ and $|p|\leq \lambda^{-1}$.  The  inequality above follows similarly for the domain $|p|> \lambda^{-1}$.

Now we sketch the proof for the general case $m>2$.  For $\epsilon_{j}\in\{<,=\}$ and $j<m$, let the set $\ell^{(\tilde{n}_{1})}(\epsilon_{1},...,\epsilon_{m-1})  $ be the collection of all $(r_{1},\dots ,r_{m})\in [0,\tilde{n}_{1}]^{m}$ satisfying the relations $$ r_{1}\,\epsilon_{1}\, r_{2} \,\dots \epsilon_{m-1} \,r_{m}.  $$  Also define,
$$f_{(\epsilon_{1},\dots, \epsilon_{m-1})}=  A_{\epsilon_{1}}\cdots A_{\epsilon_{m-1}} g_{\lambda},$$
where $ A_{=},A_{<}$ are maps on $L^{\infty}(\Sigma,\R^{+})$ in which   $A_{=}$ is multiplication by $g_{\lambda}$ and $A_{<}=A_{=}(U^{(\lambda)}+B)   $.  We can write
$$\tilde{\mathbb{E}}_{\tilde{\nu}}^{(\lambda)}\Big[\Big(\sum_{n=0}^{\tilde{n}_{1}}  g_{\lambda}(\sigma_{r_{n}}) \Big)^{m}\Big]= \begin{array}{c}\text{Lin. comb. over}\\ (\epsilon_{1},\dots ,\epsilon_{m-1})\in \{<,=\}^{m-1}    \end{array}\quad \tilde{\mathbb{E}}_{\tilde{\nu}}^{(\lambda)}\Big[\sum_{ \ell^{(\tilde{n}_{1})}(\epsilon_{1},\dots ,\epsilon_{m-1}) }  g_{\lambda}(\sigma_{r_{1}})\cdots g_{\lambda}(\sigma_{r_{m}}) \Big].$$
 However, the following inequality holds:
\begin{align}\label{Kuperberg}
 \tilde{\mathbb{E}}_{\tilde{\nu}}^{(\lambda)}\Big[\sum_{ \ell^{(\tilde{n}_{1})}(\epsilon_{1},\dots ,\epsilon_{m-1}) }  g_{\lambda}(\sigma_{r_{1}})\cdots g_{\lambda}(\sigma_{r_{m}}) \Big]\leq \tilde{\mathbb{E}}_{\tilde{\nu}}^{(\lambda)}\Big[ \sum_{n=0}^{\tilde{n}_{1}} f_{(\epsilon_{1},\dots, \epsilon_{m-1})}(\sigma_{n} ) \Big],  
 \end{align}
because we can write the difference between the right and left side of~(\ref{Kuperberg}) as a sum of positive terms $\mathbf{c}_{v-1}-\mathbf{c}_{v}$ indexed by $v\in [1,m-1]$,  where
$$
\mathbf{c}_{v}=\tilde{\mathbb{E}}_{\tilde{\nu}}^{(\lambda)}\Big[\sum_{ \ell^{(\tilde{n}_{1})}(\epsilon_{1},\dots ,\epsilon_{v}) }  g_{\lambda}(\sigma_{r_{1}})\cdots g_{\lambda}(\sigma_{r_{v}})f_{(\epsilon_{v+1},\dots, \epsilon_{m-1})}(\sigma_{r_{v+1}} ) \Big].
$$
When $\epsilon_{v-1}$ is $=$, then $\mathbf{c}_{v-1}$ and $\mathbf{c}_{v}$ are identically equal. When $\epsilon_{v-1}$ is $<$, then the difference $\mathbf{c}_{v}-\mathbf{c}_{v-1}$ is equal to 
\begin{align*}
&\tilde{\mathbb{E}}_{\tilde{\nu}}^{(\lambda)}\Big[\sum_{ \ell^{(\tilde{n}_{1})}(\epsilon_{1},\dots ,\epsilon_{v-1}) }  g_{\lambda}(\sigma_{r_{1}})\cdots g_{\lambda}(\sigma_{r_{v-1}})\Big(g_{\lambda}(\sigma_{r_{v}}) \sum_{n=r_{v}+1}^{\tilde{n}_{1}} f_{(\epsilon_{v+1},\dots, \epsilon_{m-1})}(\sigma_{n} ) -f_{(\epsilon_{v},\dots, \epsilon_{m-1})}(\sigma_{r_{v}} ) \Big) \Big]\\ &=
\tilde{\mathbb{E}}_{\tilde{\nu}}^{(\lambda)}\Big[\sum_{ \ell^{(\tilde{n}_{1})}(\epsilon_{1},\dots ,\epsilon_{v-1}) }  g_{\lambda}(\sigma_{r_{1}})\cdots g_{\lambda}(\sigma_{r_{v-1}}) \\ &\quad\quad\times \Big(g_{\lambda}(\sigma_{r_{v}}) \tilde{\mathbb{E}}^{(\lambda)}\Big[\sum_{n=r_{v}+1}^{\tilde{n}_{1}} f_{(\epsilon_{v+1},\dots, \epsilon_{m-1})}(\sigma_{n} )\Big|\tilde{\mathcal{F}}_{\tau_{r_{v}}^{-} }\Big] -f_{(\epsilon_{v},\dots, \epsilon_{m-1})}(\sigma_{r_{v}} ) \Big) \Big]
\\ &\leq \tilde{\mathbb{E}}_{\tilde{\nu}}^{(\lambda)}\Big[\sum_{ \ell^{(\tilde{n}_{1})}(\epsilon_{1},\dots ,\epsilon_{v-1}) }  g_{\lambda}(\sigma_{r_{1}})\cdots g_{\lambda}(\sigma_{r_{v-1}})\\ &\quad \quad \times \Big(g_{\lambda}(\sigma_{r_{v}})\big( U^{(\lambda)}f_{(\epsilon_{v},\dots, \epsilon_{m-1})}\big)(\sigma_{r_{v}} )+f_{(\epsilon_{v},\dots, \epsilon_{m-1})}(\sigma_{r_{v}})B\big(f_{(\epsilon_{v},\dots, \epsilon_{m-1})}\big) -f_{(\epsilon_{v},\dots, \epsilon_{m-1})}(\sigma_{r_{v}} ) \Big) \Big]\\ &=0,
\end{align*}
where the inequality follows from the strong Markov property and  Lem.~\ref{ResolventToResolvent} by the same argument as in~(\ref{Hogs}).

We are left to bound $\tilde{\mathbb{E}}_{\tilde{\nu}}^{(\lambda)}\big[ \sum_{n=0}^{\tilde{n}_{1}} f_{(\epsilon_{1},\dots, \epsilon_{m-1})}(\sigma_{n} ) \big]$.  The worst case scenario is when all the $\epsilon_{j}$ are equal to $<$ because mere multiplication by $g_{\lambda}(p)$ introduces more decay for large $|p|$. 
By our conditions on $g_{\lambda}$ and $m-1$ applications of Thm.~\ref{LifeOperatorTwo},  
$$
\hspace{3cm}\big((U^{(\lambda)})^{m-1}g_{\lambda}\big)(x,p)\leq c^{m-1}\frac{\big(1+\log(1+|p|)\big)^{m-1}}{1+|p|^{2}},\hspace{1cm} |p|\leq \lambda^{-1},
$$
and we get another bound for $|p|>\lambda^{-1}$ that is smaller than a fixed multiple of $\lambda^{-2m-1}$ for all $\lambda<1$.   Applying the inequality~(\ref{Smugger}), we obtain the bound.

\end{proof}

\subsection{Inequalities for the momentum drift} \label{SecMomentumDrift}

The first two parts in the lemma below follow from the conservation of energy and the quadratic formula and do not depend on the potential being periodic.  The third part of Lem.~\ref{DetDrift} is a statement about mixing on the torus.  If the particle begins with a high momentum $|P_{0}|\gg 1$ and is stopped at a random exponential time $\tau$, then the distribution on the torus $\mathbb{T}=[0,1)$ at the stopping time will be roughly uniform--even in the presence of the bounded periodic potential $V(x)$. 

\begin{lemma} \label{DetDrift}
Let $(X_{t},\,P_{t})$ evolve according to the Hamiltonian $H(x,p)=\frac{1}{2}p^{2}+V(x)$, for a positive potential $V(x)$ with   $\sup_{x}\big|\frac{dV}{dx}(x)\big|<\infty$.  If the initial momentum has $|P_{0}|^{2}>4\sup_{x}V(x)$, then the difference $P_{t}-P_{0}=-\int_{0}^{t}dr\frac{dV}{dx}(X_{r})$ satisfies the inequalities
\begin{enumerate}

\item $\sup_{t\in \R^{+}} \big|\int_{0}^{t}dr\frac{dV}{dx}(X_{r})  \big|\leq 2\sup_{x}V(x) |P_{0}|^{-1}  $, and

\item  $\Big| -\int_{0}^{t}dr\frac{dV}{dx}(X_{s})-\frac{ V(X_{t})-V(X_{0})    }{ P_{0}    }\Big|\leq  2t\sup_{x}\big|\frac{dV}{dx}(x)\big| \sup_{x}V(x) |P_{0}|^{-2}  $ .

\item  Suppose further that $V(x)$ has period one.  If $\tau$ is exponentially distributed with mean $\mathbf{r}^{-1}$ and $F:\mathbb{T}\rightarrow \R$ is a function on the torus and bounded, then
$$ \Big|\mathbb{E}_{(X_{0},P_{0})}\big[F(X_{\tau})\big]-\int_{\mathbb{T}}dx F(x)\Big|\leq \mathbf{r}\|F\|_{\infty}|P_{0}|^{-1}+\mathit{O}(|P_{0}|^{-2}).    $$

\end{enumerate}

\end{lemma}

\begin{proof} \text{  }\\
\noindent Part (1):\hspace{.25cm} Since $|P_{0}|^{2}>4\sup_{x}V(x)$, the momentum $P_{t}$ will not change sign at any time. By the conservation of energy
$$\frac{1}{2}\big|P_{0}+(P_{t}-P_{0})\big|^{2}-\frac{1}{2}P_{0}^{2}=-V(X_{t})+V(X_{0}).$$
Using the quadratic formula and that $P_{t},\,P_{0}$ have the same
sign,
$$|P_{t}-P_{0}|=\Big||P_{0}|-\big(P_{0}^{2}+2V(X_{0})-2V(X_{t})\big)^{\frac{1}{2}}\Big|\leq \Big| \frac{1}{2}\int_{0}^{2V(X_{0})-2V(X_{t})}dw \big(P_{0}^{2}+w\big)^{-\frac{1}{2}}\Big|<\frac{2\sup_{x}V(x)}{|P_{0}|},  $$
since $\big(P^{2}_{0}+w\big)^{-\frac{1}{2}}\leq \sqrt{2}|P_{0}|^{-1}<2|P_{0}|^{-1}$   for $|w|\leq \frac{1}{2}\,P_{0}^{2}$. 

\vspace{.5cm}

\noindent Part (2):\hspace{.25cm} With the identity $V(X_{t})-V(X_{0})=\int_{0}^{t}dr \frac{dV}{dx}(X_{r})P_{r}$, then 
\begin{align*}
\Big| \int_{0}^{t}dr \frac{dV}{dx}(X_{r})-\frac{ V(X_{t})-V(X_{0})    }{ P_{0}    }\Big| \leq  & \int_{0}^{t}dr \Big|\frac{dV}{dx}(X_{r})\big(1-\frac{P_{r}}{P_{0}}    \big)\Big| \\ \leq  & t |P_{0}|^{-1}\sup_{x}\big|\frac{dV}{dx}(x)\big|\sup_{r}|P_{r}-P_{0}\big|\\ \leq & 2t\sup_{x}\big|\frac{dV}{dx}(x)\big| \sup_{x}V(x)\, |P_{0}|^{-2} ,  
\end{align*}
where we applied Part (1) for the last inequality.  

\vspace{.5cm}

\noindent Part (3):\hspace{.25cm} Let $d_{s}:\mathbb{T}\rightarrow \R^{+}$ be the density of the particle on the torus at time $\tau$  starting from the point $s=(X_{0},P_{0})\in \Sigma$.  We have that
\begin{align*}
\mathbb{E}_{s}\big[ F(X_{\tau})  \big]=\int_{\mathbb{T}}dx d_{s}(x) F(x).
\end{align*}
This leads to the simple bound
\begin{align}\label{Homefry}
\Big| \mathbb{E}_{s}\big[ F(X_{\tau})  \big]- \int_{\mathbb{T}}dx F(x) \Big| \leq \|F\|_{\infty}\| d_{s}-1\|_{1}.    
\end{align}
Thus it is sufficient for us bound the $1$-norm of $d_{s}-1$, and, in fact, our bounds can be made in the supremum norm.   Our method for bounding~(\ref{Homefry}) will be to analyze a closed form for $d_{s}(x)$ that is possible due to the periodic form of the particle's trajectory $X_{t}$, $t\geq 0$.

 Notice that $d_{s}$ can be written as
\begin{align}\label{Liberia}
d_{s}(a)=\sum_{n=1}^{\infty} \frac{ \mathbf{r}\,e^{-\mathbf{r}\,t_{n}(a)} }{ |P_{t_{n}(a)}|    }= \frac{ \mathbf{r} e^{-\mathbf{r}\,t_{1}(a)}}{  |P_{t_{1}(a)}| }\sum_{n=0}^{\infty}\,e^{-\mathbf{r}\,n\Delta}=\frac{ \mathbf{r} e^{-\mathbf{r}\,t_{1}(a)}}{  |P_{t_{1}(a)}| \big(1-   e^{-\mathbf{r}\Delta} \big)     },
\end{align}
where $t=t_{1}(a),\, t_{2}(a),\cdots$ are the periodic sequence of times at which $X_{t}\,\textup{mod}(1)=a$ and $\Delta$ is the increment between successive times $t_{n}(a)$.  These times will exist for every $a\in \mathbb{T}$ as long as $H(s)>\sup_{x}V(x) $.
 If $4\sup_{x}V(x)\leq P_{0}^{2}$, then  $|P_{t}-P_{0}|\leq 2\big(\sup_{x}V(x)\big)|P_{0}|^{-1}$ by Part (1). Thus for large initial  momentum, $|P_{0}|\gg (\sup_{x}V(x))^{\frac{1}{2}}$, the momentum process $P_{r}$ is nearly constant $\approx P_{0}$ and  the period $\Delta$ is close to $\frac{1}{|P_{0}|}$.   To get a precise bound for the difference between $\Delta $ and $\frac{1}{|P_{0}|}$, notice that  when $|P_{0}|$ is large enough so that $|P_{t}-P_{0}|\leq 2\sup_{x}V(x)|P_{0}|^{-1}<\frac{1}{2}|P_{0}|$, then clearly $\frac{1}{2|P_{0}|}\leq \Delta\leq \frac{2}{|P_{0}|}$ since the particle always travels with speeds $|P_{t}|\in [\frac{1}{2}|P_{0}|,\frac{3}{2}|P_{0}|]$.   Hence, the difference between $\Delta$ and $\frac{1}{|P_{0}|}$ is smaller than
\begin{align}
\Big|\Delta-\frac{1}{|P_{0}|}\Big|&\leq \frac{1}{|P_{0}|}\Big|\int_{0}^{\Delta}dr\, P_{0}-S(P_{0})\Big| \leq \frac{1}{|P_{0}|} \int_{0}^{\Delta}dr |P_{r}- P_{0}| \leq  \frac{4\sup_{x}V(x)}{|P_{0}|^{3}},\label{TimePeriod}
\end{align}
where the second inequality uses that $\int_{0}^{\Delta}drP_{r}=S(P_{0})$.  

Using the triangle inequality
\begin{align}
\big| d_{s}(a)-1  \big|&\leq \Big|d_{s}(a)-\frac{ \mathbf{r} e^{-\mathbf{r}\,t_{1}(a)}}{  |P_{0}| \big(1-   e^{-\mathbf{r}\Delta} \big)     } \Big|+\Big|\frac{ \mathbf{r} e^{-\mathbf{r}\,t_{1}(a)}}{  |P_{0}| \big(1-   e^{-\mathbf{r}\Delta} \big)     }-\frac{ \mathbf{r} e^{-\mathbf{r}\,t_{1}(a)}}{  |P_{0}| \big(1-   e^{-\frac{\mathbf{r}}{ |P_{0}|} } \big)     }\Big|\nonumber \\ &\quad+  \Big|\frac{ \mathbf{r} e^{-\mathbf{r}\,t_{1}(a)}}{  |P_{0}| \big(1-   e^{-\frac{\mathbf{r}}{|P_{0}|}  } \big)     }-1\Big|\nonumber\\ &\leq \frac{2\mathbf{r}}{|P_{0}|}+\mathit{O}(\frac{1}{|P_{0}|^{2}})   ,\label{Aqui}
\end{align}
where the last inequality follows by further computations using the inequalities above.  For instance, we can  bound the first term on the first line of~(\ref{Aqui})  by
\begin{align*}
\Big|d_{s}(a)-\frac{ \mathbf{r} e^{-\mathbf{r}\,t_{1}(a)}}{  |P_{0}| \big(1-   e^{- \mathbf{r}\Delta   } \big)     }  \Big|\leq  \frac{\big| P_{t_{1}(a)}- P_{0}\big|}{|P_{t_{1}(a)}|\, |P_{0}| }\,  \frac{ \mathbf{r} }{  \big(1-   e^{- \mathbf{r}\Delta   } \big)     }  \leq  \frac{ 4  \sup_{x}V(x) }{ \Delta  |P_{0}|^{3} }   \leq   \frac{ 8  \sup_{x}V(x) }{   |P_{0}|^{2} }  ,
\end{align*}
where the inequalities hold for sufficiently large $|P_{0}|$.  The first inequality above follows from Part (1) and the second inequality uses that $\Delta \geq \frac{1}{2|P_{0}|}$ by the remark above~(\ref{TimePeriod}).

\end{proof}

Define the functions $\mathbf{C}^{(\lambda)}_{n}:\Sigma \rightarrow \R$,
\begin{eqnarray*}
 \mathbf{C}^{(\lambda)}_{0}(s) &=&\tilde{\mathbb{E}}_{\tilde{\delta}_{s} }^{(\lambda)}\Big[\chi(Z_{\tau_{1}}=0) \int_{\tau_{1}}^{\tau_{2} }dr\frac{dV}{dx}(X_{r})    \Big],  \\ \mathbf{C}^{(\lambda)}_{n}(s) &=&\tilde{\mathbb{E}}_{\tilde{\delta}_{s}}^{(\lambda)}\Big[ \Big(\chi(Z_{\tau_{1}}=0) \int_{\tau_{1}}^{\tau_{2}}dr\frac{dV}{dx}(X_{r})   - \mathbf{C}^{(\lambda)}_{0}(s)   \Big)^{2n}\Big],  \quad n\geq 1,    
 \end{eqnarray*}
where $\tau_{1},\tau_{2}$ are the first two partition times and $\tilde{\delta}_{s}= \big(1-h(s)\big)\delta_{(s,0)}+h(s)\delta_{(s,1)}    $, i.e., the splitting of the $\delta $-distribution at $s\in \Sigma $.  The presence of the factor $\chi(Z_{\tau_{1}}=0)$ in the above definitions is a small technical precaution, and if  $\chi(Z_{\tau_{1}}=0)$ is removed in the formula for $\mathbf{C}^{(\lambda)}_{0}(s) $, then we have
$$
\tilde{\mathbb{E}}_{\tilde{\delta}_{s} }^{(\lambda)}\Big[ \int_{\tau_{1}}^{\tau_{2} }dr \frac{dV}{dx}(X_{r})    \Big]=\mathbb{E}_{s }^{(\lambda)}\Big[ \int_{\tau_{1}}^{\tau_{2} }dr\frac{dV}{dx}(X_{r})    \Big]=\mathbb{E}_{s}^{(\lambda)}\Big[\int_{0}^{\infty}dt\,t\,e^{-t}\frac{dV}{dx}(S_{t})    \Big].    $$
   Part (1) of Lem.~\ref{DetDrift} is the main tool in  the proof of Part (1) of Lem.~\ref{TheCs}, and the proof for Part (2) of Lem.~\ref{TheCs} makes use of Parts (2) and (3) of Lem.~\ref{DetDrift} with $F(x):=V(x)$. 

\begin{lemma}\label{TheCs}
For any $n>1$, there exists a $C>0$ such that for all $\lambda<1$ and $(x,p)\in \Sigma$,
\begin{enumerate}
 \item $\big| \mathbf{C}^{(\lambda)}_{n}(x,p) \big| \leq C\textup{max}\big(\frac{1}{1+|p|^{2n} }, \lambda^{2n} \big)  $,  

\item $\big| \mathbf{C}^{(\lambda)}_{0}(x,p)\big|\leq C\textup{max}\big(\frac{1}{1+|p|^{2}}, \lambda  \big)$.
\end{enumerate}

\end{lemma}

\begin{proof}\text{ }\\
\noindent Part (1):\hspace{.25cm} For $v=2n$, notice that $\mathbf{C}^{(\lambda)}_{n}(s)$  is smaller than 
\begin{align}\label{Waffle}
\mathbf{C}^{(\lambda)}_{n}(s)\leq \tilde{\mathbb{E}}_{\tilde{\delta}_{s} }^{(\lambda)}\Big[ \Big|\int_{\tau_{1}}^{\tau_{2}}dr \frac{dV}{dx}(X_{r})  \Big|^{v }  \Big]=\mathbb{E}_{s }^{(\lambda)}\Big[ \Big|\int_{\tau_{1}}^{\tau_{2}}dr \frac{dV}{dx}(X_{r})  \Big|^{v }  \Big],
\end{align}
where the equality holds since the initial distribution $\tilde{\delta}_{s}$ is the splitting of $\delta_{s}$, and  the argument of the expectations only depends on the original (pre-split) statistics.  The quantity on the right side of~(\ref{Waffle}) is closely related to Part (1) of Lem.~\ref{DetDrift} except that the momentum now makes random jumps and the limits of integration $\tau_{1}$, $\tau_{2}$ are also random.   The randomness of the limits of integration is not very important here except that the integration interval should not be too long, so, for simplicity, we will bound  $\mathbb{E}_{s }^{(\lambda)}\big[ \big|\int_{0}^{\tau_{1}}dr\frac{dV}{dx}(X_{r})  \big|^{v }  \big]$ rather than the expression on the right side of~(\ref{Waffle}).

 The analysis must be split into cases  based on the size of the initial momentum $p$ since a particle  with high momentum $|p|\gg  \lambda^{-1}$, $\lambda\ll 1$ will tend to receive many collisions in a small amount of time, which contrasts with the situation  $|p|\leq \lambda^{-1}$ where only several collisions are likely to occur in the time interval $[0,\tau_{1}]$.  In the high momentum situation $|p|\gg \lambda^{-1}$, the absolute value of the momentum is likely to drift downwards due to the higher frequency of collisions with oncoming particles.  We will bound $ \mathbb{E}_{s }^{(\lambda)}\big[ \big|\int_{0}^{\tau_{1} }dr\,\frac{dV}{dx}(X_{r})  \big|^{v }  \big]$ for $s=(x,p)$ in the following three regimes:
\begin{enumerate}[(i).]
\item arbitrary $p$,

\item  $1\ll |p|\leq \lambda^{-1}$,

\item  $ \lambda^{-1}<|p|$.

\end{enumerate}
   We will use that the escape rate function $\mathcal{E}_{\lambda}(p):=\int_{\R}dp'\mathcal{J}_{\lambda}(p,p')$ has bounds of the form
\begin{align}\label{EscapeComment}
 \frac{1}{8(1+\lambda)} \leq   \mathcal{E}_{\lambda}(p)\leq \frac{1}{8(1+\lambda)}\big(1+C\lambda |p|\big)
\end{align}
for some $C>0$ and all $\lambda<1$ and $p\in \R$, which can be deduced easily from the form of the jump rates $\mathcal{J}_{\lambda}(p,p')$.   We will also use that for small $\lambda<1$
\begin{align}\label{JumpComment}
 \sup_{|p^{\prime}|>\frac{1}{\lambda} }\frac{ \int_{[-\frac{1}{\lambda}, \frac{1}{\lambda}  ]}dp^{\prime \prime } \frac{1}{1+|p^{\prime \prime} |^{v}}\mathcal{J}_{\lambda}(p^{\prime} , p^{\prime \prime}   ) }{  \int_{[-\frac{1}{\lambda},\frac{1}{\lambda}   ]}dp^{\prime \prime } \mathcal{J}_{\lambda}(p^{\prime} , p^{\prime \prime}   )  } = \mathcal{O}(\lambda^{v}).  
\end{align}
The order equality~(\ref{JumpComment}) holds because the conditional distribution for a momentum jump starting from a momentum $|p'|>\frac{1}{\lambda}$ and conditioned to jump to a value $|p''|\leq \frac{1}{\lambda}$ will be concentrated in the vicinity of the border $|p''|\approx \frac{1}{\lambda}$ where  $\frac{1}{1+|p^{\prime \prime} |^{v}}= \mathcal{O}(\lambda^{v})$.  This is a consequence of the exponential decay found in the form of the jump rates  $\mathcal{J}_{\lambda}(p,p')$.

\vspace{.25cm}
\noindent (i).\hspace{.25cm}  For arbitrary $s\in \Sigma$, we have 
\begin{align}\label{Zapatos}
\mathbb{E}^{(\lambda)}_{s}\Big[  \Big|\int_{0}^{\tau_{1} }dr \frac{dV}{dx}(X_{r})\Big|^{v} \Big] \leq  \sup_{x\in \mathbb{T} }\big|\frac{dV}{dx}(x)\big|^{v} \mathbb{E}\big[\tau_{1}^{v}\big]  \leq   v!\,\sup_{x }\big|\frac{dV}{dx}(x)\big|^{v},
\end{align}
since $\tau_{2}-\tau_{1}$ is a mean one exponential.

\vspace{.25cm}
\noindent (ii).\hspace{.25cm} Next we consider $s=(x,p)$ for the regime $1\ll |p|<\lambda^{-1}$.   As long as the momentum stays below $2\lambda^{-1}$ over the time interval $[0,\tau_{1}]$, the collisions will occur with Poisson rate smaller than $\mathcal{E}_{\lambda}(2\lambda^{-1})$, which is uniformly finite by~(\ref{EscapeComment}).  Thus, in that case, the expected number of collisions up to time $\tau_{1}$ is uniformly finite for $\lambda<1$, and as a consequence the momentum of the particle will not fluctuate significantly from its initial value $p$.  To show that $|p|$ typically stays well below $2\lambda^{-1}$, let us bound the probability of the event that $|P_{r}|\notin \big[\frac{1}{2}|p|,\frac{3}{2}|p|]$ for some $r\leq \tau_{1}$:  
\begin{align}\label{Trousers}\mathbb{P}_{s}^{(\lambda)}\Big[|P_{r}|\notin \big[\frac{1}{2}|p|,\frac{3}{2}|p|\big]  &\text{ for some } r\leq  \tau_{1}   \Big]\nonumber  \\   \leq  &\big(\frac{2}{|p|}\big)^{w}\mathbb{E}^{(\lambda)}_{s}\Big[ \sup_{0\leq r\leq \varsigma\wedge \tau_{1} }\big| P_{r}-p|^{w}      \Big] \nonumber  \\  \leq  &\big(\frac{4}{|p|}\big)^{w}\sup_{\substack{(x,p)\\  |p|\leq \lambda^{-1} }  }\mathbb{E}^{(\lambda)}_{(x,p)}\Big[ \sup_{0\leq r\leq \varsigma\wedge \tau_{1} }\big| J_{r}|^{w}      \Big] + w!\,\big(\frac{4}{|p|}\big)^{w}\sup_{x }\big|\frac{dV}{dx}(x)\big|^{w},
\end{align} 
where $w\geq 1$, $\varsigma$ is the first jump time such that $|P_{r}|$ leaves  $ \big[\frac{1}{2}|p|,\frac{3}{2}|p|\big]$, and $J_{r}=P_{r}-p+\int_{0}^{t}dr \frac{dV}{dx}(X_{r})$ is the sum of the momentum jumps up to time $r$.  The first inequality in~(\ref{Trousers}) is Jensen's, and  for the second inequality, we have used $(x+y)^{w}\leq 2^{w}(x^{w}+y^{w})$ and~(\ref{Zapatos}) to bound the contribution of the potential drift.  The probability densities of individual momentum jumps conditioned to jump from momentum $\hat{p}$, $d_{\hat{p}}(p')=\frac{ \mathcal{J}_{\lambda}(\hat{p},p')   }{ \int_{\Sigma}dp'' \mathcal{J}_{\lambda}(\hat{p},p'')  }    $,     have uniformly controlled Gaussian tails for $|\hat{p}|\leq 2\lambda^{-1}   $  and occur with Poisson rate $\mathcal{E}_{\lambda}(P_{r})\leq \mathcal{E}_{\lambda}(2\lambda^{-1})$ for $r\leq \varsigma $.   Thus the expectation of $\sup_{0\leq r\leq \varsigma \wedge \tau_{1} }\big| J_{r}|^{w}$ above is uniformly finite.  Since $w\geq 1$ is arbitrary, it follows that the probability on the first line of~(\ref{Trousers}) decays super-polynomially quickly for $|p|\gg 1$.

Now we bound $\mathbb{E}_{s }^{(\lambda)}\big[ \big|\int_{0}^{\tau_{1}}dr\,\frac{dV}{dx}(X_{r})  \big|^{v }  \big]$.
Define the times $t_{n}^{\prime}=t_{n}\wedge \tau_{1}\wedge \varsigma $, where $t_{n}$ is the time of the $n$th momentum jump.   By writing
$$\int_{0}^{\tau_{1} }dr\,\frac{dV}{dx}(X_{r}) = \sum_{n=0}^{\infty} \int_{t_{n}^{\prime} }^{t_{n+1}^{\prime} }dr\frac{dV}{dx}(X_{r})+\chi( \varsigma\leq \tau_{1})  \int_{\varsigma}^{\tau_{1}}dr\,\frac{dV}{dx}(X_{r}),    $$  
we can apply the triangle inequality  to get
\begin{align}\label{Mozey}
\mathbb{E}_{s^{\prime}}^{(\lambda)}\Big[ &  \Big|\int_{0}^{\tau_{1} }dr \frac{dV}{dx}(X_{r})  \Big|^{v} \Big]^{\frac{1}{v}} \nonumber \\ &\leq \mathbb{E}_{s^{\prime}}^{(\lambda)}\Big[  \Big(\sum_{n=0}^{\infty} \Big|\int_{t_{n}^{\prime} }^{t_{n+1}^{\prime} }dr \frac{dV}{dx}(X_{r}) \Big|\Big)^{v}\Big]^{\frac{1}{v}}+\mathbb{E}_{s^{\prime}}^{(\lambda)}\Big[   \chi( \varsigma\leq \tau_{1}) \Big| \int_{\varsigma}^{\tau_{1}}dr \frac{dV}{dx}(X_{r})  \Big|^{v} \Big]^{\frac{1}{v}} \nonumber  \\ &\leq \mathbb{E}_{s^{\prime}}^{(\lambda)}\Big[  \Big(\sum_{n=0}^{\infty} \Big|\int_{t_{n}^{\prime} }^{t_{n+1}^{\prime} }dr\frac{dV}{dx}(X_{r}) \Big|\Big)^{v}\Big]^{\frac{1}{v}}+\Big(   (2\,v)!\,\sup_{x\in \mathbb{T} }\big|\frac{dV}{dx}(x)\big|^{2v} \Big)^{\frac{1}{2v}}\mathbb{P}_{s^{\prime} }^{(\lambda)}\big[ \varsigma \leq \tau_{1}    \big]^{\frac{1}{2v}},
\end{align}
where the second inequality follows by Cauchy-Schwarz and because $\tau_{1}$ is a mean one exponential.  The probability $\mathbb{P}_{s^{\prime} }^{(\lambda)}\big[ \varsigma \leq \tau_{1}    \big]$ decays faster than any polynomial by~(\ref{Trousers}).  The first term on the right side of~(\ref{Mozey}) has the bound
\begin{align}\label{Biscut}
\mathbb{E}_{s^{\prime}}^{(\lambda)}\Big[  \Big(\sum_{n=0}^{\infty} \Big|\int_{t_{n}^{\prime} }^{t_{n+1}^{\prime} }dr \frac{dV}{dx}(X_{r}) \Big|\Big)^{v}\Big]\leq \Big(\frac{4\sup_{x\in \mathbb{T}  }V(x)}{|p|}\Big)^{v}\mathbb{E}_{s^{\prime}}^{(\lambda)}\big[\mathcal{N}_{\varsigma}^{v}  \big],
\end{align}
where $\mathcal{N}_{t}$ is the number of collisions up to time $t$.  The above inequality uses the definition of the $t_{n}^{\prime}$'s to conclude that for each $n$, either $t_{n}^{\prime}=t_{n+1}^{\prime}$ so that $\int_{t_{n}^{\prime} }^{t_{n+1}^{\prime} }dr\frac{dV}{dx}(X_{r})=0 $, or $ |P_{t_{n}^{\prime}}| \geq \frac{1}{2}|p|$ so that we can apply Part (1) of Lem.~\ref{DetDrift} to bound $\big|\int_{t_{n}^{\prime} }^{t_{n+1}^{\prime} }dr\frac{dV}{dx}(X_{r})\big|$.  The counting process $\mathcal{N}_{t}$  has Poisson rate $\mathcal{E}_{\lambda}(P_{t})$ at time $t$.  For times $t<\varsigma$, we have that  $\mathcal{E}_{\lambda}(P_{t})\leq   \sup_{\lambda<1} \mathcal{E}_{\lambda}(2\lambda^{-1}):=\mathbf{r}$ and 
$$\mathbb{E}_{s^{\prime}}^{(\lambda)}\big[\mathcal{N}_{\varsigma}^{v}  \big]\leq \mathbb{E}\big[(N^{\prime}_{\tau})^{v} \big]=\frac{1}{1+\mathbf{r}}\sum_{n=0}^{\infty}n^{v}\big( \frac{\mathbf{r}}{1+\mathbf{r} }  \big)^{n}<\infty ,$$
where $N_{t}^{\prime}$ is a Poisson process with rate $\mathbf{r}$ and the random variable $\tau$ is  mean one, exponentially distributed, and independent of $N_{t}^{\prime}$.  The first inequality can be seen by a construction $N^{\prime}_{\tau}\approx \mathcal{N}_{\varsigma}+\mathcal{N}_{\tau}^{\prime}$ for a jump process $\mathcal{N}^{\prime}_{r}$ with Poisson jump rate $ \mathbf{r}-\mathcal{E}_{\lambda}(P_{t})$ for $t\leq \varsigma$ and rate $\mathbf{r}$ for $t>\varsigma$ whose jumps are decided independently of the jumps of $\mathcal{N}_{r}$.  

\vspace{.25cm}

\noindent (iii).\hspace{.25cm} For the regime $|p|>\lambda^{-1}$, our analysis must treat the possiblity that many collisions occur over the time interval $[\tau_{1},\tau_{2}]$ (specificially, when $|p|\gg \lambda^{-1}$).   Let $\vartheta =\tau_{1}\wedge \vartheta^{\prime}    $ where $\vartheta^{\prime}$ is the hitting time that the absolute value of the momentum $|P_{t}|$ jumps below $\lambda^{-1}$.  The hitting time  $\vartheta^{\prime}  $ is finite, and, in fact, has an expectation that is bounded by a multiple of  $\lambda^{-1}$ independently of the initial momentum $|p|>\lambda^{-1}$.   However, the details for these points do not matter for this proof.  Let $\varphi_{s}$ be the distribution on $\mathbb{T}\times [-\lambda^{-1},\lambda^{-1}]$   for $(X_{\vartheta^{\prime}},P_{\vartheta^{\prime}}  )$ starting from $s\in \Sigma$.  By the triangle inequality and the strong Markov property
\begin{align}\label{Engine}
\mathbb{E}_{s}^{(\lambda)}\Big[\Big|\int_{0}^{\tau_{1}}dr \frac{dV}{dx}(X_{r})\Big|^{v}\Big]^{\frac{1}{v}}\leq  \mathbb{E}_{s}^{(\lambda)}\Big[\Big|\int_{0}^{\vartheta}dr\frac{dV}{dx}(X_{r})\Big|^{v}\Big]^{\frac{1}{v}} + \mathbb{E}_{\varphi_{s} }^{(\lambda)}\Big[\Big|\int_{0}^{\tau_{1}}dr \frac{dV}{dx}(X_{r})\Big|^{v}\Big]^{\frac{1}{v}} .
\end{align}

For the first term on the right side on~(\ref{Engine}), we can write 
\begin{align*}
\mathbb{E}_{s}^{(\lambda)}\Big[\Big|\int_{0}^{\vartheta}dr\frac{dV}{dx}(X_{r})\Big|^{v}\Big]= & \mathbb{E}_{s}^{(\lambda)}\Big[\Big|\sum_{n=1}^{\mathcal{N}_{\vartheta}}\int_{t_{n-1}}^{t_{n}}dr\,\frac{dV}{dx}(X_{r})\Big|^{v}\Big]\\  \leq & 2^{v}\big(\sup_{x}V(x)\big)^{v} \mathbb{E}_{s}^{(\lambda)}\Big[\Big|\sum_{n=1}^{\mathcal{N}_{\vartheta}}\frac{1}{|P_{t_{n}^{-}}|}  \Big|^{v}\Big]\leq  \frac{1}{c}\lambda^{v} \mathbb{E}_{s}^{(\lambda)}\Big[\Big|\sum_{n=1}^{\mathcal{N}_{\vartheta}}\frac{1}{\mathcal{E}_{\lambda}(P_{t_{n}^{-}}) }  \Big|^{v}\Big].
\end{align*}
The first inequality is Part (1) of Lem.~\ref{DetDrift}, which is applied for the Hamiltonian evolution on each interval $[t_{n-1},t_{n})$.   The second inequality  holds since there is a $c>0$ such that $ \mathcal{E}_{\lambda}(p)\leq c\lambda |p|$ for all $|p|\geq \lambda^{-1}$ as a consequence of  (\ref{EscapeComment}).   Notice that the difference $\sum_{n=1}^{\mathcal{N}_{r}}\frac{1}{\mathcal{E}_{\lambda}(P_{t_{n}^{-}}) } -r$ is a martingale with predictable quadratic variation $\int_{0}^{r}ds \frac{1}{\mathcal{E}_{\lambda}(P_{s}) }$ since the counting process $\mathcal{N}_{r}$ has jump rate $\mathcal{E}_{\lambda}(P_{r})$.   By the triangle inequality and the relation $\vartheta\leq \tau_{1}$,
\begin{align*}
\mathbb{E}_{s}^{(\lambda)}\Big[\Big|\sum_{n=1}^{\mathcal{N}_{\vartheta}}\frac{1}{\mathcal{E}_{\lambda}(P_{t_{n}^{-} }) }  \Big|^{v}\Big]^{\frac{1}{v}}&\leq \mathbb{E}_{s}^{(\lambda)}\Big[\Big|\sum_{n=1}^{\mathcal{N}_{\tau_{1} }}\frac{1}{\mathcal{E}_{\lambda}(P_{t_{n}^{-}}) } -\tau_{1} \Big|^{v}\Big]^{\frac{1}{v}}+ \mathbb{E}_{s}^{(\lambda)}\big[\tau_{1}^{v}\big]^{\frac{1}{v}} \\ &\leq C'\mathbb{E}_{s}^{(\lambda)}\Big[\Big|  \int_{0}^{\tau_{1} }ds\frac{1}{ \mathcal{E}_{\lambda}(P_{s}) } \Big|^{v} \Big]^{\frac{1}{v}}+\mathbb{E}_{s}^{(\lambda)}\Big[\sup_{1\leq n\leq \mathcal{N}_{\tau_{1}} } \frac{1}{\big|\mathcal{E}_{\lambda}(P_{t_{n}^{-}}) \big|^{v}  } \Big]^{\frac{1}{v}} +\mathbb{E}_{s}^{(\lambda)}\big[\tau_{1}^{v}\big]^{\frac{1}{v}}\\ &\leq  8(C'+1)(1+\lambda)(v!)^{\frac{1}{v}}+(v!)^{\frac{1}{v}},
\end{align*}
where the constant $C'$ arises from an application of Rosenthal's inequality, and the third inequality holds since  $\mathcal{E}_{\lambda}(p) \geq \frac{1}{8(1+\lambda)}$ by (\ref{EscapeComment}) and because the random variable $\tau_{1}$ is exponential with mean one.

For the second term on the right side of~(\ref{Engine}), we can apply our results (i) and (ii) above to guarantee the existence of a $C>0$ such that 
\begin{align}\label{Hausdorff}
 \mathbb{E}_{\varphi_{s} }^{(\lambda)}\Big[\Big|\int_{0}^{\tau_{1}}dr\frac{dV}{dx}(X_{r})\Big|^{v}\Big]\leq C\int_{\Sigma}d\varphi_{s}(x^{\prime},p^{\prime})\frac{1}{1+|p^{\prime} |^{v}}  \leq C\sup_{|p^{\prime}|>\frac{1}{\lambda} }\frac{ \int_{[-\frac{1}{\lambda}, \frac{1}{\lambda}  ]}dp^{\prime \prime } \frac{1}{1+|p^{\prime \prime} |^{v}}\mathcal{J}_{\lambda}(p^{\prime} , p^{\prime \prime}   ) }{  \int_{[-\frac{1}{\lambda},\frac{1}{\lambda}   ]}dp^{\prime \prime } \mathcal{J}_{\lambda}(p^{\prime} , p^{\prime \prime}   )  }  ,
 \end{align}
where the third expression should be understood as the supremum over all $|p^{\prime}|>\lambda^{-1}$ for the expectation of $\frac{1}{1+|P_{\tau}|^{v}  }$ conditioned on $p^{\prime}=P_{\tau^{-} } $.  The final term in~(\ref{Hausdorff}) is $\mathcal{O}(\lambda^{v})$ by  (\ref{JumpComment}).

\vspace{.5cm}

\noindent Part (2):\hspace{.25cm}   We now seek to take full advantage of the averaging that results from integrating $ \frac{dV}{dx}(X_{r}) $ between two random times $r\in [\tau_{1},\tau_{2}]$.     If only the upper limit of  integration  were random, such as for the expression $ \big| \mathbb{E}_{s }^{(\lambda)}\big[ \int_{0}^{\tau_{1}}dr\frac{dV}{dx}(X_{r})   \big] \big|$, then we would only have an upper bound  proportional to $\textup{max}\big(\frac{1}{1+|p| }, \lambda \big)$.   The bound for $ \mathbf{C}^{(\lambda)}_{0}(s) $ in the  region $|p|\geq \lambda^{-1}$ follows from  Part (1), so we will focus our analysis on the regime $1 \ll  | p| < \lambda^{-1}$.    We will proceed by approximating the quantity $ \mathbf{C}^{(\lambda)}_{0}(s) $    by expressions that are progressively easier to analyze.

By~(\ref{Zapatos}), we have an uniform  upper bound for $\sup_{s}|\mathbf{C}_{0}^{(\lambda)}(s)|$.  The difference between  $\mathbf{C}^{(\lambda)}_{0}(s)$ and $\mathbb{E}_{s}^{(\lambda)}\big[\int_{\tau_{1}}^{\tau_{2}}dr\frac{dV}{dx}(X_{r})\big]$ is small when $|p|\gg 1$ since
\begin{align}\label{Shutdown}
\Big| \mathbb{E}_{s}^{(\lambda)}\Big[&\int_{\tau_{1}}^{\tau_{2}}dr\frac{dV}{dx}(X_{r})\Big]-\mathbf{C}^{(\lambda)}_{0}(s)    \Big| =\Big| \tilde{\mathbb{E}}_{\tilde{\delta}_{s} }^{(\lambda)}\Big[\int_{\tau_{1}}^{\tau_{2}}dr\frac{dV}{dx}(X_{r})\Big]-\mathbf{C}^{(\lambda)}_{0}(s)    \Big|\nonumber  \\
&=\Big|\tilde{\mathbb{E}}_{\tilde{\delta}_{s} }^{(\lambda)}\Big[\chi( z_{\tau_{1}}=1   )\int_{\tau_{1} }^{\tau_{2}}dr\frac{dV}{dx}(X_{r})\Big] \Big|    \leq \tilde{\mathbb{E}}_{\tilde{\delta}_{s} }\Big[\Big|\int_{\tau_{1} }^{\tau_{2}}dr\frac{dV}{dx}(X_{r})\Big|^{2}\Big]^{\frac{1}{2}}\tilde{\mathbb{P}}_{\tilde{\delta}_{s} }^{(\lambda)}\big[  z_{\tau_{1}}=1  \big]^{\frac{1}{2}} \nonumber \\ &= \mathbb{E}_{s }\Big[\Big|\int_{\tau_{1} }^{\tau_{2}}dr\frac{dV}{dx}(X_{r})\Big|^{2}\Big]^{\frac{1}{2}}\mathbb{E}_{s }^{(\lambda)}\big[ h(S_{\tau_{1}}) \big]^{\frac{1}{2}} \leq 2^{\frac{1}{2}}\sup_{x}\big|\frac{dV}{dx}(x)\big|^{\frac{1}{2}}\mathbb{E}_{s }^{(\lambda)}\big[ h(S_{\tau_{1}}) \big]^{\frac{1}{2}} .
\end{align}
   The first and third equalities  use that $ \mathbb{E}_{s}^{(\lambda)}=\tilde{\mathbb{E}}_{\tilde{\delta}_{s} }^{(\lambda)}$, and the identity  $\tilde{\mathbb{P}}_{\tilde{\delta}_{s} }^{(\lambda)}\big[  z_{\tau_{1}}=1  \big]= \mathbb{E}_{s }^{(\lambda)}\big[ h(S_{\tau_{1}}) \big]  $ used in the third equality can be shown using Part (3) of Prop.~\ref{BasicsOfNum}.   The first inequality is Cauchy-Schwarz, and the second inequality uses that $\tau_{2}-\tau_{1}$ is a mean one exponential. 
  The function $h(s)\leq 1$ has compact support, and  there is a $c>0$ such that $\mathbb{E}_{(x,p) }^{(\lambda)}\big[ h(S_{\tau_{1}}) \big]\leq ce^{-\lambda^{-1}}\vee e^{-|p|}  $.  In fact, the bound can be given a Gaussian form as a consequence of the Gaussian tails found in the jump rates~(\ref{JumpRates}).   It follows that the difference  of  $\mathbf{C}^{(\lambda)}_{0}(s)$ and $\mathbb{E}_{s}^{(\lambda)}\big[\int_{\tau_{1}}^{\tau_{2}}dr\frac{dV}{dx}(X_{r})\big]$ is negligible for our purpose.

By the above remarks, we may work with $\mathbb{E}_{s}^{(\lambda)}\big[\int_{\tau_{1}}^{\tau_{2}}dr\frac{dV}{dx}(X_{r})\big]$.   Now we will  show that the difference of this term with the expression $\frac{1}{p}\mathbb{E}_{s}^{(\lambda)}\big[ V(X_{\tau_{2}})-V(X_{\tau_{1}})  \big]$ is  $\mathit{O}(p^{-2})$.     For $p$ in the regime $1\ll |p|\leq \lambda^{-1}$,  define $\varsigma$ as in Part (1) and   define $t_{n}$ as the sequence of collision times starting after $\tau_{1}$ with $t_{0}=\tau_{1}$,  and  $t_{n}^{\prime}=t_{n}\wedge \varsigma\wedge \tau_{2}  $.   Similarly to Part (1), we can write 
$$\int_{\tau_{1}}^{\tau_{2} }dr\frac{dV}{dx}(X_{r}) = \sum_{n=0}^{\infty} \int_{t_{n}^{\prime} }^{t_{n+1}^{\prime} }dr\frac{dV}{dx}(X_{r})+\chi( \varsigma\leq \tau_{2})  \int_{\varsigma \vee \tau_{1} }^{\tau_{2}}dr\frac{dV}{dx}(X_{r}). $$
The difference between the expressions is bounded by
\begin{align}
\Big| \mathbb{E}_{s}^{(\lambda)}\Big[\int_{\tau_{1}}^{\tau_{2} }dr\frac{dV}{dx}(X_{r})\Big]&-\frac{1}{p}\mathbb{E}_{s}^{(\lambda)}\big[ V(X_{\tau_{2}})-V(X_{\tau_{1}})  \big]   \Big|\leq  \Big(\frac{\sup_{x}V(x)}{|p|}     +\sup_{x}\big|\frac{dV}{dx}(x)\big| \Big)\mathbb{P}_{s}\big[ \varsigma\leq \tau_{2}   \big] \nonumber \\ &+ \mathbb{E}_{s}^{(\lambda)}\Big[ \sum_{n=0}^{\mathcal{N}_{\varsigma}-1 }\Big| \int_{t_{n}^{\prime} }^{t_{n+1}^{\prime} }dr\frac{dV}{dx}(X_{r})-\frac{V( X_{t_{n+1}^{\prime}} )-V(X_{t_{n}^{\prime} }  )    }{P_{t_{n}^{\prime}}  }     \Big|   \Big] \nonumber \\ &+      \mathbb{E}_{s}^{(\lambda)}\Big[ \sum_{n=0}^{\mathcal{N}_{\varsigma} -1}\big|V( X_{t_{n+1}^{\prime}} )-V(X_{t_{n}^{\prime} }  )\big|\,\Big|  \frac{  1  }{P_{t_{n}^{\prime}}  } -\frac{1}{p}\Big|   \Big], \label{Hammer}
\end{align}
where $\mathcal{N}_{r}$, $r\geq \tau_{1}$ is the number of collision times $t_{n}$ in the interval $(\tau_{1},r]$, and the first term on the right side  bounds the expectations of $\frac{V(X_{\tau_{2}})-V(X_{\varsigma})}{p}$ and $ \chi( \varsigma\leq \tau_{2})  \int_{\varsigma \vee \tau_{1} }^{\tau_{2}}dr\frac{dV}{dx}(X_{r})$.   The inequality~(\ref{Hammer}) follows by adding and subtracting terms $\frac{V( X_{t_{n+1}^{\prime}} )-V(X_{t_{n}^{\prime} }  )    }{P_{t_{n}^{\prime}}  }$ for $n\in [0,\mathcal{N}_{\varsigma})$ and applying the triangle inequality.   By the same analysis as in (\ref{Trousers}), $ \mathbb{P}_{s}\big[ \varsigma\leq \tau_{2}   \big]$ decays super-polynomially with  $|p|\gg 1$.   We will bound the second and third lines of~(\ref{Hammer}) below.  

The second line of~(\ref{Hammer}) has the bound
\begin{align*}
 \mathbb{E}_{s}^{(\lambda)}\Big[ \sum_{n=0}^{\mathcal{N}_{\varsigma}-1 }\Big| \int_{t_{n}^{\prime} }^{t_{n+1}^{\prime} }dr\frac{dV}{dx}(X_{r})-\frac{V( X_{t_{n+1}^{\prime}} )-V(X_{t_{n}^{\prime} }  )    }{P_{t_{n}^{\prime}}  }     \Big|   \Big]  \leq & 2\sup_{x}\big|\frac{dV}{dx}(x)\big| \sup_{x}V(x)\, \mathbb{E}_{s}^{(\lambda)}\Big[ \sum_{n=0}^{\mathcal{N}_{\varsigma}-1 }\frac{ t_{n+1}^{\prime}- t_{n}^{\prime} }{ |P_{\tau_{n}}|^{2} }  \Big]\\ \leq & \frac{8}{|p|^{2}}\sup_{x}\big|\frac{dV}{dx}(x)\big| \sup_{x}V(x),
\end{align*}
where the second inequality uses that $|P_{\tau_{n}^{\prime} }|\geq \frac{1}{2}|p|$, by definition, for $n\leq \mathcal{N}_{\varsigma}$, and also uses that 
 $$\sum_{n=0}^{\mathcal{N}_{\varsigma} }t_{n+1}^{\prime}-t_{n}^{\prime}  =\varsigma-\tau_{1}\leq \tau_{2}-\tau_{1}\quad \text{and hence }\quad \mathbb{E}_{s}^{(\lambda)}\Big[ \sum_{n=0}^{\mathcal{N}_{\varsigma} }t_{n+1}^{\prime}-t_{n}^{\prime}\Big]\leq \mathbb{E}_{s}^{(\lambda)}\big[ \tau_{2}-\tau_{1} \big] =1.  $$

For the third line on the right side of~(\ref{Hammer}),
 \begin{align*}
   \mathbb{E}_{s}^{(\lambda)}\Big[ &\sum_{n=0}^{\mathcal{N}_{\varsigma}-1 }\big|V( X_{t_{n+1}^{\prime}} )-V(X_{t_{n}^{\prime} }  )\big|\,\Big|  \frac{  1  }{P_{t_{n}^{\prime}}  } -\frac{1}{p}\Big|   \Big] \\ &\leq \sup_{x}V(x) \,\mathbb{E}_{s}^{(\lambda)}\Big[ \sum_{n=0}^{\mathcal{N}_{\varsigma}-1}\frac{\big| p-P_{t_{n}^{\prime}}   |}{|pP_{t_{n}^{\prime}}|     }   \Big]   \leq  \frac{2}{|p|^{2}}\big( \sup_{x}V(x)\big) \mathbb{E}_{s}^{(\lambda)}\Big[ \sum_{n=0}^{\mathcal{N}_{\varsigma}-1}\big| p-P_{t_{n}^{\prime}} |  \Big] \\ &\leq  \frac{2}{|p|^{2}}\big( \sup_{x}V(x)\big)\, \mathbb{E}_{s}^{(\lambda)}\Big[ \sup_{0\leq r\leq \varsigma}\big| P_{r}-p\big|^{2}\Big]^{\frac{1}{2}} \mathbb{E}_{s}^{(\lambda)}\big[\mathcal{N}_{\varsigma}^{2}\big]^{\frac{1}{2}} =\mathit{O}(|p|^{-2}).
   \end{align*}
 The second inequality uses the definition of $\varsigma$ to conclude that $\frac{1}{2}|p|\leq |P_{t_{n}^{\prime}}|$ for $n\leq \mathcal{N}_{\varsigma}$, and the third inequality is Cauchy-Schwarz.     Arbitrary moments of $\mathcal{N}_{\varsigma}$ are finite by~(\ref{Biscut}).      

Our final task is to bounding the expression $\frac{1}{p}\mathbb{E}_{s}^{(\lambda)}\big[V(X_{\tau_{2}})-V(X_{\tau_{1}})    \big]$, and we only need to show that $\big|\mathbb{E}_{s}^{(\lambda)}\big[V(X_{\tau_{2}})-V(X_{\tau_{1}})    \big]\big|=\mathit{O}(|p|^{-1})$ for $|p|\gg 1$.  By the triangle inequality, 
$$\big|\mathbb{E}^{(\lambda)}_{s}\big[V(X_{\tau_{2}})-V(X_{\tau_{1}})    \big]\big|\leq  \Big|\mathbb{E}^{(\lambda)}_{s}\big[V(X_{\tau_{2}})\big]-\int_{\mathbb{T}}dxV(x)\big]\Big|+\Big|\mathbb{E}^{(\lambda)}_{s}\big[V(X_{\tau_{1}})\big]-\int_{\mathbb{T}}dxV(x)\big]\Big|. $$
The terms on the right side are similar, so we will study the second.  Bounding the difference between $\mathbb{E}^{(\lambda)}_{s}\big[V(X_{\tau_{1}})\big]$  and  $\int_{\mathbb{T}}dx V(x)$ is very close in spirit to Part (3) of Lem.~\ref{DetDrift} except that we now must treat an Hamiltonian evolution perturbed by some random  momentum kicks.  As in Part (1), the assumption that $|p|\leq \lambda^{-1}$ ensures that not many momentum kicks are likely to occur.

 We can reconstruct the counting process $\mathcal{N}_{t}$ for the number of collisions up to time $\varsigma$ as follows.  Let $N^{\prime}$ be a Poisson clock with rate $\mathbf{r}=\mathcal{E}_{\lambda}(2\lambda^{-1})$ as in Part (1).  The Poisson rate of jumps $\mathcal{E}_{\lambda}(P_{t})$ for the process $\mathcal{N}_{t}$  satisfies $\mathcal{E}_{\lambda}(P_{t})\leq \mathbf{r}$ for times $t\leq \varsigma$.  At each jump time $r_{n}\leq \varsigma $ for the Poisson process $N^{\prime}$, we then flip  an independent coin with weight $\mathbf{r}^{-1} \mathcal{E}_{\lambda}(P_{r_{n}})$ to determine if a jump for $\mathcal{N}_{t}$ (i.e. a collision) occurred at time $r_{n}$.  This construction recovers the statistics for $\mathcal{N}_{t}$.  We then define $r_{n}^{\prime}=r_{n}\wedge \tau_{1}$ for $n\leq N^{\prime}_{\varsigma\wedge \tau_{1}}$.  Conditioned on the past $\mathcal{F}_{r^{\prime}_{n}}$ and  the event $\tau_{1}>r_{n}^{\prime}$,  the increment $r_{n+1}^{\prime}-r_{n}^{\prime}$ is exponentially distributed with mean $(1+\mathbf{r})^{-1}$.  When conditioned on the event $\tau_{1}=r_{n+1}^{\prime}$, the increment $r_{n+1}^{\prime}-r_{n}^{\prime}$  is exponential with mean $1$.  

We can rewrite the expectation $\mathbb{E}^{(\lambda)}_{s}\big[V(X_{\tau_{1}})\big]$ in terms of the complimentary  events $\tau_{1}>\max_{n}\,r_{n}^{\prime}$  and $\tau_{1}=r_{n}^{\prime}$ for some $n$ as follows:
\begin{align}\mathbb{E}^{(\lambda)}_{s}\big[V(X_{\tau_{1}})\big]= &\mathbb{E}^{(\lambda)}_{s}\big[ V(X_{\tau_{1}})\chi(\tau_{1}>\max_{n}\,r_{n}^{\prime}) \big]\nonumber  \\   &+ \mathbb{E}^{(\lambda)}_{s}\Big[ \sum_{n=0}^{\infty}\chi(\tau_{1}=r_{n+1}^{\prime} )\, \mathbb{E}^{(\lambda)}_{s}\big[V(X_{\tau_{1}})\,\big|\,\mathcal{F}_{r^{\prime}_{n}},\,\tau_{1}=r^{\prime}_{n+1}  \big] \Big].
\end{align}
The first term on the right is smaller than $\sup_{x}V(x)$ times the probability of the event $\max_{n}r_{n}^{\prime}\neq \tau_{1}$, which can also be phrased as the event that $\varsigma<\tau$.  By the analysis in Part (1), $\mathbb{P}_{s}^{(\lambda)}[\varsigma<\tau]$ is super-polynomially small in $|p|\gg 1$.  Thus $\sum_{n=0}^{\infty}\mathbb{P}^{(\lambda)}_{s}[\tau_{1}=r_{n+1}^{\prime} ]$ is super-polynomially close to $1$. Since $r_{n+1}^{\prime}-r_{n}^{\prime}$ is exponentially distributed,  by Part (3) of Lem.~\ref{DetDrift} we have that
\begin{align*}
\Big|\mathbb{E}^{(\lambda)}_{s}\big[V(X_{\tau_{1}})\,\big|\,\mathcal{F}_{r^{\prime}_{n}},\,\tau_{1}=r^{\prime}_{n+1}  \big]-\int_{\mathbb{T}}dx V(x)\Big|\leq  (\text{const})|P_{r_{n}^{\prime}}|^{-1}\leq 2(\text{const})|p|^{-1}.
\end{align*}
The second inequality above follows since $|P_{r_{n}^{\prime}}|\geq \frac{1}{2}|p|$ by the definition of the times $r_{n}^{\prime}$, which are less than $\varsigma$.

\end{proof}

\subsection{Proof of Prop.~\ref{FurtherNum}}

\begin{proof}[Proof of Prop.~\ref{FurtherNum}]\text{   }\\
\noindent Part (1): \hspace{.1cm}
Recall that $\sigma_{n}=S_{\tau_{n}}$ and that $ \mathbf{N}_{t}$ is defined as the number partition times $\tau_{n}$, $n\geq 1$ to have occurred up to time $t$.    For $0\leq t< R_{1}$ we can write $\int_{0}^{t}dr\frac{dV}{dx}(X_{r})$ as
\begin{align}\label{George}
\int_{0}^{t}dr\frac{dV}{dx}(X_{r})= &\int_{0}^{\tau_{1}} dr\frac{dV}{dx}(X_{r}) - \chi(\zeta_{\mathbf{N}_{t}}=0)\int_{t}^{\tau_{\mathbf{N}_{t}+1} }    dr\frac{dV}{dx}(X_{r})\nonumber \\ &- \chi(\zeta_{\mathbf{N}_{t}+1}=0)\int_{\tau_{\mathbf{N}_{t}+1}}^{\tau_{\mathbf{N}_{t}+2} }    dr\frac{dV}{dx}(X_{r})   +\sum_{n=0}^{ \mathbf{N}_{t}} \mathbf{C}^{(\lambda)}_{0}(\sigma_{n})+\mathbf{m}_{t}+\mathbf{m}_{t}',
\end{align}
where $\mathbf{m}_{t}$ and $\mathbf{m}_{t}'$ correspond to odd and even contributions of the form
\begin{eqnarray*}
\mathbf{m}_{t}& := & \sum_{n=1}^{\lfloor \frac{1}{2} \mathbf{N}_{t}-\frac{1}{2}\rfloor+1}\Big( \chi(\zeta_{2n}=0) \int_{\tau_{2n}}^{\tau_{2n+1}} dr\frac{dV}{dx}(X_{r})- \mathbf{C}^{(\lambda)}_{0}(\sigma_{2n-1})      \Big), \\ \mathbf{m}_{t}'&:=& \sum_{n=0}^{\lfloor \frac{1}{2}\mathbf{N}_{t}\rfloor}\Big( \chi(\zeta_{2n+1}=0) \int_{\tau_{2n+1}}^{\tau_{2n+2}} dr\frac{dV}{dx}(X_{r})- \mathbf{C}^{(\lambda)}_{0}(\sigma_{2n})      \Big).
\end{eqnarray*}
The processes $\mathbf{m}_{t},\mathbf{m}_{t}'$  are not adapted to $ \tilde{\mathcal{F}}_{t}$ since, for instance, $\mathbf{m}_{t}'$ makes the jump $$\chi(\zeta_{2n+1}=0) \int_{\tau_{2n+1}}^{\tau_{2n+2}} dr\frac{dV}{dx}(X_{r})- \mathbf{C}^{(\lambda)}_{0}(\sigma_{2n}) $$  at time $\tau_{2n}$, and the size of the jump depends on $X_{t}$ up to time $\tau_{2n+2}$.   Let $\tilde{\mathcal{F}}_{t}''$ be the $\sigma$-algebra of all information before time $\tau_{n+2}$, where $\tau_{n}\leq t<\tau_{n+1}$, plus knowledge of the time $\tau_{n+2}$.  The process $\mathbf{m}_{t}'$ is a  martingale with respect to   $\tilde{\mathcal{F}}_{t}''$.  To see this  let us consider a time $t<\tau_{2n-1}$,  then the following equalities hold:
\begin{align}
  \tilde{\mathbb{E}}\Big[  \chi(\zeta_{2n+1}=0) \int_{\tau_{2n+1}}^{\tau_{2n+2}} dr\frac{dV}{dx}(X_{r})\,\Big|\,\tilde{\mathcal{F}}_{t}''  \Big] &=   \tilde{\mathbb{E}}\Big[  \tilde{\mathbb{E}}\Big[   \chi(\zeta_{2n+1}=0) \int_{\tau_{2n+1}}^{\tau_{2n+2}} dr\frac{dV}{dx}(X_{r})\,\Big|\,\tilde{\mathcal{F}}_{\tau_{2n}^{-}}\Big]\,\Big|\,\tilde{\mathcal{F}}_{t}''  \Big]  \nonumber \\ &= \tilde{\mathbb{E}}\Big[ \tilde{\mathbb{E}}_{\tilde{\delta}_{\sigma_{2n}} }\Big[   \chi(\zeta_{2n+1}=0) \int_{\tau_{2n+1}}^{\tau_{2n+2}} dr\frac{dV}{dx}(X_{r})\Big]\,\Big|\,\tilde{\mathcal{F}}_{t}''  \Big] \nonumber  \\ &=   \tilde{\mathbb{E}}\Big[ \mathbf{C}^{(\lambda)}_{0}(\sigma_{2n})\,\Big|\,\tilde{\mathcal{F}}_{t}''  \Big] \label{Hacks}.
\end{align}
The nested conditional expectation on the first line uses that $\tilde{\mathcal{F}}_{t}''\subseteq \tilde{\mathcal{F}}_{\tau_{2n}^{-}}$, and the third equality is by the definition of $\mathbf{C}^{(\lambda)}_{0}$.  The second equality applies  Part 3 of Prop.~\ref{BasicsOfNum} and  the strong Markov property starting from the time $\tau_{2n}$; recall that $\tilde{\delta}_{s}$ is the splitting of the $\delta$-distribution at $s\in \Sigma $.  The predictable quadratic variation $\langle \mathbf{m}_{t}' \rangle $  for the martingale $\mathbf{m}_{t}'$ with respect to $\tilde{\mathcal{F}}_{t}''$ has the form
\begin{align}\label{FlapJack}
\langle \mathbf{m}' \rangle_{t} =\sum_{n=1}^{\lfloor \frac{1}{2}\mathbf{N}_{t}\rfloor}\mathbf{C}^{(\lambda)}_{1}(\sigma_{2n}).
\end{align}
The analogous statements hold for  $\mathbf{m}_{t}$.

By the triangle inequality for~(\ref{George}),
\begin{multline}\label{ThisCanNotStandMan}
\tilde{\mathbb{E}}_{\tilde{\nu}}^{(\lambda)}\Big[\sup_{0\leq t\leq R_{1}}\Big|\int_{0}^{t}dr\frac{dV}{dx}(X_{r})  \Big|^{2m}   \Big]^{\frac{1}{2m}} \leq 6\sup_{x}\big|\frac{dV}{dx}(x)\big|\tilde{\mathbb{E}}_{\tilde{\nu}}^{(\lambda)}\big[\sup_{0\leq t< R_{1}}(\tau_{\mathbf{N}_{t}+1}-\tau_{\mathbf{N}_{t}} )^{2m}   \big]^{\frac{1}{2m}}\\+\tilde{\mathbb{E}}_{\tilde{\nu}}^{(\lambda)}\Big[\sup_{0\leq t< R_{1}}\Big(\sum_{n=0}^{ \mathbf{N}_{t}}| \mathbf{C}^{(\lambda)}_{0}(\sigma_{n})|  \Big)^{2m}   \Big]^{\frac{1}{2m}} + \tilde{\mathbb{E}}_{\tilde{\nu}}^{(\lambda)}\big[\sup_{0\leq t< R_{1}}\big|\mathbf{m}_{t}\big|^{2m}   \big]^{\frac{1}{2m}}+\tilde{\mathbb{E}}_{\tilde{\nu}}^{(\lambda)}\big[\sup_{0\leq t< R_{1}}\big|\mathbf{m}_{t}'\big|^{2m}   \big]^{\frac{1}{2m}}  ,
\end{multline}
where we have bounded each of the first three terms on the right side of (\ref{George}) by the supremum of $|\frac{dV}{dx}(x)|$ multiplied by the longest interval $\tau_{n+1}-\tau_{n}$ for $n\leq \tilde{n}_{1}$.  We used a factor of $6$ rather than $3$ since there is one term for which the interval $[\tau_{2n+1},\tau_{2n+2}]$ will have $\tau_{2n+1}\in [R_{1},R_{2}]$ rather than $< R_{1}$.  We thus double the bound since we can apply the strong Markov property   starting  from time $R_{1}$.  We now look at the terms on the right side one-by-one.

For the first term on the right side of~(\ref{ThisCanNotStandMan}),
\begin{align}\label{TheExponentials}
\tilde{\mathbb{E}}_{\tilde{\nu}}^{(\lambda)}\big[\sup_{0\leq t< R_{1}}(\tau_{\mathbf{N}_{t}+1}-\tau_{\mathbf{N}_{t}} )^{2m}   \big] &=\tilde{\mathbb{E}}_{\tilde{\nu}}^{(\lambda)}\big[\sup_{0\leq n\leq \tilde{n}_{1}}(\tau_{n+1}-\tau_{n})^{2m}   \big]\nonumber   \\ &\leq c^{2m}\tilde{\mathbb{E}}_{\tilde{\nu}}^{(\lambda)}\Big[ \mathbb{E}\big[\sup_{0\leq n\leq \tilde{n}_{1}}\mathbf{e}_{n}^{2m}\,\big|\,\tilde{n}_{1} \big]  \Big]\nonumber \\ &\leq  c'+c'\tilde{\mathbb{E}}_{\tilde{\nu}}^{(\lambda)}\big[\big(\log(1+\tilde{n}_{1}) \big)^{2m} \big],
\end{align}
where  $\mathbf{e}_{n}$ are i.i.d. mean one exponential random variables independent of everything else.  The  $c>0$ in the first inequality is from Lem.~\ref{LemPartTime} and replacing the $(\tau_{n+1}-\tau_{n})$'s with the $\mathbf{e}_{n}$'s.  The $c'>0$ for the second inequality exists by an elementary analysis of $\tilde{\mathbb{E}}^{(\lambda)}\big[\sup_{0\leq n\leq N}\mathbf{e}_{n}^{2m}\big]  $ for $m>0$ and independent exponential random variables $\mathbf{e}_{n}$ with mean one.  The value $\tilde{\mathbb{E}}_{\tilde{\nu}}^{(\lambda)}\big[\big(\log(1+\tilde{n}_{1}) \big)^{2m} \big]$ is finite by Prop.~\ref{FracMoment} since the fractional moments $\tilde{\mathbb{E}}_{\tilde{\nu}}^{(\lambda)}\big[\tilde{n}_{1}^{\alpha}  \big]$ are finite for $0<\alpha<\frac{1}{2}$.

For the second term on the right side of~(\ref{ThisCanNotStandMan}), obviously
$$\tilde{\mathbb{E}}_{\tilde{\nu}}^{(\lambda)}\Big[\sup_{0\leq t< R_{1}}\Big(\sum_{n=0}^{ \mathbf{N}_{t}}| \mathbf{C}^{(\lambda)}_{0}(\sigma_{n})|  \Big)^{2m}   \Big] =  \tilde{\mathbb{E}}_{\tilde{\nu}}^{(\lambda)}\Big[\Big(\sum_{n=0}^{ \tilde{n}_{1} }| \mathbf{C}^{(\lambda)}_{0}(\sigma_{n})|  \Big)^{2m}   \Big]
   $$
since $\mathbf{N}_{t}=\tilde{n}_{1}$ for $t \in[R_{n-1},R_{n})$.    By Lem.~\ref{TheCs} $g_{\lambda}= \mathbf{C}^{(\lambda)}_{0}$ has the inequality $|g_{\lambda}(x,p)|\leq C\textup{max}(\frac{1}{1+p^{2}}, \lambda)$ for $\lambda<1$. Hence, by Lem.~\ref{Marsupial}, the above sum is bounded independently of $\lambda<1$.

The last two terms on the right side of~(\ref{ThisCanNotStandMan}) are similar, so we will only treat the last.   By Rosenthal's inequality there is a $C''$ such that    
\begin{align*}
\tilde{\mathbb{E}}_{\tilde{\nu}}^{(\lambda)}&\big[\sup_{0\leq t< R_{1}}\big|\mathbf{m}_{t}' \big|^{2m } \big] \\ &\leq C''\tilde{\mathbb{E}}_{\tilde{\nu}}^{(\lambda)}\big[\langle \mathbf{m}'\rangle_{R_{1}}^{m}  \big] +  C''\tilde{\mathbb{E}}_{\tilde{\nu}}^{(\lambda)}\Big[\sup_{0\leq n\leq \lfloor \frac{ \tilde{n}_{1}}{2}\rfloor }\Big| \chi(\zeta_{2n+1}=0) \int_{\tau_{2n+1}}^{\tau_{2n+2}} dr\frac{dV}{dx}(X_{r})- \mathbf{C}^{(\lambda)}_{0}(\sigma_{2n})   \Big|^{2m}\Big] \\ &\leq C''\tilde{\mathbb{E}}_{\tilde{\nu}}^{(\lambda)}\big[\langle \mathbf{m}'\rangle_{R_{1}}^{m}  \big] +  C''\tilde{\mathbb{E}}_{\tilde{\nu}}^{(\lambda)}\Big[\sum_{n=0}^{\tilde{n}_{1} }\Big| \chi(\zeta_{n+1}=0) \int_{\tau_{n+1}}^{\tau_{n+2}} dr\frac{dV}{dx}(X_{r})- \mathbf{C}^{(\lambda)}_{0}(\sigma_{n})   \Big|^{4m}\Big]^{\frac{1}{2}} \\
&\leq C''\tilde{\mathbb{E}}_{\tilde{\nu}}^{(\lambda)}\Big[\Big(\sum_{n=0}^{\tilde{n}_{1}+1 }\mathbf{C}^{(\lambda)}_{1}(\sigma_{n}) \Big)^{m}\Big]+C''\tilde{\mathbb{E}}_{\tilde{\nu}}^{(\lambda)}\Big[\sum_{n=0}^{\tilde{n}_{1} }\mathbf{C}^{(\lambda)}_{2m}(\sigma_{n})   \Big]^{\frac{1}{2}} .
\end{align*}
For  the second inequality, we have used the standard technique to bound the supremum in the second term by using $\big(\sup_{n}a_{n}\big)^{2}\leq \sum_{n}a_{n}^{2}$ and Jensen's inequality, and we  have also included the odd-numbered terms.   The first term in the third inequality is bounded with the equality~(\ref{FlapJack}) and by including the odd-numbered terms.  The second term in the third equality is bounded by inserting a nested conditional expectation with respect to $\tilde{\mathcal{F}}_{\tau_{n}^{-}}$ for the $n$th term in the sum and applying the argument in~(\ref{Hacks}).    To bound both terms on the right side, we can apply Lem.~\ref{TheCs} and Lem.~\ref{Marsupial} as above.

\vspace{.5cm}
\noindent Part (2): \hspace{.1cm} Similarly to  Part (1), we begin by writing $\int_{0}^{R_{1}}dr\frac{dV}{dx}(X_{r})$ as in~(\ref{George}) and using the triangle inequality to get
\begin{multline}\label{Orielly}
\tilde{\mathbb{E}}_{\tilde{s}}^{(\lambda)}\Big[ \Big|\int_{0}^{R_{1}}dr\frac{dV}{dx}(X_{r})\Big|\Big]\leq   \tilde{\mathbb{E}}_{\tilde{s}}^{(\lambda)}\Big[\Big|   \int_{0}^{\tau_{1}} dr\frac{dV}{dx}(X_{r}) \Big|\Big]+\tilde{\mathbb{E}}_{\tilde{s}}^{(\lambda)}\Big[\Big|   \int_{t}^{\tau_{\tilde{n}_{1}+1} } dr\frac{dV}{dx}(X_{r}) \Big|\Big] \\+\tilde{\mathbb{E}}_{\tilde{s}}^{(\lambda)}\Big[\Big|   \int_{\tau_{\tilde{n}_{1 }+1}}^{\tau_{\tilde{n}_{1}+2} } dr\frac{dV}{dx}(X_{r}) \Big|\Big]+\tilde{\mathbb{E}}_{\tilde{s}}^{(\lambda)}\Big[    \sum_{n=0}^{ \tilde{n}_{1} }\big|\mathbf{C}^{(\lambda)}_{0}(\sigma_{n}) \big| \Big] + \tilde{\mathbb{E}}_{\tilde{s}}^{(\lambda)}\big[ \big|   \mathbf{m}_{R_{1}}\big| ^{2} \big]^{\frac{1}{2}} + \tilde{\mathbb{E}}_{\tilde{s}}^{(\lambda)}\big[ \big|  \mathbf{m}_{R_{1}}'\big|^{2}   \big]^{\frac{1}{2}},
\end{multline}
where we have also applied Jensen's inequality to the last two terms.  The first three terms on the right side and $\tilde{\mathbb{E}}_{\tilde{s}}^{(\lambda)}\big[    |\mathbf{C}^{(\lambda)}_{0}(\sigma_{0})| \big]=|\mathbf{C}^{(\lambda)}_{0}(s)|$ (the $n=0$ summand from the fourth term on the right) are bounded by $c\sup_{x}|\frac{dV}{dx}(x)|$, where $c>0$ is from Lem.~\ref{LemPartTime}.  This follows for the first term, for instance, since
$$\tilde{\mathbb{E}}_{\tilde{s}}^{(\lambda)}\Big[\Big|   \int_{0}^{\tau_{1}} dr\frac{dV}{dx}(X_{r}) \Big|\Big]\leq \sup_{x}\big|\frac{dV}{dx}(x)\big| \tilde{\mathbb{E}}_{\tilde{s}}^{(\lambda)}\big[ \tau_{1} \big] \leq c \sup_{x}\big|\frac{dV}{dx}(x)\big|,  $$
where the second inequality is by Lem.~\ref{LemPartTime}.

 Since $\mathbf{m}_{t}'$ is a martingale, we have the first equality below
$$
\tilde{\mathbb{E}}_{\tilde{s}}^{(\lambda)}\big[ \big|   \mathbf{m}_{R_{1}}' \big|  ^{2} \big]= \tilde{\mathbb{E}}_{\tilde{s}}^{(\lambda)}\big[    \langle \mathbf{m}'\rangle_{R_{1}}  \big]  = \tilde{\mathbb{E}}_{\tilde{s}}^{(\lambda)}\Big[    \sum_{n=1}^{\lfloor\frac{1}{2} \tilde{n}_{1} \rfloor }\mathbf{C}^{(\lambda)}_{1}(\sigma_{2n})  \Big]\leq \tilde{\mathbb{E}}_{\tilde{s}}^{(\lambda)}\Big[    \sum_{n=1}^{ \tilde{n}_{1}}\mathbf{C}^{(\lambda)}_{1}(\sigma_{n})  \Big] .   
$$
A similar calculation holds for the term $\tilde{\mathbb{E}}_{\tilde{s}}^{(\lambda)}\big[ \big|   \mathbf{m}_{R_{1}}\big| ^{2} \big]$.  With the above remarks,
\begin{align*}
\tilde{\mathbb{E}}_{\tilde{s}}^{(\lambda)}\Big[ \Big|\int_{0}^{R_{1}}dr\frac{dV}{dx}(X_{r})\Big|\Big] \leq & 4c\sup_{x}\big|\frac{dV}{dx}(x)\big| +\tilde{\mathbb{E}}_{\tilde{s}}^{(\lambda)}\Big[ \sum_{n=1}^{\tilde{n}_{1}  }\big| \mathbf{C}^{(\lambda)}_{0}(\sigma_{n})\big|   \Big] +  2\tilde{\mathbb{E}}_{\tilde{s}}^{(\lambda)}\Big[    \sum_{n=1}^{\tilde{n}_{1} }\mathbf{C}^{(\lambda)}_{1}(\sigma_{n})  \Big]^{\frac{1}{2}}\\ \leq & 4c \sup_{x}\big|\frac{dV}{dx}(x)\big| +c\tilde{\mathbb{E}}_{\tilde{\delta}_{s} }^{(\lambda)}\Big[ \sum_{n=1}^{\tilde{n}_{1}  }\big| \mathbf{C}^{(\lambda)}_{0}(\sigma_{n})\big|   \Big] +  2c^{\frac{1}{2}} \tilde{\mathbb{E}}_{\tilde{\delta}_{s}}^{(\lambda)}\Big[    \sum_{n=1}^{\tilde{n}_{1} }\mathbf{C}^{(\lambda)}_{1}(\sigma_{n})  \Big]^{\frac{1}{2}} \\ \leq & 4c \sup_{x}\big|\frac{dV}{dx}(x)\big| +c\big(U^{(\lambda)}\mathbf{C}^{(\lambda)}_{0}\big)(s)+c\sup_{s\in \supp(h) } \big(U^{(\lambda)}\mathbf{C}^{(\lambda)}_{0}\big)(s)\\ &+ 2c^{\frac{1}{2}}\Big( \big(U^{(\lambda)}\mathbf{C}^{(\lambda)}_{1}\big)(s)+\sup_{s\in \supp(h) } \big(U^{(\lambda)}\mathbf{C}^{(\lambda)}_{1}\big)(s)\Big)^{\frac{1}{2}},
\end{align*} 
where the map $U^{(\lambda)} :L^{\infty}(\Sigma)\rightarrow L^{\infty}(\Sigma)$ was defined in~(\ref{Rockies}).  The second inequality uses that $\tilde{\mathbb{E}}_{\tilde{\delta}_{s} }$ is defined to be $(1-h(s))\tilde{\mathbb{E}}_{(s,0)} +h(s)\tilde{\mathbb{E}}_{(s,1)}    $.  The third inequality follows by Lem.~\ref{ResolventToResolvent}.  Applying  Part (1) from Lem.~\ref{TheCs} for $g_{\lambda}(s)= \mathbf{C}^{(\lambda)}_{0}(s)$ and  $g_{\lambda}(s)= \mathbf{C}^{(\lambda)}_{1}(s)$ in combination with  Lem.~\ref{LifeOperatorTwo}, we obtain a logarithmic bound in $|p|$ for the right side.

\end{proof}

\section{Bounds for the cumulative potential forcing}\label{SecDrift}

In this section we prove Thm.~\ref{LemNullDrift}.

\subsection{The martingale approximating the potential drift process }\label{SecMartDrift}

As discussed in Sect.~\ref{SecProofStrat}, we will  approximate the potential drift $D_{t}$ by the process 
 $$ \tilde{M}_{t}:= \sum_{n=1}^{ \tilde{N}_{t} }\Big(\int_{R_{n}}^{R_{n+1}}dr\frac{dV}{dx}(X_{r})-\big(\frak{R}^{(\lambda)}\frac{dV}{dx}\big)( S_{R_{n}} ) +\big(\frak{R}^{(\lambda)}\frac{dV}{dx}\big)( S_{R_{n+1}} ) \Big).    $$
Part (1) of Prop.~\ref{FurtherNum} and the strong Markov property imply that the increments of $\tilde{M}_{t}$ have finite moments.  The lemma below states that $\tilde{M}_{t}$ is a martingale and gives a closed form for its predictable quadratic variation.  

\begin{lemma}\label{LemKeyMart}
 The process  $\tilde{M}_{t}$ is a martingale with respect to the filtration $\tilde{\mathcal{F}}_{t}'$.    Moreover, the predictable quadratic variation $\langle \tilde{M}\rangle_{t}$ has the form 
$$
   \langle\tilde{M}\rangle_{t}= \sum_{n=1}^{\tilde{N}_{t}}\check{\upsilon}_{\lambda}\big(S_{R_{n}}   \big),
$$
where $\check{\upsilon}_{\lambda}:\Sigma\rightarrow \R^{+}$ is defined as
 \begin{align*} \check{\upsilon}_{\lambda}\big(s  \big)=& 2\tilde{\mathbb{E}}^{(\lambda)}_{\tilde{\delta}_{s} }\Big[\int_{0}^{R_{1}}dr\frac{dV}{dx}(X_{r})\big(\frak{R}^{(\lambda)}\frac{dV}{dx}\big)(S_{r} ) \Big]\\ &+\int_{\Sigma}d\nu (s')\Big(\big(\frak{R}^{(\lambda)}\frac{dV}{dx}\big)( s' )\Big)^{2}
  -\Big( \big(\frak{R}^{(\lambda)}\frac{dV}{dx}\big)( s )    \Big)^{2}.  
 \end{align*}
In the above, $\tilde{\delta}_{s}$ is the splitting of the $\delta$-distribution at $s\in \Sigma$.

\end{lemma}

\begin{proof}   Recall that for a partition time $\mathbf{t}$,   $\tilde{\mathcal{F}}_{\mathbf{t}^{-}}$ refers to the $\sigma$-algebra containing all information before time $\mathbf{t}$  along with  the additional information that  $\mathbf{t}$ is a partition time.  
 The jump times $R_{n}'$ for the process $\tilde{M}_{t}$ are predictable with respect to the filtration  $\tilde{\mathcal{F}}_{t}' $, although we will show that the values of the jumps still have mean zero with respect to the information known before the time of the jump.   We can rewrite the martingale as 
$$ \tilde{M}_{t}= \sum_{n=1}^{ \tilde{N}_{t} }\int_{R_{n}}^{R_{n+1}}dr\frac{dV}{dx}(X_{r})-\tilde{\mathbb{E}}^{(\lambda)}\Big[ \int_{R_{n}}^{R_{n+1}}dr\frac{dV}{dx}(X_{r})   \,\Big|\,\tilde{\mathcal{F}}_{R_{n}^{-}}\Big]  +\tilde{\mathbb{E}}^{(\lambda)}\Big[ \int_{R_{n+1}}^{R_{n+2}}dr\frac{dV}{dx}(X_{r})   \,\Big|\,\tilde{\mathcal{F}}_{R_{n+1}^{-}}\Big]$$
since, by applying Part (3) of Prop.~\ref{BasicsOfNum} and the strong Markov property at time $R_{n}$, we have the first equality below
\begin{align}\label{HouseFly}
\tilde{\mathbb{E}}^{(\lambda)}\Big[ \int_{R_{n}}^{R_{n+1}}dr\frac{dV}{dx}(X_{r})\,\Big|\,\tilde{\mathcal{F}}_{R_{n}^{-}}   \Big]=\tilde{\mathbb{E}}_{ \tilde{\delta}_{S_{R_{n}} } }^{(\lambda)}\Big[ \int_{0}^{R_{1}}dr\frac{dV}{dx}(X_{r})   \Big]=\big(\frak{R}^{(\lambda)}\frac{dV}{dx}\big)(S_{R_{n}})+c,
\end{align}
and the second equality in~(\ref{HouseFly}) is for some $c\in \R$   by Part (3) of Prop.~\ref{BasicsOfNumII}.

For fixed $t\in \R^{+}$ knowledge of whether or not the event $t<R_{n}'$ occurred will be contained in the $\sigma$-algebra $\tilde{\mathcal{F}}_{t}' $ for each $n\in \mathbb{N}$.   The jumps of $\tilde{M}_{t}$ have mean zero since for any $n\in \mathbb{N}$ such that $t<R_{n}'$ the conditional expectation of the  $n$th jump given $\tilde{\mathcal{F}}_{t}'$ is
\begin{align*}
\tilde{\mathbb{E}}^{(\lambda)}\Bigg[\int_{R_{n}}^{R_{n+1}}dr\frac{dV}{dx}(X_{r})-&\tilde{\mathbb{E}}^{(\lambda)}\Big[ \int_{R_{n}}^{R_{n+1}}dr\frac{dV}{dx}(X_{r})   \,\Big|\,\tilde{\mathcal{F}}_{R_{n}^{-}}\Big]  +\tilde{\mathbb{E}}^{(\lambda)}\Big[ \int_{R_{n+1}}^{R_{n+2}}dr\frac{dV}{dx}(X_{r}) \,\Big|\,\tilde{\mathcal{F}}_{R_{n+1}^{-}}  \Big]\,\Bigg|\,\tilde{\mathcal{F}}_{t}' \Bigg] \\ &  =\tilde{\mathbb{E}}_{ \tilde{\nu} }^{(\lambda)}\Big[ \int_{0}^{R_{1}}dr\frac{dV}{dx}(X_{r})   \Big]=0 .
\end{align*}
To get  the first equality above, we can insert a nested conditional expectation with respect to $\tilde{\mathcal{F}}_{R_{n}^{-}}$  since  $\tilde{\mathcal{F}}_{t}'\subset \tilde{\mathcal{F}}_{R_{n}^{-}}$, and the first two terms in the expectation then cancel.  For the third term in the first equality, we use the strong Markov property at the time $R_{n+1}$ and that $\tilde{S}_{R_{n+1}}$ has distribution $\tilde{\nu}$ when conditioned on $\tilde{\mathcal{F}}_{R_{n}}$ (and thus when conditioned on $\tilde{\mathcal{F}}_{t}'$)  by Part (1) of Prop.~\ref{IndependenceProp}.  The second equality is by Part (2) of  Prop.~\ref{BasicsOfNumII} with $g(x,p)=\frac{dV}{dx}(x)$ since $\Psi_{\infty,\lambda}(\frac{dV}{dx})=0$.  Thus $\tilde{M}_{t}$ is a martingale.  

For $ t\in [R_{n-1}',R_{n}') $ the  $\sigma$-algebra $\tilde{\mathcal{F}}_{t}'$
 contains all information before time $R_{n}$, i.e., $\tilde{\mathcal{F}}_{t}'=\tilde{\mathcal{F}}_{R_{n}^{-} }$, along with knowledge of the value $R_{n}$.  The predictable quadratic variation $\langle \tilde{M}\rangle_{t}$ must have the form of a discrete sum over $\sum_{n=1}^{\tilde{N}_{t}}$ because the jump times $R_{n}'$ are predictable.  As a consequence of  Part (3) of Prop.~\ref{BasicsOfNum}, the conditional distribution for $\tilde{S}_{R_{n}}$ given $\tilde{\mathcal{F}}_{t}'=\tilde{\mathcal{F}}_{R_{n}^{-} }$ is $\tilde{\delta}_{S_{R_{n}}} $. The variance of the $n$th jump for $\tilde{M}_{t}$  conditioned on $\tilde{\mathcal{F}}_{t}'$  is $\check{\upsilon}_{\lambda}(S_{R_{n}}) $, where
\begin{align*}
\check{\upsilon}_{\lambda}(s)  =&\tilde{\mathbb{E}}_{\tilde{\delta}_{s} }^{(\lambda)}\Big[\Big(\int_{0}^{R_{1}}dr\frac{dV}{dx}(X_{r})-\tilde{\mathbb{E}}_{ \tilde{\delta}_{s} }^{(\lambda)}\Big[ \int_{0}^{R_{1}}dr\frac{dV}{dx}(X_{r})   \Big] +\tilde{\mathbb{E}}_{ \tilde{\delta}_{S_{R_{1}} } }^{(\lambda)}\Big[ \int_{0}^{R_{1}}dr\frac{dV}{dx}(X_{r})   \Big]\Big)^{2}\Big]\\ = & 2\tilde{\mathbb{E}}^{(\lambda)}_{\tilde{\delta}_{s}}\Big[\int_{0}^{R_{1}}dr\frac{dV}{dx}(X_{r})\int_{r}^{R_{2}}dr'\frac{dV}{dx}(X_{r'} )\Big]+\int_{\Sigma}d\nu(s')\Big(\tilde{\mathbb{E}}_{\tilde{\delta}_{s'}}^{(\lambda)}\Big[  \int_{0}^{R_{1}}dr\frac{dV}{dx}(X_{r})        \Big]\Big)^{2}
 \\ &   -\Big(\tilde{\mathbb{E}}_{\tilde{\delta}_{s} }^{(\lambda)}\Big[  \int_{0}^{R_{1}}dr\frac{dV}{dx}(X_{r})        \Big]\Big)^{2}. 
\end{align*}
This expression for $\check{\upsilon}_{\lambda}(s)$ can be written in terms of $ \big( \frak{R}^{(\lambda)}\,\frac{dV}{dx}\big)(s)$ as done in the statement of this lemma by using that  $ \big( \frak{R}^{(\lambda)}\,\frac{dV}{dx}\big)(s)=\tilde{\mathbb{E}}_{ \tilde{\delta}_{s} }^{(\lambda)}\big[ \int_{0}^{R_{1}}dr\frac{dV}{dx}(X_{r})   \big]-c$ and the same reasoning in the proof of Part (4) of Prop.~\ref{BasicsOfNumII}.    
\end{proof}

 Define  $\upsilon_{\lambda}:=\int_{\Sigma}d\nu(s) \check{\upsilon}_{\lambda}(s)$.
\begin{lemma}\label{LemUpsilon}
The value $\upsilon_{\lambda}\in \R^{+}$ is uniformly bounded for $\lambda<1$, and $\upsilon_{\lambda}$ depends continuously on the parameter $\lambda$.
\end{lemma}

\begin{proof}
By the proof of Prop.~\ref{BasicsOfNumII}, the value $\upsilon_{\lambda}\in \R^{+}$ can be written as
\begin{align}\label{Lionel}
\upsilon_{\lambda}=&\tilde{\mathbb{E}}_{\tilde{\nu}}^{ (\lambda)}   \Big[  \int_{0}^{R_{1}}dr \frac{dV}{dx}(X_{r})  \int_{r}^{R_{2}}dr'\frac{dV}{dx}(X_{r'})
    \Big] \nonumber   \\  =&\tilde{\mathbb{E}}_{\tilde{\nu}}^{ (\lambda)}   \Big[  \int_{0}^{R_{1}}dr \frac{dV}{dx}(X_{r})  \int_{r}^{R_{1}}dr'\frac{dV}{dx}(X_{r'})
    \Big]+ \tilde{\mathbb{E}}_{\tilde{\nu}}^{ (\lambda)}   \Big[  \int_{0}^{R_{1}}dr \frac{dV}{dx}(X_{r})  \int_{R_{1}}^{R_{2}}dr'\frac{dV}{dx}(X_{r'})
    \Big] \nonumber \\
 \leq  &  \frac{3}{2}\tilde{\mathbb{E}}_{\tilde{\nu}}^{ (\lambda)}   \Big[  \sup_{0\leq t \leq R_{1}}\Big(\int_{0}^{t}dr \frac{dV}{dx}(X_{r}) \Big)^{2}
    \Big]  +\frac{1}{2} \tilde{\mathbb{E}}_{\tilde{\nu}}^{ (\lambda)}   \Big[  \Big( \int_{R_{1}}^{R_{2}}dr'\frac{dV}{dx}(X_{r'})
\Big)^{2}    \Big] \nonumber \\  =  &  2\tilde{\mathbb{E}}_{\tilde{\nu}}^{ (\lambda)}   \Big[  \sup_{0\leq t \leq R_{1}}\Big(\int_{0}^{t}dr \frac{dV}{dx}(X_{r}) \Big)^{2}
    \Big]  ,
\end{align}
where the inequality applies the relation $2ab\leq  a^{2}+b^{2}$ for $a= \int_{0}^{R_{1}}dr \frac{dV}{dx}(X_{r}) $ and $b=\int_{R_{1}}^{R_{2}}dr'\frac{dV}{dx}(X_{r'})$,       and the third equality uses the strong Markov property and the fact that $\tilde{S}_{R_{1}}$ has distribution $\tilde{\nu}$ by Part (1) of Prop.~\ref{IndependenceProp}.  The right side of~(\ref{Lionel}) is uniformly bounded for $\lambda<1$ by Part (1) of Prop.~\ref{FurtherNum}, and  hence the values $\upsilon_{\lambda}$ are uniformly bounded for small $\lambda$.   By the same reasoning as above, the second moment of the random variable $Y^{(\lambda)}:=\int_{0}^{R_{1}}dr \frac{dV}{dx}(X_{r})  \int_{r}^{R_{2}}dr'\frac{dV}{dx}(X_{r'})|_{\tilde{\mathbb{P}}_{\tilde{\nu}}^{ (\lambda)}}$  is uniformly bounded for the  parameter values $\lambda<1$.  To  clarify, ``$Y^{(\lambda)}$" refers to the stated random variable with respect to the parameter value $\lambda\in \R^{+}$ for the underlying statistics.     With the uniform bound on the second moment,  it is  sufficient that the distribution for the random variable $Y^{(\lambda)}$ is weakly continuous as a function of $\lambda\in \R^{+}$ to guarantee that   $\upsilon_{\lambda}$ is continuous.  If the random times $R_{1},R_{2}$ had deterministic upper bounds, then the weak convergence would be clear. By introducing a time cut-off $t>0$, the random variable $Y_{t}^{(\lambda)}:=\int_{0}^{R_{1}\wedge t}dr \frac{dV}{dx}(X_{r})  \int_{r}^{R_{2}\wedge t}dr'\frac{dV}{dx}(X_{r'})|_{\tilde{\mathbb{P}}_{\tilde{\nu}}^{ (\lambda)}}$ is weakly continuous as a function of $\lambda$.    The random variables $Y_{t}^{(\lambda)}$ converge in probability uniformly to $Y^{(\lambda)}$ as $t\rightarrow \infty$ for $\lambda<1$ since $R_{2}$ has uniform fractional moments $\tilde{\mathbb{E}}^{(\lambda)}_{\tilde{\nu}}[R_{2}^{\frac{1}{4}}]$ for $\lambda<1$ by  Lem.~\ref{FracMoment}.  For $\lambda_{n}\rightarrow \lambda<1$, the following diagram characterizes the convergences involved: 
     $$  \renewcommand{\arraystretch}{1.5} \begin{array}[c]{ccc}  Y_{t}^{(\lambda_{n})} & \stackrel{w}{\Longrightarrow   }  &   Y_{t}^{(\lambda)}\\   \Big \downarrow\scriptstyle{ \frak{P}}&  &  \Big \downarrow\scriptstyle{\frak{P}}\\      Y^{(\lambda_{n})} &\Longrightarrow  &   Y^{(\lambda)}
\end{array},
$$
where the downward arrows signify a uniform sense for the convergence in probability as $t\rightarrow \infty$.  The above convergences imply that $Y^{(\lambda_{n})} $ converges weakly to $ Y^{(\lambda)}$ as $n\rightarrow \infty$.

\end{proof}

The following lemma relates the predictable quadratic variation of $ \tilde{M}_{t}$ to the counting process $\tilde{N}_{t}$ and is somewhat stronger than we require.

\begin{lemma}\label{LocalTimeBndII} As $\lambda\searrow 0$ the following order equality holds:
$$ \tilde{\mathbb{E}}^{(\lambda)}\Big[\sup_{0\leq t\leq T}\Big|\lambda^{\frac{1}{2}} \langle \tilde{M}\rangle_{\frac{t}{\lambda}}    -\lambda^{\frac{1}{2} }\upsilon_{\lambda} \tilde{N}_{\frac{t}{\lambda}}\Big|      \Big]=\mathit{O}(\lambda^{\frac{1}{4}}).     $$
Also, for any $t\in \R^{+}$ the following expectations are equal: 
 $\tilde{\mathbb{E}}^{(\lambda)}\big[\langle \tilde{M}\rangle_{t}  \big] =\upsilon_{\lambda}\tilde{\mathbb{E}}^{(\lambda)}\big[ \tilde{N}_{t}\big]$.

\end{lemma}

\begin{proof}
We only prove the equality $\tilde{\mathbb{E}}^{(\lambda)}\big[\langle \tilde{M}\rangle_{t}  \big] =\upsilon_{\lambda}\tilde{\mathbb{E}}^{(\lambda)}\big[ \tilde{N}_{t}\big]$ here, which is what we use in the proof of Thm.~\ref{LemNullDrift}.  The remainder of the proof is placed in Sect.~\ref{SecProofMartDrift}.    Lemma~\ref{LemKeyMart} yields the first equality below:
\begin{align*}
\tilde{\mathbb{E}}^{(\lambda)}\big[\langle \tilde{M}\rangle_{t}\big]=\tilde{\mathbb{E}}^{(\lambda)}\Big[ \sum_{n=1}^{\tilde{N}_{t}} \check{\upsilon}_{\lambda}(S_{R_{n}}) \Big]=\tilde{\mathbb{E}}^{(\lambda)}\Big[ \sum_{n=1}^{\tilde{N}_{t}} \mathbb{E}\big[\check{\upsilon}_{\lambda}(S_{R_{n}})\,\big|\,\tilde{\mathcal{F}}_{R_{n}'} \big]\Big]=\upsilon_{\lambda}\tilde{\mathbb{E}}^{(\lambda)}\big[ \tilde{N}_{t}\big].
\end{align*}
The third equality holds since $S_{R_{n}}$ has distribution $\nu$ when conditioned on $\tilde{\mathcal{F}}_{R_{n}'}$ by  Part (1) of Prop.~\ref{IndependenceProp}.

\end{proof}

\subsection{Proof of Thm.~\ref{LemNullDrift}}

\begin{proof}[Proof of Thm~\ref{LemNullDrift}]
  We will work with the split dynamics to show that
$$\mathbb{E}^{(\lambda)}\Big[\sup_{t\in [0,\frac{T}{\lambda}]}|D_{t}|\Big]=\tilde{\mathbb{E}}^{(\lambda)}\Big[\sup_{t\in [0,\frac{T}{\lambda}]}|D_{t}|\Big]=\mathit{O}( \lambda^{-\frac{1}{4}}).  $$
Let  $\tilde{M}_{t}$ be the martingale from Lem.~\ref{LemKeyMart}.  We can write $D_{t}$ as
\begin{align}
D_{t}= & \int_{0}^{R_{1}}dr\frac{dV}{dx}(X_{r})-\int_{t}^{R_{\tilde{N}_{t}+1} }dr\frac{dV}{dx}(X_{r}) \nonumber \\ &+\tilde{\mathbb{E}}_{ \tilde{S}_{R_{1}} }\Big[ \int_{0}^{R_{1}}dr\frac{dV}{dx}(X_{r})   \Big] -\tilde{\mathbb{E}}_{ \tilde{S}_{R_{\tilde{N}_{t}+1}} }\Big[ \int_{0}^{R_{1}}dr\frac{dV}{dx}(X_{r})   \Big]+ \tilde{M}_{t}   ,   
\end{align}
where we have used Part (3) of Prop.~\ref{BasicsOfNumII}.
The triangle inequality  gives
\begin{align}\label{FirstLast}
\mathbb{E}^{(\lambda)}\Big[\sup_{t\in [0,\frac{T}{\lambda}]}|D_{t}|\Big]\leq & 2\tilde{\mathbb{E}}^{(\lambda)}\Big[\Big|\int_{0}^{R_{1}}dr\frac{dV}{dx}(X_{r})\big|\Big]\nonumber  \\ &+3\tilde{\mathbb{E}}^{(\lambda)}\Big[\sup_{R_{1}\leq t\leq \frac{T}{\lambda} }\Big|\int_{t}^{R_{\tilde{N}_{t}+1 } }dr\frac{dV}{dx}(X_{r})\Big| \Big] +\mathbb{E}^{(\lambda)}\Big[\sup_{t\in [0,\frac{T}{\lambda}]}|\tilde{M}_{t}|\Big].
\end{align}

 For the first term on the right side of~(\ref{FirstLast}), 
we can apply Part (2) of Prop.~\ref{FurtherNum} to get
\begin{align}
\tilde{\mathbb{E}}^{(\lambda)}\Big[\Big|\int_{0}^{R_{1}}dr\frac{dV}{dx}(X_{r})\big|\Big]&\leq \int d\tilde{\mu}(x,p,z) \tilde{\mathbb{E}}_{(x,p,z) }^{(\lambda)}\Big[\Big|\int_{0}^{R_{1}}dr\frac{dV}{dx}(X_{r})\Big|\Big]\nonumber\\ &\leq C\int_{\R}d\mu(x,p)\big(1+\log(1+|p|)\big)<   C\int_{\R}d\mu(x,p)\big(1+|p|\big),   \label{Qaddafi}
\end{align}
where $\mu$ is the initial measure on $\Sigma$ and $\tilde{\mu}$ is its splitting.  By our assumption on $\mu$, the first moment is finite, and hence~(\ref{Qaddafi}) is finite.

The second term on the right side of~(\ref{FirstLast}) is less simple since it depends on $t$.  Writing  $$\int_{t}^{R_{\tilde{N}_{t+1}} }dr\frac{dV}{dx}(X_{r})=\int_{R_{\tilde{N}_{t} } }^{R_{\tilde{N}_{t} +1} }dr\frac{dV}{dx}(X_{r})   -\int_{R_{\tilde{N}_{t} } }^{t}dr\frac{dV}{dx}(X_{r}),$$ then with the triangle inequality     
\begin{align}
\tilde{\mathbb{E}}^{(\lambda)}\Big[\sup_{0\leq t\leq \frac{T}{\lambda} }\Big|\int_{t}^{R_{\tilde{N}_{t}+1 } }dr\frac{dV}{dx}(X_{r})\Big| \Big]&\leq  \tilde{\mathbb{E}}^{(\lambda)}\Big[\sup_{0\leq t\leq \frac{T}{\lambda} }\Big(\Big|\int_{R_{\tilde{N}_{t}} }^{t}dr\frac{dV}{dx}(X_{r})\Big|+\Big|\int_{R_{\tilde{N}_{t}} }^{R_{\tilde{N}_{t}+1}}dr\frac{dV}{dx}(X_{r})\Big| \Big)\Big] \nonumber\\  &\leq 2 \tilde{\mathbb{E}}^{(\lambda)}\Big[\sup_{1\leq m\leq \tilde{N}_{t}} \sup_{t\in [R_{m},R_{m+1}]}   \Big|\int_{R_{m} }^{t   }dr\frac{dV}{dx}(X_{r})\Big| \Big].\label{Chavez}
\end{align}
We can get rid of the first supremum with the inequality
\begin{align*}
\tilde{\mathbb{E}}^{(\lambda)}\Big[\sup_{1\leq m\leq \tilde{N}_{\frac{T}{\lambda} }} \sup_{t\in [R_{m},R_{m+1}]}   \Big|\int_{R_{m} }^{t  }dr\frac{dV}{dx}(X_{r})\Big| \Big]&\leq \tilde{\mathbb{E}}^{(\lambda)}\Big[\sum_{m=1}^{\tilde{N}_{\frac{T}{\lambda} }} \sup_{t\in [R_{m},R_{m+1}]}   \Big|\int_{R_{m} }^{t  }dr\frac{dV}{dx}(X_{r})\Big|^{2} \Big]^{\frac{1}{2}}\nonumber\\  &= \tilde{\mathbb{E}}^{(\lambda)}\Big[\sum_{m=1}^{\tilde{N}_{\frac{T}{\lambda} }}\tilde{\mathbb{E}}^{(\lambda)}\Big[ \sup_{t\in [R_{m},R_{m+1}]}   \Big|\int_{R_{m} }^{t  }dr\frac{dV}{dx}(X_{r})\Big|^{2}\,\Big|\,\tilde{\mathcal{F}}_{R_{m}'}  \Big] \Big]^{\frac{1}{2}}\nonumber  \\ &=  \tilde{\mathbb{E}}^{(\lambda)}[\tilde{N}_{\frac{T}{\lambda} }]^{\frac{1}{2}}  \tilde{\mathbb{E}}_{\tilde{\nu}}^{(\lambda)}\Big[ \sup_{t\in [0,R_{1}]}   \Big|\int_{0 }^{t   }dr\frac{dV}{dx}(X_{r})\Big|^{2} \Big]^{\frac{1}{2}}\nonumber \\ &=\mathit{O}(\lambda^{-\frac{1}{4}}).
\end{align*}
In the second equality, we have applied the strong Markov property at the times $R_{m}'$ and have used that $\tilde{S}_{R_{m}}$ is distributed as $\tilde{\nu}$ given $\tilde{\mathcal{F}}_{R_{m}'}$ by Part (1) of Prop.~\ref{IndependenceProp}.  By Part (1) of Prop.~\ref{FurtherNum}, the expectation $\tilde{\mathbb{E}}_{\tilde{\nu}}^{(\lambda)}$ above is bounded uniformly for $\lambda<1$.  The expectation $\tilde{\mathbb{E}}^{(\lambda)}[\tilde{N}_{\frac{T}{\lambda} }]$ is $\mathit{O}(\lambda^{-\frac{1}{2}})$ by Lem.~\ref{LocalTimeBnd}.

Now for the last term in~(\ref{FirstLast}). Applying Jensen's and then Doob's maximal inequality, we get the first two relations below:
\begin{align}\label{Radish}
\tilde{\mathbb{E}}^{(\lambda)}\Big[ \sup_{0\leq t\leq \frac{T}{\lambda} }\big|\tilde{M}_{t}\big| \Big] & \leq \tilde{\mathbb{E}}^{(\lambda)}\Big[ \sup_{0\leq t\leq \frac{T}{\lambda} }\big|\tilde{M}_{t}\big|^{2} \Big]^{\frac{1}{2}}\leq 2\tilde{\mathbb{E}}^{(\lambda)}\big[ \big|\tilde{M}_{\frac{T}{\lambda} }\big|^{2} \big]^{\frac{1}{2}}\nonumber \\ &=2\tilde{\mathbb{E}}^{(\lambda)}\big[ \langle\tilde{M}\rangle_{\frac{T}{\lambda} }\big]^{\frac{1}{2}}   =2\upsilon_{\lambda}^{\frac{1}{2}}\tilde{\mathbb{E}}^{(\lambda)}\big[  \tilde{N}_{\frac{T}{\lambda}}  \big]^{\frac{1}{2}}=\mathit{O}(\lambda^{-\frac{1}{4}}).
\end{align}
The second equality is by Lem.~\ref{LocalTimeBndII}, and we again apply Lem.~\ref{LocalTimeBnd} to get $\tilde{\mathbb{E}}^{(\lambda)}\big[  \tilde{N}_{\frac{T}{\lambda}}  \big]=\mathit{O}(\lambda^{-\frac{1}{2}})$. Also,  $\upsilon_\lambda$ is uniformly bounded for $\lambda<1$ by Lem.~\ref{LemUpsilon}.

\end{proof}

\section{Convergence to the Ornstein-Uhlenbeck process}\label{BrownProof}

In this section we will now prove Thm.~\ref{ThmMain}. As a preliminary,   Sect.~\ref{SubSecJumpProc} characterizes the martingale and drift components in the semi-martingale decomposition for the jump process $J_{t}$ defined in~(\ref{ContToMom}).  The main ingredient for the study of $J_{t}$ is the bound in Lem.~\ref{FirstEnergyLem} on the typical energy of the particle obtained over the time interval $[0,\frac{T}{\lambda}]$  for small $\lambda$.  

\subsection{Limiting behavior for the jump process} \label{SubSecJumpProc}

The jump processes $J_{t}$  can be written as a sum of a martingale $M_{t}$ and a predictable part as
$$J_{t}=M_{t}+\int_{0}^{t}dr\mathcal{D}_{\lambda}(P_{r}) \quad  \text{for}\quad         \mathcal{D}_{\lambda}(p)=\int_{\R}dp^{\prime}(p^{\prime}-p) \mcJ_{\lambda}(p,p^{\prime}).$$
In order to write  an expression for the predictable quadratic variation $\langle M\rangle_{t}$ in terms of the jump rates $\mathcal{J}_{\lambda}$, let us define for $m\in\N$
\begin{align}\label{BigPi}
\Pi_{\lambda,m}(p)=\int_{\R}dp^{\prime}(p^{\prime}-p)^{m} \mcJ_{\lambda}(p,p^{\prime}). 
\end{align}
Note that  $\mathcal{E}_{\lambda}(p):=\Pi_{\lambda,0}(p)$ is the escape rate for the jump process when $P_{t}=p$, and $\mathcal{D}_{\lambda}$ is equal to $\Pi_{\lambda ,m}$ for  $m=1$.  We also define $\mathcal{Q}_{\lambda}(p)=\Pi_{\lambda,2}(p)$.  The predictable quadratic variation for $M_{t}$ can then be written  in the closed form
$$ \langle M\rangle_{t}= \int_{0}^{t}dr\mathcal{Q}_\lambda(P_{r}) .   $$

The proposition below collects some facts regarding the functions $\mathcal{D}_{\lambda}(p)$.  

\begin{proposition}\label{PreStuff} \text{  }
There are constants $C, C_m >0$ such that for all $\lambda\leq 1$ and all $p\in\R$,
\begin{enumerate}
\item $\frac{1}{8(\lambda +1)} \leq \mathcal{E}_\lambda(p)\leq \frac{1}{8(\lambda +1)}\left(1+C\lambda |p|\right)$ and $\lambda |p| \leq C\mathcal{E}_\lambda(p),$

\item $ \big|\mathcal{D}_{\lambda}(p)  +\frac{\lambda p}{2}\big|\leq C\lambda^2 (|p|+p^{2})$,

\item $  \left|  \mathcal{Q}_{\lambda}(p)-1 \right| \leq C\lambda+C\lambda|p|+C\lambda^{3}|p|^3$,

\item $ \Pi_{\lambda,2m}(p)\leq C_m(1+\lambda|p|)^{2m+1}  $.  

\end{enumerate}

\end{proposition}

Next we show that the drift rate $\mathcal{D}_{\lambda}(p)$ can be effectively replaced by the linear form $-\frac{\lambda p}{2}$.  Let $P_{t}^{\prime}$ be the solution to the integral equation
\begin{align}\label{LinDrift}
P_{t}^{\prime}=P_{0}-\int_{0}^{t}dr\frac{dV}{dx}( X_{r})-\frac{1}{2}\lambda \int_{0}^{t}drP_{r}^{\prime} +M_{t}.   
\end{align}
Also  let us denote  $P_{t}^{(\lambda)}:=\lambda^{\frac{1}{2}}P_{\frac{t}{\lambda} }^{(\lambda)}$ and $P_{t}^{(\lambda),\prime}:=\lambda^{\frac{1}{2}}P_{\frac{t}{\lambda} }^{\prime}$.

\begin{lemma}\label{ChangeDrift}
 Fix $T>0$.  The difference between $P_{t}^{(\lambda),\prime}$ and $P_{t}^{(\lambda),\prime}$ for small $\lambda>0 $ satisfies
$$ \mathbb{E}^{(\lambda)}\big[\sup_{0\leq t\leq T}\big|P_{t}^{(\lambda),\prime}- P_{t}^{(\lambda)}\big|\big]=\mathit{O}(\lambda^{\frac{1}{2}}). $$

\end{lemma}

\begin{proof}
Since $P_{t}$ satisfies the integral equation
$$ P_{t}=P_{0}-\int_{0}^{t}dr\frac{dV}{dx}(X_{r})+ \int_{0}^{t}dr \mathcal{D}_{\lambda}(P_{r}) +M_{t},  $$
and $P_{t}'$ satisfies~(\ref{LinDrift}),  we have the identity
$$P_{t}^{(\lambda),\prime}- P_{t}^{(\lambda)}= \int_{0}^{t}dr R_{\lambda}(P_{r}^{(\lambda) })-\frac{1}{2}\int_{0}^{t}dr e^{-\frac{1}{2}(t-r)}\int_{0}^{r}ds R_{\lambda}(P_{s}^{(\lambda) })=\int_{0}^{t} dre^{-\frac{1}{2}(t-r)} R_{\lambda}(P_{r}^{(\lambda) }) , $$
where $R_{\lambda}(p)= -\lambda^{-\frac{1}{2}}D_{\lambda}(\frac{p}{\lambda^{\frac{1}{2}}})-\frac{1}{2} p$. Thus, we have the first inequality below:
\begin{align*}
\mathbb{E}^{(\lambda)}\Big[\sup_{0\leq t\leq T}\big|P_{t}^{(\lambda),\prime}- P_{t}^{(\lambda)}\big|\Big] & \leq 2(1-e^{-\frac{1}{2}T }) \mathbb{E}^{(\lambda)}\Big[\sup_{0\leq t\leq T} \big|R_{\lambda}(P_{t}^{(\lambda)} )   \big| \Big]\\ & \leq  2C\lambda^{\frac{3}{2}}+2C\lambda^{\frac{1}{2}} \mathbb{E}^{(\lambda)}\Big[\sup_{0\leq t\leq T} \big| \lambda^{\frac{1}{2}}P_{\frac{t}{\lambda}}     \big|^{2} \Big]\\ & =\mathit{O}(\lambda^{\frac{1}{2}}).
\end{align*}
The second inequality follows by Part (2) of Prop.~\ref{PreStuff}.  By bounding $|P_{t}|\leq (2H_{t})^{\frac{1}{2}}$ and applying Lem.~\ref{FirstEnergyLem}, the expectation on the second line is uniformly bounded for $\lambda<1$.  Thus, the above is $\mathit{O}(\lambda^{\frac{1}{2}})$.

\end{proof}

The following lemma gives a central limit theorem for the martingale $M_{t}^{(\lambda)}=\lambda^{\frac{1}{2}}M_{\frac{t}{\lambda}}$.

\begin{lemma}\label{MartPart}
As $\lambda\searrow 0$ the martingale $M^{(\lambda)}_{t}=\lambda^{\frac{1}{2}} M_{\frac{t}{\lambda}}$ converges in law with respect to the uniform metric to a standard Brownian motion $\mathbf{B}$ over the interval $t\in [0,T]$.

\end{lemma}

\begin{proof}
To prove the central limit theorem, we prove the following:
\begin{enumerate}[(i).]
\item  For each $t\in \R^{+}$, the predictable quadratic variation process $\langle M^{(\lambda)}\rangle _{t}$ converges in probability to $t$ as $\lambda\searrow 0$.   

\item  For any $\epsilon>0$,  then as $\lambda\rightarrow 0$
$$\mathbb{P}^{(\lambda)}\Big[\sup_{0\leq r\leq  \frac{T}{\lambda} }\big(M_{r}-M_{r-}\big)^{2}> \frac{\epsilon}{\lambda}    \Big]\longrightarrow 0.    $$

\end{enumerate}
By~\cite[Thm. VIII.2.13]{Pollard} (i) and (ii) are sufficient to prove that $M^{(\lambda)}_{t}$ converges in law to a Brownian motion.     

 \vspace{.25cm} 

\noindent (i).\hspace{.1cm} We prove a somewhat stronger statement.  Note that
$$\langle M^{(\lambda)}\rangle_{t}-t= \lambda \int_{0}^{\frac{t}{\lambda}}dr\big(\mathcal{Q}_{\lambda}(P_{r})-1\big).     $$
For the expectation of the supremum of the difference between $\langle M^{(\lambda)}\rangle_{t}$ and $t$ over the interval $[0,T]$, we have
\begin{align*}
 \mathbb{E}^{(\lambda)}\Big[\sup_{t\in[0,T]} \Big| \langle M^{(\lambda)}\rangle_{t}-t   \Big|\Big]& \leq \lambda \mathbb{E}^{(\lambda)}\Big[ \int_{0}^{\frac{T}{\lambda}}dr\big|\mathcal{Q}_{\lambda}(P_{r})-1\big|\Big]\\   
 & \leq CT\lambda^{\frac{1}{2}}\Big(\lambda^{\frac{1}{2}}+  \mathbb{E}^{(\lambda)}\Big[ \sup_{0\leq r\leq \frac{T}{\lambda} } \lambda^{\frac{1}{2}} |P_{r}|\Big]+ \lambda\mathbb{E}^{(\lambda)}\Big[ \sup_{0\leq r\leq \frac{T}{\lambda} } \lambda^{\frac{3}{2}} |P_{r}|^{3}\Big]\Big)\\ &=\mathit{O}(\lambda^{\frac{1}{2}}),
\end{align*}
where the second inequality is for some $C>0$ by Part (3) of Prop.~\ref{PreStuff}.   The expectations in the second line above are bounded uniformly for $\lambda<1$ by  Lem.~\ref{FirstEnergyLem} since $|P_{r}|\leq 2^{\frac{1}{2}}H_{r}^{\frac{1}{2}}$.  The above implies that $\langle M^{(\lambda)}\rangle_{t}$ converges in probability to $t$ as $\lambda\searrow 0$.

\vspace{.5cm}

\noindent (ii).\hspace{.1cm} Recall that $\calN_{t}$ is the number of collisions over the time interval $[0,t]$ and that $t_{1},\dots, t_{\calN_{t}}$ are the corresponding jump times.  The probability has the following bounds:
\begin{align}
\mathbb{P}^{(\lambda)}\Big[\sup_{0\leq r\leq  \frac{T}{\lambda} }\big(M_{r}-M_{r^-}\big)^{2}> \frac{\epsilon}{\lambda}    \Big] &\leq \frac{\lambda}{\epsilon}\mathbb{E}^{(\lambda)}\Big[\Big( \sum_{n=1}^{\calN_{\frac{T}{\lambda} } } \big(M_{t_{n}}-M_{t_{n}^-}\big)^{4}\Big)^{\frac{1}{2}}\Big]\nonumber  \\ & \leq \frac{\lambda}{\epsilon}\mathbb{E}^{(\lambda)}\Big[ \sum_{n=1}^{\calN_{\frac{T}{\lambda} } } \big(M_{t_{n}}-M_{t_{n}^-}\big)^{4}\Big]^{\frac{1}{2}} \nonumber\\  &= \frac{\lambda}{\epsilon}\mathbb{E}^{(\lambda)}\Big[ \sum_{n=1}^{\calN_{\frac{T}{\lambda} } } \mathbb{E}^{(\lambda)}\Big[\big(M_{r}-M_{r^-}\big)^{4}\,\big|\,P_{r^-},\,\calN_{r}=\calN_{r^{-}}+1   \Big]\Big|_{r=t_{n}}\Big]^{\frac{1}{2}}\nonumber\\ &=  \frac{\lambda}{\epsilon}\mathbb{E}^{(\lambda)}\Big[ \sum_{n=1}^{\calN_{\frac{T}{\lambda}} } \frac{\Pi_{\lambda,4}(P_{t_{n}^-}) }{\mathcal{E}_{\lambda}(P_{t_{n}^-})}   \Big]^{\frac{1}{2}}= \frac{\lambda}{\epsilon}\mathbb{E}^{(\lambda)}\Big[ \int_{0}^{\frac{T}{\lambda} }dr \Pi_{\lambda, 4}(P_{r})  \Big]^{\frac{1}{2}}. \label{Gomer}
\end{align}
The second inequality is Jensen's, and the first inequality is Chebyshev's followed by the elementary relation
$$\hspace{2cm} \sup_{1\leq m\leq n} a_{m}\leq  \Big(\sum_{m=1}^{n} a_{m}^{2}\Big)^{\frac{1}{2}},\hspace{1cm} a_{n}\geq 0.  $$
The first equality in~(\ref{Gomer})  holds since the process
$$ \sum_{n=1}^{\calN_{t } }\Big( \big(M_{t_{n}}-M_{t_{n}^-}\big)^{4}- \mathbb{E}^{(\lambda)}\Big[\big(M_{r}-M_{r^-}\big)^{4}\,\big|\,P_{r^-},\,\calN_{r}=\calN_{r^{-}}+1   \Big]\Big|_{r=t_{n}}  \Big) $$
is a mean zero martingale.  The second equality uses that a jump for $M_{r}$ is a jump for $P_{r}$ (since they differ by a continuous process) and  that the conditional expectation for $\big(P_{r}-P_{r^{-}}\big)^{4}$ given the value $P_{r^{-}}$ and the information that $r\in \R^{+}$ is a jump time is given by the ratio of $\Pi_{\lambda,4}(P_{r^-})$ by $\mathcal{E}_{\lambda}(P_{r^{-}})$:
$$ \mathbb{E}^{(\lambda)}\Big[\big(M_{r}-M_{r^{-} }\big)^{4}\,\big|\,P_{r^{-}},\,\calN_{r}=\calN_{r^{-}}+1   \Big]= \mathbb{E}^{(\lambda)}\Big[\big(P_{r}-P_{r^-}\big)^{4}\,\big|\,P_{r^-},\,\calN_{r}=\calN_{r^{-}}+1   \Big]=\frac{\Pi_{\lambda,4}(P_{r^-}) }{\mathcal{E}_{\lambda}(P_{r^-})} . $$
The last equality follows because the jump times $t_{n}$ occur with Poisson rate  $\mathcal{E}_{\lambda}(P_{r})$.

Squaring the right side of~(\ref{Gomer}),
$$
\frac{\lambda^{2}}{\epsilon^{2}}\mathbb{E}^{(\lambda)}\Big[ \int_{0}^{\frac{T}{\lambda} }dr \Pi_{\lambda,4}(P_{r})  \Big]\leq C\frac{\lambda^2}{\epsilon^2}\mathbb{E}^{(\lambda)}\Big[ \int_{0}^{\frac{T}{\lambda} }dr\big(1+\lambda|P_{r}|\big)^{5}  \Big]\leq C\frac{\lambda}{\epsilon^{2}}\mathbb{E}^{(\lambda)}\Big[ \lambda \int_{0}^{\frac{T}{\lambda} }dr\big(1+\lambda 2^{\frac{1}{2}} H_{r}^{\frac{1}{2}} \big)^{5}  \Big],
$$
where we have applied Part (4) of Prop.~\ref{PreStuff} in the first inequality, and the bound $|P_{r}|^{2}\leq 2H_{r}$ for the second. By Lem.~\ref{FirstEnergyLem} the expectation on the right side is uniformly bounded for $\lambda<1$.   Thus, the above  is $\mathit{O}(\lambda)$ and the left side of~(\ref{Gomer}) is $\mathit{O}(\lambda^{\frac{1}{2}})$, which proves the Lindberg condition.

\end{proof}

\subsection{Proof of Thm.~\ref{ThmMain}}\label{SecPreProof}

\begin{proof}[Proof of Thm.~\ref{ThmMain}]

Let $P_{t}^{(\lambda),\prime}$ be defined as in Lem.~\ref{ChangeDrift}.   By Lem.~\ref{ChangeDrift} the difference $P_{t}^{(\lambda)}-P_{t}^{(\lambda),\prime}$ converges to zero with respect to the uniform metric over the interval $t\in [0,T]$ as $\lambda\searrow 0$.  Thus we can work with $P_{t}^{(\lambda),\prime}$ rather than $P_{t}^{(\lambda)}$.  Define the map $\mathcal{G}:L^{\infty} ([0,T])\rightarrow L^{\infty}([0,T])$  given by
$$ \hspace{2cm} \mathcal{G}(h)_{t}=h_{t}-\frac{1}{2}\int_{0}^{t}dr e^{-\frac{1}{2}(t-r)}h_{r},  \hspace{1cm} h\in L^{\infty}([0,T]).    $$
Notice that the solution $\frak{p}_{t}$ to the Langevin equation~(\ref{TheLimit}) has the explicit solution
\begin{align}\label{SolRemark}
 \frak{p}_{t}=\mathcal{G}(\mathbf{B})_{t},   
 \end{align}
where we have assumed $\frak{p}_{0}=0$.  Moreover, the integral equation~(\ref{LinDrift}) for $P_{t}^{(\lambda),\prime}$ admits the closed form 
\begin{align}\label{Florida}
P_{t}^{(\lambda),\prime}= e^{-\frac{1}{2}t}\lambda^{\frac{1}{2} } P_{0}+\mathcal{G}(\lambda^{\frac{1}{2}}D_{\frac{\cdot}{\lambda}})_{t}+   \mathcal{G}(M^{(\lambda)})_{t}.
\end{align}

 By our assumption (2) of List~\ref{Assumptions}, the moment $\mathbb{E}^{(\lambda)}[|P_{0}|]$ is finite, and thus the first term on the right side of~(\ref{Florida}) converges in probability to zero as $\lambda\searrow 0$.  The random variable  $\sup_{0\leq t\leq T}\big|\mathcal{G}(\lambda^{\frac{1}{2}}D_{\frac{\cdot}{\lambda}})_{t}\big|$ converges in probability to zero also because 
$$ \mathbb{E}^{(\lambda)}\Big[\sup_{0\leq t\leq T}\big|\mathcal{G}(\lambda^{\frac{1}{2}}D_{\frac{\cdot}{\lambda}})_{t}\big|\Big]\leq 2\mathbb{E}^{(\lambda)}\Big[\sup_{0\leq t\leq T}\big|\lambda^{\frac{1}{2}}D_{\frac{t}{\lambda}}\big|\Big]=\mathit{O}(\lambda^{\frac{1}{4}}),   $$
where the order equality follows by Thm.~\ref{LemNullDrift}.  By Lem.~\ref{MartPart}  $M_{t}^{(\lambda)}$ converges in law to a standard Brownian motion $ \mathbf{B}$ with respect to the uniform metric.   Since the map $\mathcal{G}$ is continuous with respect to the supremum norm,  $\mathcal{G}(M^{(\lambda)})_{t}$ converges in law to the process $\mathcal{G}( \mathbf{B}    )_{t}$ with respect to the uniform metric.  The process $\mathcal{G}( \mathbf{B}    )_{t}$ is Ornstein-Uhlenbeck, and thus   $P_{t}^{(\lambda),\prime}$ converges in law to the Ornstein-Uhlenbeck process and  $P_{t}^{(\lambda)}$ does also.

\end{proof}

\section{Miscellaneous proofs}\label{SecMiscProof}

\subsection{Proofs for Sect.~\ref{SecNum}}\label{SecNumProofs}

\begin{proof}[Proof of Prop.~\ref{BasicsOfNum}.] \text{  } \\
Part (1):   For $t\in \R^{+}$ let $\big(\mathbf{x}_{t}(x,p),\mathbf{p}_{t}(x,p)\big)$ be the trajectory starting from the phase-space point $(x,p)$ and evolving according to the Hamiltonian $H(x,p)=\frac{1}{2}p^{2}+V(x)$.   The kernel $\mathcal{T}_{\lambda}$ satisfies the following closed integral equation:
\begin{align*}
 \mathcal{T}_{\lambda}\big(x,p;dx',dp'\big)=&\int_{0}^{\infty}dt\delta\big(\mathbf{x}_{t}(x,p)-x',\mathbf{p}_{t}(x,p)-p'\big)e^{-\int_{0}^{t}ds(1+\mathcal{E}_{\lambda}(\mathbf{p}_{s}(x,p) )) }\\ &+\int_{0}^{\infty}dt \int_{\R}dp''\mathcal{J}_{\lambda}\big( \mathbf{p}_{t}(x,p), p'' \big)\mathcal{T}_{\lambda}\big(\mathbf{x}_{t}(x,p),p'';dx',dp'\big)e^{-\int_{0}^{t}ds(1+\mathcal{E}_{\lambda}(\mathbf{p}_{s}(x,p))) } . 
\end{align*}
 We obtain a series expansion for $\mathcal{T}_{\lambda}$ by iterating the above integral equation such that the $n$th term corresponds to the event that $n-1$ collisions occur over the time interval $[0,t]$.  By considering the contribution to the transition  kernel $\mathcal{T}_{\lambda}$ resulting from  two collisions  over the mean one exponential time interval, we have the following lower bound: 
\begin{align}\label{Kremlin}
 \mathcal{T}_{\lambda}\big(& x,p;dx',dp'\big)\nonumber  \\ \geq  & dx'dp'\int_{\R}dp''\int_{0}^{\infty}dt \int_{0}^{t}dt_{2}\int_{0}^{t_{2}}dt_{1} \mathcal{J}_{\lambda}\big( \mathbf{p}_{t_{1} }(x,p), p''\big)\mathcal{J}_{\lambda}\big(  \mathbf{p}_{t_{2}-t_{1} }(\mathbf{x}_{t_{1}}(x,p),p''), \mathbf{p}_{-(t-t_{2})}(x',p') \big)\nonumber \\ &\times e^{-\int_{0}^{t-t_{2}}ds(1+\mathcal{E}_{\lambda}(\mathbf{p}_{-s}(x',p')) ) }  e^{-\int_{0}^{t_{2}-t_{1} }ds(1+\mathcal{E}_{\lambda}(\mathbf{p}_{s}(\mathbf{x}_{t_{1}}(x,p),p'') )) }e^{-\int_{0}^{t_{1}}ds(1+\mathcal{E}_{\lambda}(\mathbf{p}_{s}(x,p) ) ) }.
  \end{align}

Let $A\subset \R$ be the set of $p$ with  $ 2 l \leq \frac{1}{2}p^{2} \leq 5l$.  Let $c>0$ be the minimum of the values 
$$       \inf_{\substack{\lambda < 1\\ p^{2} \leq l,\,p'\in A } }     \mathcal{J}_{\lambda}( p,p' ),\quad  \inf_{\substack{\lambda < 1\\ p^{2} \leq l,\,p'\in A } }     \mathcal{J}_{\lambda}( p',p ),\quad \text{and}\quad  \inf_{\lambda < 1 } \int_{0}^{\infty}dt \int_{0}^{t}dt_{2}\int_{0}^{t_{2}}dt_{1}e^{-t(1+\mathcal{E}_{\lambda}(\sqrt{10 l}) ) }  .  $$
 We have defined the set $A$ to exclude the arguments $p,p'$ of $\mathcal{J}_{\lambda}$ from being close since the rates $\mathcal{J}_{\lambda}( p,p' )$ have a zero along the line $p'=p$.  If $H(x,p),H(x',p')\leq l$, the conservation of energy guarantees that $\frac{1}{2}\mathbf{p}_{t_{1} }^{2}(x,p)\leq l$ and $\frac{1}{2}\mathbf{p}_{-(t-t_{2})}^{2}(x',p') \leq l$.  Also, if $ 3l\leq \frac{1}{2}(p'')^{2}\leq 4l $, then $\mathbf{p}_{t_{2}-t_{1} }(\mathbf{x}_{t_{1}}(x,p),p'')\in A$ since the kinetic energy can not fluctuate by more than $l$ through the Hamiltonian evolution.  For all $(x,p),(x',p')$ with  $H(x,p),H(x',p')\leq l$ and all $\lambda<1$,
$$
 \mathcal{T}_{\lambda}\big(x,p;dx',dp'\big)\geq c^{3} dx'dp' \Big( \int_{\R}dp''\chi(3l\leq \frac{1}{2}(p'')^{2}\leq 4l)  \Big).
$$
The above follows by restricting the integration over $p''\in \R$ to $3l\leq \frac{1}{2}(p'')^{2}\leq 4l$. 
Thus we can take $ \mathbf{u}:= c^{3}U^{2}  \int_{\R}dp''\chi(3l\leq \frac{1}{2}(p'')^{2}\leq 4l)  $.

Now we prove the upper bound for $\mathcal{T}_{\lambda}(s,ds')$.  Notice that for $\Psi\in L^{1}(\Sigma)\cap L^{\infty}(\Sigma)$, then  $    \| \mathcal{T}_{\lambda}(\Psi)\|_{\infty}  \leq \|\Psi\|_{\infty} $.  In other words, $\mathcal{T}_{\lambda}$ is a contraction in the supremum norm.  This is evident from the resolvent form $\mathcal{T}_{\lambda}=\int_{0}^{\infty}dt e^{-t}\Phi_{t,\lambda}$, where $\Phi_{t,\lambda}$ are the dynamical maps for the master  equation~(\ref{ReTheModel}), and by the inequalities
$$\big\| \mathcal{T}_{\lambda}\Psi \big\|_{\infty}\leq \int_{0}^{\infty}dt e^{-t}\|\Phi_{t,\lambda}(\Psi) \big\|_{\infty}\leq \|\Psi\|_{\infty}.$$
  The maps $\Phi_{t,\lambda}$ are contractive in the supremum norm since the dynamics is driven by an Hamiltonian flow, which preserves the supremum norm, and a noise satisfying the detailed balance condition 
  $$ e^{-\frac{1}{2}p_{1}^{2}}\mathcal{J}_{\lambda}(p_{1},p_{2})= e^{-\frac{1}{2}p_{2}^{2}   }\mathcal{J}_{\lambda}(p_{2},p_{1}).$$ 

When $H(s')\neq H(s)$ then a collision must occur over the time interval $[0,\tau_{1}]$ in order for $S_{0}=s$ and $S_{\tau_{1}}=s'$. Considering the event that the first collision occurs before time $\tau_{1}$, then the strong Markov property gives the first equality below:
\begin{align}\label{Hassel}
\mathcal{T}_{\lambda}(s,ds')=\mathbb{E}^{(\lambda)}_{s}\Big[\chi(t_{1}\leq \tau_{1})\mathcal{T}_{\lambda}(S_{t_{1}},ds')   \Big]\leq \| D_{s}^{(\lambda)}\|_{\infty},  
\end{align}  
where $D_{s}^{(\lambda)}$ is the probability density of the first collision when starting from $s\in \Sigma$. The density has the closed form
$$
 D_{s}^{(\lambda)}(x',p')=\int_{0}^{\infty}dt \delta\big(\mathbf{x}_{t}(s)-x'\big) \mathcal{J}_{\lambda}\big( \mathbf{p}_{t }(s),p' \big)  e^{-\int_{0}^{t}dr\mathcal{E}_{\lambda}(\mathbf{p}_{r}(s) )  }    .
    $$
When $H(s)\geq l=1+2\sup_{x}V(x)$, the particle will revolve around the torus freely with speed $|p|\geq \big(2+2\sup_{x}V(x) \big)^{\frac{1}{2}}$.  Using the above form for $D_{s}^{(\lambda)}$:   
\begin{align*}
 D_{s}^{(\lambda)}(x',p')  &\leq \Big(\sup_{p,p'\in \R} \frac{ \mathcal{J}_{\lambda}\big(p , p'\big)  }{ \mathcal{E}_{\lambda}(p )  } \Big)\int_{0}^{\infty}dt\delta\big(\mathbf{x}_{t}(s)-x'\big)   e^{-\int_{0}^{t}dr\mathcal{E}_{\lambda}(\mathbf{p}_{r}(s) )  }   \mathcal{E}_{\lambda}(\mathbf{p}_{t}(s) )\\ &\leq  \Big(\sup_{p,p'\in \R} \frac{ \mathcal{J}_{\lambda}\big(p , p'\big)  }{ \mathcal{E}_{\lambda}(p )  } \Big) \Big(\sum_{n=1}^{\infty} \frac{\mathcal{E}_{\lambda}(q(s,x')  )     }{q(s,x')  }     e^{-n\int_{\mathbb{T}}da\frac{\mathcal{E}_{\lambda}(q(s,a)  )     }{q(s,a)  }      }    \Big)\\ &\leq \Big(\sup_{p,p'\in \R} \frac{ \mathcal{J}_{\lambda}\big(p, p' \big)  }{ \mathcal{E}_{\lambda}(p )  } \Big) \Big(\frac{ \frac{\mathcal{E}_{\lambda}(q(s,x')  )     }{q(s,x')  }  }{1-   e^{-\int_{\mathbb{T}}da\frac{\mathcal{E}_{\lambda}(q(s,a)  )     }{q(s,a)  }      }  }  \Big),
\end{align*}
where $q(s,a):=2^{\frac{1}{2}}( H(s)-V(a) )^{\frac{1}{2}}$.   The two terms in the product on the right are bounded uniformly for all $\lambda<1$ and $s\in \Sigma $ with $H(s)> l$. \vspace{.5cm}

\noindent Part (2):  This follows easily from the construction of the split process. \vspace{.5cm}

\noindent  Part (3):   It is almost surely true that a collision will not occur at the partition time $\mathbf{t}$.  Since the trajectory $S_{t}$ is continuous between collision times $\lim_{t\nearrow \mathbf{t}}S_{t}=S_{\mathbf{t}}$.  Thus, information about the state $S_{\mathbf{t}}$ will be contained in the $\sigma$-algebra $\tilde{\mathcal{F}}_{\mathbf{t}^{-}}$.  We claim that the  probability  of the binary component $ Z_{\mathbf{t}}$ being $1$ given  $\tilde{\mathcal{F}}_{\mathbf{t}^{-}}$  has probability $h(S_{\mathbf{t}})$, which would imply that $\tilde{S}_{\mathbf{t}}$ has distribution $\tilde{\delta}_{S_{\mathbf{t}}}$  given  $\mathcal{F}_{\mathbf{t}^{-}}$.

To verify that $\tilde{\mathbb{P}}^{(\lambda)}[Z_{\mathbf{t}}=1|\tilde{\mathcal{F}}_{\mathbf{t}^{-}}]=h(S_{\mathbf{t}})$ as claimed above,  let $ \mathbf{t}' $ be the partition time preceding $\mathbf{t}$.  By the strong Markov property at the time $ \mathbf{t}' $ and since $Z_{t}$ is constant over the interval $t\in  [\mathbf{t}',\mathbf{t} )$, we have the first equality below: 
\begin{align}\label{TryMe}
\tilde{\mathbb{P}}^{(\lambda)}[Z_{\mathbf{t}}=1\,|\,\tilde{\mathcal{F}}_{\mathbf{t}^{-}}]=\tilde{\mathbb{P}}^{(\lambda)}\big[Z_{\mathbf{t}}=1\,|\, (S_{r}; r\in [\mathbf{t}',\mathbf{t} ]),\,Z_{\mathbf{t}'} \big]=\tilde{\mathbb{P}}^{(\lambda)}\big[Z_{\mathbf{t}}=1\,|\,  S_{\mathbf{t}'},\, S_{\mathbf{t}},\,Z_{\mathbf{t}'} \big]=  h(S_{\mathbf{t}}).
\end{align}
The second equality holds since the distribution of the bridge $(S_{r}; r\in (\mathbf{t}',\mathbf{t} ) )$ is independent of $Z_{\mathbf{t}'}$ and $Z_{\mathbf{t}}$ given the endpoints $ S_{\mathbf{t}'}$,  $S_{\mathbf{t}}$.  The probability $\tilde{\mathbb{P}}^{(\lambda)}\big[Z_{\mathbf{t}}=1\,|\,  S_{\mathbf{t}'},\, S_{\mathbf{t}},\,Z_{\mathbf{t}'} \big]$ is well-defined within the context of the split resolvent chain, and   the last equality in~(\ref{TryMe}) follows from the form of the transition kernel $\tilde{\mathcal{T}}_{\lambda}(s_{1}, z_{1}; ds_{2},z_{2})$, defined above~(\ref{SplitMeasure}), for  $s_{1}=S_{\mathbf{t}'}$, $ z_{1}=Z_{\mathbf{t}'}$, $s_{2}=S_{\mathbf{t}}$, and $ z_{2}=Z_{\mathbf{t}}$.

\end{proof}

\subsection{Proofs for Propositions~\ref{AMinus} and~\ref{PreStuff} }\label{SecEnergyLemProof}

We will begin with the proof of Prop.~\ref{PreStuff} since the proof of Part (1) of Prop.~\ref{AMinus} depends on it.  

\begin{proof}[Proof of Prop.~\ref{PreStuff}] To ease the notations, we denote $g(q)=(2\pi)^{-\frac{1}{2}}e^{-\frac{q^2}{2}}$.  Since $\mathcal{E}_\lambda(p) $, $\mathcal{D}_\lambda(p) $,  and $\Pi_\lambda^{(2m)}(p)$
are even functions, we will assume without loss of generality that $p>0$.\vspace{.5cm}

\noindent Part (1):
After the change of variables $q=\frac{\lambda+1}{2}p'+\frac{\lambda-1}{2}p=\frac{\lambda+1}{2}(p'-\p)+\lambda p$, we have
\begin{align}
\mathcal{E}_\lambda(p) = \frac{2\eta}{\lambda +1}\int_\R dq|q-\lambda p|g(q)
=\frac{2\eta}{\lambda+1}\left(2\int_{\lambda p}^\infty qg(q)dq+\lambda p\int_{-\lambda p}^{+\lambda p}dqg(q)\right).\label{eq-elambda-taylor}
\end{align}
We have $\int_{\lambda p }^\infty qg(q) dq = g(\lambda p)$, and $\int_{-\lambda p}^{+ \lambda p}dqg(q) \leq \min(1,\lambda|p|)$. By calculus we also see that $\alpha\mapsto g(\alpha) + \alpha \int_0^\alpha dq g(q)$ has a minimum over $\R$ at $0$. With the above remarks and $\eta=\frac{(2\pi)^{\frac{1}{2}}}{32}$,
\begin{equation} \label{eq-elambda-taylor2}
 \frac{1}{8(\lambda+1)}=\mathcal{E}_\lambda(0) \leq \mathcal{E}_\lambda(p)
\leq \frac{4\eta g(0)+2\eta\min(\lambda |p|,\lambda^2p^2)}{(\lambda+1)}
=\frac{1+C\min(\lambda|p|,\lambda^2|p|^2)}{8(\lambda+1)}.
\end{equation}\vspace{.5cm}

\noindent Part (2):  We use the same technique to estimate $\mathcal{D}_\lambda(p)$.  Here we find
\begin{align}
\mathcal{D}_\lambda(p) &=\frac{4\eta}{(\lambda+1)^2}\Big(-\int_{-\lambda p}^{+\lambda p} q^2 g(q) dq -4 \lambda p\int_{\lambda p}^\infty qg(q)dq - \lambda^2 p^2\int_{-\lambda p}^{+\lambda p}g(q)dq\Big)\nonumber \\
&=-\frac{8\eta\lambda p}{(\lambda+1)^{2}}g(\lambda p) -\frac{4\eta\lambda p}{(\lambda+1)^{2}}\int_{-\lambda p}^{+\lambda p}g(q)dq-\frac{4\eta \lambda^2 p^2}{(\lambda+1)^2}  \int_{-\lambda p}^{+\lambda p}g(q)dq.\label{eq-dlambda-taylor}
\end{align}
It follows that for $\lambda p\ll 1$
$$ \mathcal{D}_\lambda(p) = -\frac{ \lambda p}{2(\lambda+1)^2} + O(\lambda^2 p^2), $$
so $|\mathcal{D}_\lambda(p)+\frac{\lambda p}{2}|$ is bounded by a constant multiple of $\lambda^{2}(|p|+|p|^{2})$ in that regime.   When  $\lambda p$ is not small,  we can  also find a bound  for $|\mathcal{D}_\lambda(p)+\frac{\lambda p}{2}|$ of the same form through~(\ref{eq-dlambda-taylor}) using that $\int_{\R}g(q)dq=1$ and $g(q)\leq 1$.

\vspace{.5cm}

\noindent Part (3):  Now we repeat the computation for $\mathcal{Q}_\lambda (p)$ and we get
\begin{align}
\mathcal{Q}_\lambda (p) 
&=\frac{8\eta}{(\lambda+1)^3}\bigg( (4+2\lambda^{2}p^{2})    g(\lambda p) + \lambda p (3+\lambda^2 p^2)\int_{-\lambda p}^{+\lambda p} g(q) dq\bigg)\label{pi-lambda-lambda2-p2}\\
&=\frac{8}{(\lambda+1)^2} \mathcal{E}_\lambda(p) +\frac{16\eta\lambda^{2}p^{2}}{(\lambda+1)^3} g(\lambda p)+ \frac{8\eta\lambda p(1+ \lambda^{2} p^{2} )}{(\lambda+1)^3}  \int_{-\lambda p}^{+\lambda p} g(q) dq.\nonumber
\end{align}
With the above
\begin{align*}
	\Big|\mathcal{Q}_\lambda (p) -\frac{1}{(\lambda+1)^3 }  \Big|\leq  \frac{1}{(\lambda+1)^3 } \Big|1- 8 \mathcal{E}_\lambda(p)   \Big|+\frac{16\eta\lambda^{2}p^{2}}{(\lambda+1)^3} g(\lambda p)+ \frac{8\eta\lambda p(1+ \lambda^{2} p^{2} )}{(\lambda+1)^3} \int_{-\lambda p}^{+\lambda p} g(q) dq . 
\end{align*}
With the bounds for  $\mathcal{E}_\lambda(p) $  from Part (1), the right side above is bounded by a constant multiple of  $\lambda p+(\lambda p)^{3}$.  

\vspace{.5cm}

\noindent Part (4): Finally, reasoning as above, it is easy to produce an upper bound for $\Pi_\lambda^{(2m)}(p)$ which is a polynomial of degree $2m+1$ in $\lambda|p|$.

\end{proof}

\begin{proof}[Proof of Proposition~\ref{AMinus}]
 Our first observation is that $\mathcal{A}_\lambda(x,p)$,  $\mathcal{K}_{\lambda,n}(x,p)$,  and $\mathcal{V}_\lambda(x,p)$ are symmetric in $p\in \R$.  Hence we can assume without loss of generality that $p>0$. \vspace{.5cm}

\noindent Part (1): 
 We note that
\begin{align}\label{IntroGamma}
2^{\frac{1}{2}} H^{\frac{1}{2}}(x,p') -2^{\frac{1}{2}} H^{\frac{1}{2}}(x,p) &= \frac{ p'^2- p^2}{2^{\frac{1}{2}}\left(\frac{p'^2}{2}+V(x)\right)^{\frac{1}{2}}+2^{\frac{1}{2}}\left(\frac{p^2}{2}+V(x)\right)^{\frac{1}{2}}}
=(p'-p)\Gamma(p',p),
\end{align}
where $\Gamma(p',p):=2^{-\frac{1}{2}}(p'+p)\big( H^{\frac{1}{2}}(x,p')+ H^{\frac{1}{2}}(x,p)\big)^{-1}$.
Thus one can write
\begin{align}\label{alhambra}
\mathcal{A}_\lambda(x,p) = \mathcal{D}_\lambda(p) -  \int_\R dp'(p'-p)(1-\Gamma(p',p))\mathcal{J}_\lambda(p,p').
\end{align}
For the proof of $\mathcal{A}_\lambda^-(x,p)\leq |\mathcal{D}_\lambda(p)|$,  we have to show that the integral in (\ref{alhambra}) is non-positive. For this we use the monotonicity of $\Gamma(p,p')$ in $p'$ for fixed $p$ and the following straightforward inequality is valid for all $p,r\geq 0$:
$$ \mathcal{J}_\lambda(p,p+r)\leq \mathcal{J}_\lambda(p,p-r). $$
Then we have
\begin{align*}
\int_\R dp'&(p'-p)(1-\Gamma(p,p'))\mathcal{J}_\lambda(p,p') \\
&= \int_0^\infty rdr(1-\Gamma(p,p+r))\mathcal{J}_\lambda(p,p+r)
+ \int_0^\infty rdr(\Gamma(p,p-r)-1)\mathcal{J}_\lambda(p,p-r) \\
&\leq \int_0^\infty rdr(\Gamma(p,p-r)-\Gamma(p,p+r))\mathcal{J}_\lambda(p,p-r) \leq 0.
\end{align*}
Finally, we know as a consequence of  Part (2) of  Proposition \ref{PreStuff} that there is a $C'>0$ such that 
$$  \left|\mathcal{D}_\lambda(p)\right| \leq  C'\lambda| p|+ C'\lambda^2 p^2 
.
$$

\vspace{.5cm}
\noindent Part (2):  Now we prove the bound for $\mathcal{A}_\lambda^+(x,p)$,  which will  follow by finding an upper bound for  $\mathcal{A}_\lambda (x,p)$.    For $\tilde{p}=\frac{1-\lambda}{1+\lambda}p$ we can write $p-\tilde{p}=\frac{2\lambda p}{\lambda+1}$ and
\begin{align}\label{MathA}
\mathcal{A}_\lambda(x,p) &= \int_\R dp'(p'-\tilde{p})\Gamma(p,p')\mathcal{J}_\lambda(p,p') 
- \frac{2\lambda p}{\lambda +1} \int_\R dp'\Gamma(p,p')\mathcal{J}_\lambda(p,p').
\end{align}
For the first term on the right side of~(\ref{MathA}), we have
\begin{align*}
- \frac{2\lambda p}{\lambda +1} \int_\R dp'\Gamma(p,p')\mathcal{J}_\lambda(p,p')
&\leq - \frac{2\lambda p}{\lambda +1} \int_{-\infty}^{-p} dp'\Gamma(p,p')\mathcal{J}_\lambda(p,p')
\leq \frac{2\lambda p}{\lambda +1} \int_{-\infty}^{-p} dp' \mathcal{J}_\lambda(p,p'),
\end{align*}
where we have simply thrown away the negative part of the integration to get the first inequality.   The above is exponentially decreasing in $p$ (uniformly in $\lambda$).

To bound the second term on the right side of~(\ref{MathA}),  we begin with a bound for $ \frac{\partial\Gamma(p,p')}{\partial p'}$.   There is $C>0$ such that for all $p>0$ and $\frac{p}{2}\leq p'\leq \frac{3p}{2}$
\begin{equation} \label{gamma-derivative}
 0 \leq \frac{\partial\Gamma(p,p')}{\partial p'}=\frac{(p^2+2V(x))^\frac{1}{2}(p'^2+2V(x))^\frac{1}{2} - pp' +2V(x)}{(p'^2+2V(x))^{\frac{1}{2}}\left((p^2+2V(x))^{\frac{1}{2}}+(p'^2+2V(x))^{\frac{1}{2}}\right)^2} \leq \frac{C}{1+p^3}.
\end{equation}
By writing $r=p'-\tilde{p}$ and $g(q)=(2\pi)^{-\frac{1}{2}}e^{-\frac{q^2}{2}}$, the first term on the right side of~(\ref{MathA}) is bounded by
\begin{align*}
 \int_\R dp'(p'-\tilde{p})&\Gamma(p,p')\mathcal{J}_\lambda(p,p')\\
=& \frac{\eta(\lambda+1)}{2}\int_0^\infty rdr \left(\Gamma(p,\tilde{p}+r)\left|r -\frac{2\lambda p}{\lambda +1}\right| - \Gamma(p,\tilde{p}-r)\left|r+\frac{2\lambda p}{\lambda+1}\right|\right)g\left(\frac{\lambda +1}{2} r\right) \\
\leq & \frac{\eta(\lambda+1)}{2} \int_0^{\frac{2\lambda p}{\lambda +1}} rdr \frac{2\lambda p}{\lambda +1}\left(\Gamma(p,\tilde{p}+r)-\Gamma(p,\tilde{p}-r)\right) g\left(\frac{\lambda +1}{2}r\right) \\
&+ \frac{\eta(\lambda+1)}{2} \int_{\frac{2\lambda p}{\lambda +1}}^\infty r^2dr\left(\Gamma(p,\tilde{p}+r)-\Gamma(p,\tilde{p}-r)\right) g\left(\frac{\lambda +1}{2}r\right).
\end{align*}
For the inequality above, we used the fact that $\Gamma(p,\tilde{p}+r)+\Gamma(p,\tilde{p}-r)\geq 0$, which is true because $\Gamma(p,\tilde{p}-r)\geq \Gamma(p,-\tilde{p}-r)$ and $|\Gamma(p,-\tilde{p}-r)|\leq \Gamma(p,\tilde{p}+r)$ for all $r\geq 0$. Now, for the first integral, we can use (\ref{gamma-derivative}) and we find that it is bounded by $\frac{C \lambda}{1+p^2}$. For the second integral, we first observe that the integral over $r\geq \frac{\tilde{p}}{2}$ is decaying exponentially, and for the integral over $0\leq r\leq \frac{\tilde{p}}{2}$, we again use that (\ref{gamma-derivative}) and get a $\frac{C}{1+p^3}$ bound. \vspace{.5cm}

\noindent Part (3):   For $\lambda=0$ we have that $\mathcal{J}_0(p,p') =\frac{1}{64}|p'-p|e^{-\frac{(p'-p)^{2}}{8}}$.  
Since $H^{\frac{1}{2}}(x,p)$ is a convex function in $p\in \R$ for each $x \in \mathbb{T}$, we have 
$$\mathcal{A}_0(x,p)=\frac{1}{64} \int_{\R}dp'\Big(  2^{\frac{1}{2}} H^{\frac{1}{2}}(x,p')-2^{\frac{1}{2}} H^{\frac{1}{2}}(x,p) \Big) |p'-p|e^{-\frac{(p'-p)^{2}}{8}}\geq 0   . $$
In other terms, $\mathcal{A}_0^+(x,p)=\mathcal{A}_0 (x,p)$.   Our first task will be to establish that $\int_{\R}dp \mathcal{A}_0^+(x,p)=1$ for each $x$.    By Part (2) we also know that   $\mathcal{A}_0^+(x,p)$ has a bound of the form $\frac{C}{1+p^{2}}$, so $\int_{|p|\geq L }dp\mathcal{A}_0^+(x,p)=\mathit{O}(L^{-1})  $ for large $L$.  Hence, $\int_{\R }dp\mathcal{A}_0^+(x,p)$ can be approximated by
\begin{align}\label{Moses}
\int_{|p|\leq  L }dp\mathcal{A}_0^+(x,p)=& \frac{1}{64} \int_{|p|\leq L }dp \int_{\R}dp'\Big(  2^{\frac{1}{2}} H^{\frac{1}{2}}(x,p')-2^{\frac{1}{2}} H^{\frac{1}{2}}(x,p) \Big) |p'-p|e^{-\frac{(p'-p)^{2}}{8}} \nonumber  \\  = & \frac{1}{64} \int_{|p|\leq L }dp \int_{|p'|\geq L}dp \Gamma(p,p') (p'-p) |p'-p|e^{-\frac{(p'-p)^{2}}{8}}
\nonumber   \\  =& \frac{1}{64} \int_{|p|\leq L }dp \int_{|p'|\geq L}dp' (p'-p)^2e^{-\frac{(p'-p)^{2}}{8}} \nonumber \\ & + \frac{1}{64} \int_{|p|\leq L }dp \int_{|p'|\geq L}dp'\big(1-\Gamma(p,p')\big) (p'-p)^2e^{-\frac{(p'-p)^{2}}{8}}.
\end{align}
The first term term on the right side of~(\ref{Moses}) satisfies
$$  \frac{1}{64} \int_{|p|\leq L }dp \int_{|p'|\geq L}dp' (p'-p)^2e^{-\frac{(p'-p)^{2}}{8}}\approx \frac{1}{32} \int_{\R^{+}}dp p^{3} e^{-\frac{p^{2}}{8}}=1,$$
where the error of the  approximation is exponentially small  for $L\gg 1$.   For the second term on the right side of~(\ref{Moses}),  
\begin{align*}
 \int_{|p|\leq L }dp \int_{|p'|\geq L}dp'\big(1-\Gamma(p,p')\big) (p'-p)^2e^{-\frac{(p'-p)^{2}}{8}}\approx  2\int_{\R^{+}}dr r^2e^{-\frac{r^{2}}{8}} \int_{ 0}^{r}dv\big(1-\Gamma(L+v-r,L+v)\big),
\end{align*}
where this approximation also has an exponentially small error.  We can see that the above expression is $\mathit{O}(L^{-2})$ by observing that  for $\frac{ p}{2}\leq p'$ 
\begin{align}\label{first-gamma-bound}
0\leq 1- \Gamma(p',p) \leq 1-\Gamma\left(\frac{p}{4},\frac{p}{4}\right)
= 1- \frac{p}{\left(p^2+32V(x)\right)^{\frac{1}{2}}}=\mathit{O}(p^{-2}).
\end{align}

Next we focus on  bounding the difference between $\int_{\R}dp\mathcal{A}_\lambda^{+}(x,p) $ and $\int_{\R}dp\mathcal{A}_0^{+}(x,p) =1$  in the small  $\lambda$ limit.    By the analysis at the beginning of the proof of Part (2), we have the first  equality below:
\begin{align}\label{Ezekial}
\int_{\R}dp\mathcal{A}_\lambda^{+}(x,p)=& \int_{\R} dp \Big[   \int_\R dp'(p'-\tilde{p})\Gamma(p,p')\mathcal{J}_\lambda(p,p')      \Big]_{+}+\mathit{O}(\lambda),\nonumber \\ =& \int_{-\lambda^{-\frac{1}{2}} }^{\lambda^{-\frac{1}{2}}} dp \Big[   \int_\R dp'(p'-\tilde{p})\Gamma(p,p')\mathcal{J}_\lambda(p,p')      \Big]_{+}+\mathit{O}(\lambda^{\frac{1}{2}})
\nonumber \\ =& \int_{-\lambda^{-\frac{1}{2}}}^{\lambda^{-\frac{1}{2}}} dp  \mathcal{A}_0^{+}(x,p)      +\mathit{O}(\lambda^{\frac{1}{2}})\nonumber  \\ =& \int_{\R} dp\mathcal{A}_0^{+}(x,p)+\mathit{O}(\lambda^{\frac{1}{2}}).
\end{align}
In the above  $[\cdot ]_{+}$  refers to the positive part of a function.  The second equality in~(\ref{Ezekial}) uses that the integrand for the integration in $p$ is bounded by a multiple of $\frac{1}{1+p^2}$ by the analysis in Part (2).   For the third equality above, we have used that   $\mathcal{A}_0^{+}(x,p)=  \int_\R dp'(p'-\tilde{p})\Gamma(p,p')\mathcal{J}_0(p,p') $ is positive, $|\Gamma(p,p')|\leq 1$, and   
\begin{align*}
\Big|\int_{-\lambda^{-\frac{1}{2}} }^{\lambda^{-\frac{1}{2}}} dp\Big( & \Big[   \int_\R dp'(p'-\tilde{p})\Gamma(p,p')\mathcal{J}_\lambda(p,p')      \Big]_{+} -\mathcal{A}_0^{+}(x,p)\Big)\Big|\\  &\leq \int_{-\lambda^{-\frac{1}{2}} }^{\lambda^{-\frac{1}{2}}} dp   \int_\R dp'|p'-\tilde{p}| \Big|\mathcal{J}_\lambda(p,p')      -\mathcal{J}_0(p,p')     \Big|\\ &= \mathit{O}(\lambda^{\frac{1}{2}}).
\end{align*}

\vspace{.5cm}

\noindent Part (4):  The bounds for $\mathcal{K}_{\lambda,n}$ are straightforward since by~(\ref{IntroGamma}) and $|\Gamma(p',p)|\leq 1$  we have
$$ \mathcal{K}_{\lambda,n}(x,p) \leq 2^{-\frac{n}{2}} \int_\R dp' |p-p'|^n \mathcal{J}_\lambda(p,p'). $$
 Writing $r=p'-p$ there is a constant $C_{n}'$ such that for all $\lambda<1$ and $(x,p)\in \Sigma$
\begin{align*}
\mathcal{K}_{\lambda,n}(x,p) &\leq C_n' \left(\int_0^\infty r^{n+1} dr g\Big(\frac{\lambda+1}{2} r +\lambda p\Big) +\int^{\infty}_{2p}r^{n+1} dr g\Big(\frac{\lambda+1}{2} r -\lambda p\Big)\right)\\
&\leq C_n' \left(\int_0^\infty r^{n+1} dr g\Big(\frac{\lambda+1}{2} r\Big) + \int_{2p}^\infty r^{n+1}dr g\Big(\frac{r}{2}\Big)\right).
\end{align*}
Therefore, $\mathcal{K}_{\lambda,n}(x,p)$ is bounded by a constant.

\end{proof}

\subsection{Proofs for Sect.~\ref{SecDrift}}\label{SecProofMartDrift}

\begin{proof}[Proof of Lem.~\ref{LocalTimeBndII}]
 By Lem.~\ref{LemKeyMart} $\langle \tilde{M}\rangle_{t}$  is a sum of terms $ \check{\upsilon}_{\lambda}(S_{R_{n}})$, and so the difference between $\langle \tilde{M}\rangle_{t}$ and $\upsilon_{\lambda}\tilde{N}_{t}$ can be written
\begin{align*}
\langle \tilde{M}\rangle_{t}-\upsilon_{\lambda}\tilde{N}_{t}=& \sum_{n=1}^{\tilde{N}_{t}} \big(\check{\upsilon}_{\lambda}(S_{R_{n}}) -\upsilon_{\lambda}\big)  = \check{\upsilon}_{\lambda}(S_{R_{1}})-\check{\upsilon}_{\lambda}(S_{R_{\tilde{N}_{t}+1}})-\tilde{\mathbb{E}}^{(\lambda)}_{\tilde{\delta}_{S_{R_{\tilde{N}_{t}+1}} }} \big[  \check{\upsilon}_{\lambda}(S_{R_{1}})   \big]  +\tilde{\mathbb{E}}^{(\lambda)}_{\tilde{\delta}_{S_{R_{1}}} } \big[  \check{\upsilon}_{\lambda}(S_{R_{1}})   \big]   \\ &+\sum_{n=1}^{\tilde{N}_{t}} \big(\check{\upsilon}_{\lambda}(S_{R_{n+1}}) -\tilde{\mathbb{E}}^{(\lambda)}_{\tilde{\delta}_{S_{R_{n}} }} \big[  \check{\upsilon}_{\lambda}(S_{R_{1}})   \big]+\tilde{\mathbb{E}}^{(\lambda)}_{ \tilde{\delta}_{S_{R_{n+1}} } } \big[  \check{\upsilon}_{\lambda}(S_{R_{1}})   \big] -\upsilon_{\lambda} \big),
\end{align*}
where $\tilde{\delta}_{s}$ is the splitting of the $\delta$-distribution at $s\in \Sigma$.  Notice that $\upsilon_{\lambda}=\int_{\Sigma}d\nu(s)\check{\upsilon}_{\lambda}(s)  $.  The  sum on the right is a martingale with respect to $\tilde{\mathcal{F}}_{t}'$ by the same reasoning that $\tilde{M}_{t}$ is a martingale.  Since $S_{R_{n}}\in\textup{Supp}(\nu)$ for $n\geq 1$,  we  have the standard inequalities:
\begin{align*}
\lambda^{\frac{1}{2}}\tilde{\mathbb{E}}^{(\lambda)}\Big[&\sup_{0\leq t\leq \frac{T}{\lambda} }\Big|  \langle \tilde{M}\rangle_{t}-\upsilon_{\lambda}\tilde{N}_{t}\Big|\Big] \leq 4 \lambda^{\frac{1}{2}}\sup_{s\in \textup{Supp}(\nu) }\check{\upsilon}_{\lambda}(s)\\&+2\lambda^{\frac{1}{2}}\tilde{\mathbb{E}}^{(\lambda)}\Big[\sum_{n=1}^{\tilde{N}_{\frac{T}{\lambda} }} \Big(\check{\upsilon}_{\lambda}(S_{R_{n+1}}) -\tilde{\mathbb{E}}^{(\lambda)}_{\tilde{\delta}_{S_{R_{n}}} } \big[  \check{\upsilon}_{\lambda}(S_{R_{1}})   \big]+\tilde{\mathbb{E}}^{(\lambda)}_{\tilde{\delta}_{S_{R_{n+1}} } } \big[  \check{\upsilon}_{\lambda}(S_{R_{1}})   \big] -\upsilon_{\lambda} \Big)^{2}\Big]^{\frac{1}{2}}\\ \leq & 4 \lambda^{\frac{1}{2}}\sup_{s\in \textup{Supp}(\nu) }\check{\upsilon}_{\lambda}(s)+2^{\frac{3}{2}}\lambda^{\frac{1}{2}}\tilde{\mathbb{E}}^{(\lambda)}\big[\tilde{N}_{\frac{T}{\lambda} }\big]^{\frac{1}{2}}\Big(\int_{\Sigma}ds\check{\upsilon}^{2}_{\lambda}(s  )-\Big(\int_{\Sigma}ds\check{\upsilon}_{\lambda}(s  )\Big)^{2} \Big)^{\frac{1}{2}}=\mathit{O}(\lambda^{\frac{1}{4}}).
\end{align*}

\end{proof}

\section*{Acknowledgments}
We are grateful to Professor H\"opfner for sending a copy of Touati's unpublished paper and giving helpful comments.   We also thank an anonymous referee for offering many useful suggestions towards improving the presentation of this article.    This work is supported by the European Research Council grant No. 227772 and NSF grant DMS-08446325.

\begin{appendix}

\section{Exponential ergodicity}\label{AppendErgod}

In this section we  prove the relaxation of our dynamics to an equilibrium state.  Our dynamics is driven by the forward Kolmogorov equation
\begin{align}\label{ReReTheModel}
\frac{d}{dt}\tuP_{t,\lambda}(x,\,p)= &  (\mathcal{L}_{\lambda}^{*}\tuP_{t,\lambda})(x,\,p)= -p \frac{\partial}{\partial x}\tuP_{t,\lambda}(x,p)+\frac{dV}{dx}\big(x\big)\frac{\partial}{\partial p}\tuP_{t,\lambda}(x,p) \nonumber    \\ &+\int_{\R}dp^{\prime}\big(\mcJ_{\lambda}(p^{\prime},p)\tuP_{t,\lambda}(x,p^{\prime})-\mcJ_{\lambda}(p,p^{\prime})\tuP_{t,\lambda}(x,p)     \big),     
\end{align}
which has equilibrium state $
\Psi_{\infty,\lambda}$. The Kolmogorov equation~(\ref{ReReTheModel}) determines a transition semigroup $\Phi_{t,\lambda}$, which we take to operate on measures from the right and bounded measurable functions from the left:
$$(\mu \,\Phi_{t,\lambda})(ds)=\int_{\Sigma}\mu(ds')\Phi_{t,\lambda}(s',ds)\quad \text{and}\quad  (\Phi_{t,\lambda}\,g)(s)=\int_{\Sigma}\Phi_{t,\lambda}(s,ds')g(s'),        $$
for $\mu\in M(\Sigma)$ and $g\in B(\Sigma)$.  We identify probability densities with their corresponding measures and denote the total variation norm for finite measures on $\Sigma$ by $\|\cdot\|_{1}$.  The exponential ergodicity that we prove in Theorem~\ref{ThmErgodic} is not used critically anywhere in our proofs although we need some degree of ergodicity in order to make sense of certain expressions such as the reduced resolvent $\frak{R}_{\lambda}$ of the function $\frac{dV}{dx}$, for instance.  A more thorough study of the ergodicity would give some control of the the exponential rate $\alpha(\lambda)$ appearing in Theorem~\ref{ThmErgodic} for $\lambda\ll 1$.   We believe that  $\alpha(\lambda)$  can be taken to scale as $\propto\frac{\lambda}{\log(\lambda^{-1})}$ for small $\lambda$ (i.e. a bit slower than linear in $\lambda$), and if $C$ is replaced on the right side of~(\ref{ShoeHorn}) by $C\|\Psi\|_{w}$ for  the weighted norm  $\|\Psi\|_{w}=\int_{\Sigma}|\Psi|(dx\,dp)(1+\lambda^{\frac{1}{2}}|p|)\, $, then the exponential rate scales as $\propto \lambda$ for $\lambda\ll 1$.

We denote the space of probability measures on $\Sigma$ by $M_{+}^{1}(\Sigma)$.
\begin{theorem}\label{ThmErgodic}
Let $ \Phi_{t,\lambda}$ be the transition semigroup corresponding to the Kolmogorov equation~(\ref{ReReTheModel}).  There exist  $\alpha(\lambda), C>0$ such that for all $\Psi\in M_{+}^{1}(\Sigma) $ and $t\in \R^+$, 
\begin{align}\label{ShoeHorn}
\big\|\Psi\Phi_{t,\lambda}-\Psi_{\infty,\lambda}\big\|_{1}\leq Ce^{-t\alpha(\lambda) }.
\end{align}

\end{theorem}

\begin{proof}
Since the parameter $\lambda>0$ is fixed in this proof, we will not attach it as a subscript or superscript to mathematical objects  that depend upon it as we did in earlier sections.  It is sufficient to work with the resolvent chain rather than the original process and show that  there are $C,\alpha>0$
\begin{align}\label{Airlift}
 Ce^{-n\alpha } &\geq \big\|\Psi \mathcal{T}^{n}-\Psi_{\infty} \big\|_{1} \nonumber \\ 
&= \sup_{\|g\|_{\infty}\leq 1  } \Big|\int_{\Sigma}(\Psi \mathcal{T}^{n})(ds)g(s)  -\Psi_{\infty}(g)  \Big|,
\end{align}
where $\mathcal{T}:B(\Sigma)\rightarrow B(\Sigma)$ is the  transition operator for the resolvent chain and has the form
$$ \mathcal{T}=\int_{0}^{\infty}dre^{-r}\Phi_{r}.$$

Let $\mu_{L}\geq 0$  be defined as the Lebesgue measure of  the region  $\{ H(s)\leq L\}\subset \Sigma$ for $L>0$.  There exist $L>0$ and $0<\epsilon<1$ such that the following statements hold:
\begin{enumerate}[(i).]  

\item For all $s,s'\in \Sigma$ with $H(s),H(s')\leq L$, 
$$ \mathcal{T}(s,ds')>\epsilon\,\frac{ ds'}{  \mu_{L} }. $$

\item  For all $s\in \Sigma$ with $H(s)> L$,
$$  \int_{H(s')\leq L}\mathcal{T}(s,ds') \geq \epsilon .  $$

\end{enumerate}
  Given any $L>0$, we can find an $\epsilon$ such that statement (i) is true by the argument in the proof of Part (1) of Proposition~\ref{BasicsOfNum} (in which the cutoff was $l=1+2\sup_{x}V(x)$ rather than an arbitrary fixed $L$, but this does not change the argument).  Statement (ii) follows from the extreme momentum-contractive nature of the jump rates:
$$  \mathcal{J}(p',p)=\frac{1+\lambda}{64}|p'-p|e^{-\frac{1}{2}(\frac{1-\lambda}{2}p'-\frac{1+\lambda}{2}p)^{2}    }   .     $$ 
A test particle with momentum $|p'|\gg 1$ will suffer a collision after a waiting time on the order of $\frac{1}{|p'|}$, and the resulting momentum of the test particle will be concentrated around the contracted value $p'\frac{1-\lambda}{1+\lambda}$.  Consequently, the test particle descends from arbitrarily  high energy to the  low energy region $ \{ H(s)\leq L\}\subset \Sigma$  in finite time.   The statements (i) and (ii) imply that the probability of jumping into the region $\{ s \in \Sigma\,\big|\,H(s)\leq  L  \}$ is at least $\epsilon>0$ from any starting point.  The statements (i) and (ii) and the fact that $\Psi_{\infty,\lambda}$ is a stationary state for the dynamics are sufficient to prove the exponential ergodicity without further reference to the dynamics at hand.  The proof of (ii) is  below, and we proceed next  with the remainder of the proof.

  By (i) we have the minorization condition   $\mathcal{T}(s,ds')\geq h(s)\nu(ds')$ for 
 $$  h(s)= \epsilon\, \chi\big( H(s)\leq L   \big)  \quad \text{and} \quad \nu(ds)= \frac{ds}{ \mu_{L}}\chi\big( H(s)\leq L   \big).      $$    
 We can thus apply standard Nummelin splitting~\cite{Nummelin} to define a  chain $\tilde{\sigma}_{n}\in \tilde{\Sigma}=\Sigma\times \{0,1\}$ with an atom set  $\Sigma\times 1$ using the pair $h, \nu$ above.  Let $\tilde{n}_{m}$, $m\geq 1$ be the sequence of return times to the atom. The times $\tilde{n}_{m}$ form a delayed renewal chain in which the delay distribution is $\tilde{\mathbb{P}}_{\tilde{\Psi}}[\tilde{n}_{1}= n]  $  and the jumps have distribution $ \tilde{\mathbb{P}}_{\tilde{\nu}}[\tilde{n}_{1}= n] $; recall from~(\ref{SplitMeasure}) that $\tilde{\mu}$ refers to the splitting of a measure $\mu$ on $\Sigma$.

We denote the first component of the chain $\tilde{\sigma}_{n}\in \tilde{\Sigma}$ by $\sigma_{n}$.  
We will treat the difference between $\Psi \mathcal{T}^{n}$ and   $\Psi_{\infty}$ in the norm $\|\cdot\|_{1}$ through the formula ~(\ref{Airlift}).  For $g\in B(\Sigma)$,       
\begin{align}\label{Chief}
\int_{\Sigma}(\Psi\mathcal{T}^{n})(ds)g(s) & =\mathbb{E}_{\Psi}\big[ g(\sigma_{n})   \big]=\tilde{\mathbb{E}}_{\tilde{\Psi}}\big[ g(\sigma_{n})     \big]\nonumber  \\ &= \tilde{\mathbb{E}}_{\tilde{\Psi}}\big[ g(\sigma_{n})\chi( \tilde{n}_{1}\geq n)   \big]+
\tilde{\mathbb{E}}_{\tilde{\Psi}}\Big[\sum_{m=1}^{n}\sum_{r=1}^{\infty} \chi(m-1=\tilde{n}_{r}) g(\sigma_{n})\chi(\tilde{n}_{r+1}\geq n  )   \Big] \nonumber \\ &= 
 \tilde{\mathbb{E}}_{\tilde{\Psi}}\big[ g(\sigma_{n})\chi( \tilde{n}_{1}\geq n)   \big]+\sum_{m=1}^{n} F_{\Psi}(m-1)\, \tilde{\mathbb{E}}_{\tilde{\nu}}\big[ g(\sigma_{n-m})\chi(\tilde{n}_{1}\geq  n-m)   \big],
\end{align}
where  $F_{\Psi}:\mathbb{N}\rightarrow \R^{+} $ is the renewal function
$$F_{\Psi}(m)= \sum_{r=1}^{\infty}\tilde{\mathbb{P}}_{\tilde{\Psi}}\big[\tilde{n}_{r}=m\big],   $$
and the last equality in~(\ref{Chief}) uses the strong Markov property.    We will delay the demonstration of the following two statements until the end of the proof.  
\begin{enumerate}[(I).]

\item There is a $c>0$ such that for all  $m\in \mathbb{N}$ and probability measures $\Psi\in M_{+}^{1}(\Sigma)$, 
$$ \tilde{\mathbb{P}}_{\tilde{\Psi}}\big[\tilde{n}_{1}\geq  m \big] \leq   c e^{-\epsilon\,m } . $$

\item There is a $c>0$ such that for all $m\in \mathbb{N}$ and probability measures $\Psi\in M_{+}^{1}(\Sigma)$,
$$ \| F_{\Psi}(m)- \Psi_{\infty}(h)   \|_{\infty}\leq ce^{-\epsilon\,m}.$$

\end{enumerate}

With~(\ref{Chief}) we can write the difference between $\int_{\Sigma}(\Psi  \mathcal{T}^{n})(ds )g(s)$ and $\Psi_{\infty}(g)$ as
\begin{align}\label{Vicks}
\int_{\Sigma}(\Psi  \mathcal{T}^{n})(ds )g(s)-\Psi_{\infty}(g)= &  \tilde{\mathbb{E}}_{\tilde{\Psi}}\big[ g(\sigma_{n})\chi( \tilde{n}_{1}\geq n)   \big]\nonumber \\ &+\sum_{m=1}^{\lfloor \frac{n}{2} \rfloor}F_{\Psi}(m-1)\, \tilde{\mathbb{E}}_{\tilde{\nu}}\big[ g(\sigma_{n-m})\chi(\tilde{n}_{1}\geq  n-m)   \big]\nonumber \\ &+
\sum_{m=\lfloor \frac{n}{2} \rfloor+1  }^{n}\big(F_{\Psi}(m-1)-\Psi_{\infty,\lambda}(h) \big)\, \tilde{\mathbb{E}}_{\tilde{\nu}}\big[ g(\sigma_{n-m})\chi(\tilde{n}_{1}\geq  n-m)   \big]\nonumber \\  &+\Psi_{\infty}(h)   \Big(\sum_{m=\lfloor \frac{n}{2} \rfloor+1}^{n }\tilde{\mathbb{E}}_{\tilde{\nu}}\big[ g(\sigma_{n-m})\chi(\tilde{n}_{1}\geq  n-m)   \big] -\frac{ \Psi_{\infty}(g)      }{ \Psi_{\infty}(h)  } \Big) .
\end{align}
Notice that the bottom line of~(\ref{Vicks}) can be rewritten as $-\Psi_{\infty}(h) \sum_{m=\lfloor \frac{n}{2} \rfloor }^{\infty }\tilde{\mathbb{E}}_{\tilde{\nu}}\big[ g(\sigma_{m})\chi(\tilde{n}_{1}\geq  m)   \big]   $ 
since by   Part (2) of Proposition~\ref{BasicsOfNumII} we have the second equality below:
\begin{align}\label{Gristle}
\sum_{m=0}^{\infty}  \tilde{\mathbb{E}}_{\tilde{\nu}}\big[ g(\sigma_{m})\chi(\tilde{n}_{1}\geq  m)   \big] =  \tilde{\mathbb{E}}_{\tilde{\nu}}\Big[\sum_{m=0}^{\tilde{n}_{1}}g(\sigma_{m})  \Big]=\frac{ \Psi_{\infty}(g)      }{ \Psi_{\infty}(h)  }.
\end{align}
By applying the triangle inequality to~(\ref{Vicks}),  statements (I) and (II), and that $F_{\Psi}^{(\lambda)}(m)\leq 1$, we have the inequality
\begin{align*}
\Big|\int_{\Sigma}(\Psi  \mathcal{T}^{n})(ds )g(s)-\Psi_{\infty}(g)\Big|\leq   c\|g\|_{\infty} \Big((1-\epsilon)^{n}+\frac{e^{-\epsilon \lfloor \frac{n}{2} \rfloor     }}{\epsilon}+ \frac{e^{-\epsilon \lfloor \frac{n}{2} \rfloor     }}{\epsilon}+\frac{e^{-\epsilon \lfloor \frac{n}{2} \rfloor     }}{ \epsilon}\Psi_{\infty}(h)   \Big) .
\end{align*}
 The above is $\|g\|_{\infty}$ multiplied by $ \mathit{O}(e^{-\frac{n\epsilon}{2}})$,  so we have  proven~(\ref{Airlift}) for $\alpha=\frac{\epsilon}{2}$ when assuming statements (ii), (I), and (II).

\vspace{.5cm}

\noindent (ii).  We will return to a consideration of the process $S_{t}$ underlying the resolvent chain $\sigma_{n}$.  Let $\tau$ be a mean one exponential and $S_{0}=s\in \Sigma$.  For  $H(s)> L$ let $\varsigma$ be the hitting time that $H(S_{\varsigma})\leq L$.  Since the Hamiltonian evolution preserves energy, we can give a lower bound for our quantity of interest through
\begin{align}\label{Lambast}
\int_{H(s')\leq L}\mathcal{T}(s,ds')&\geq \mathbb{P}_{s}\big[\varsigma\leq \tau \text{ and no collisions occur over the time interval $(\varsigma,\tau]$ }   \big] \nonumber  \\  & \geq  \Big(\inf_{H(x',p')\leq L  }\frac{1}{1+\mathcal{E}(p')}     \Big) \mathbb{P}_{s}\big[\varsigma\leq \tau  \big].
\end{align}
The second inequality uses that the collisions occur with Poisson rate $\mathcal{E}(S_{t})$ and   that $\tau-\varsigma$ is a mean one exponential in the event that $\varsigma\leq \tau $ and conditioned on the information up to time $\varsigma$.   The escape rates $\mathcal{E}(S_{t})$ are  uniformly bounded over any compact region, and it is thus sufficient for us to give a lower bound for  $\mathbb{P}_{s}\big[\varsigma\leq \tau  \big] $. For any $T>0$ that we choose,  
\begin{align*}
\inf_{H(s)>L} \mathbb{P}_{s}\big[\varsigma\leq \tau  \big]&=\inf_{H(s)>L} \int_{0}^{\infty}dt e^{-t}\mathbb{P}_{s}\big[\varsigma\leq t  \big]  \nonumber \\
&\geq  e^{-T}- e^{-T}\sup_{H(s)>L}   \mathbb{P}_{s}\big[\varsigma > T  \big].
\end{align*}
If we show $ \sup_{H(s)>L}   \mathbb{P}_{s}\big[\varsigma > T  \big]$ is small for large $T$ (or even merely bounded away from one), then combining this result with~(\ref{Lambast}) completes the proof of statement (ii).

Define the function $w:\Sigma\rightarrow [\frac{1}{2},1]$ such that
$$  w(s)=1-\frac{1}{2}\frac{ 1 }{ 1+H^{\frac{1}{2}}(s)    } .  $$
There exists $L$ large enough so that for some $\delta>0$
\begin{align}\label{Sinbad}
\sup_{ H(x,p)>L    } \int_{\R}dp'\mathcal{J}(p,p' )\Big( w(x,p') -w(x,p)     \Big)\leq -\delta.   
\end{align}
The inequality~(\ref{Sinbad}) follows from the momentum-contractive nature of the jump rates $\mathcal{J}(p,p' )$ described below the statement of (ii), and we will not go through its proof.   The significance of the inequality ~(\ref{Sinbad}) is that $t\,\delta +w(S_{t})$ behaves as a supermartingale
over the time interval $[0,\varsigma]$ since for $s=(X_{t},P_{t})$
$$ \frac{d}{dr}\Big|_{r=0}\mathbb{E}_{s}\big[r\delta+ w(X_{r}, P_{r})\big]=\delta+ \int_{\R}dp'\mathcal{J}(P_{t},p')\Big( w(X_{t},p') -w(X_{t},P_{t})     \Big)\leq 0 .   $$
The equality above follows since the function $w$ is invariant of the Hamiltonian evolution.  We have the following sequence of inequalities:
\begin{align*} 
  \sup_{H(s)>L}\mathbb{P}_{s}\big[ \varsigma>T   \big] & \leq     \Big(\sup_{H(s)>L}\mathbb{P}_{s}\big[\varsigma> \delta^{-1} \big] \Big)^{\lfloor \delta T \rfloor}\leq  \Big(\sup_{H(s)>L} \delta\, \mathbb{E}_{s}\big[\varsigma\wedge \delta^{-1} \big] \Big)^{\lfloor \delta T \rfloor} \\ &\leq  \Big(\sup_{H(s)>L}\mathbb{E}_{s}\big[w(s)-w(S_{ \varsigma\wedge \frac{1}{\delta}   }) \big] \Big)^{\lfloor \delta T \rfloor}\leq  2^{-\lfloor \delta T \rfloor}    .   
  \end{align*}
For the first inequality, the event $\varsigma>T $ requires that the time $ \varsigma$ failed to occur over disjoint time intervals $( \frac{n}{\delta} ,\frac{n+1}{\delta }]$ for $0 \leq n\leq \lfloor \delta T \rfloor -1$.  The second inequality is Chebyshev's, and the third is by the optional stopping theorem since $w(s)-w(S_{t})- \delta \,t$ is submartingale over the time interval $ t\in [0,  \varsigma\wedge \frac{1}{ \delta}  ]$.  The last inequality follows from the constraint  $\frac{1}{2}\leq w(s)\leq 1$.   We can choose $T$ to make the right side arbitrarily small.

\vspace{.5cm}

\noindent (I).  The second equality below follows from an inductive argument using the fact that the probability of the event $\tilde{\sigma}_{n}\in \Sigma\times 1$ is $h(\sigma_{n})$ when conditioned on the information up to time $n-1$ and the value $\sigma_{n}$.  This property is visible from the form of the transition operator $\tilde{\mathcal{T}}$ given in Sect.~\ref{SecNum}.  
\begin{align*}
  \tilde{\mathbb{P}}_{\tilde{\Psi}}\big[ \tilde{n}_{1}\geq  m   \big]   & = \tilde{\mathbb{E}}_{\tilde{\Psi}}\big[ \chi(\tilde{n}_{1}\neq  0)\cdots  \chi(\tilde{n}_{1}\neq  m-1) \big]\\ & = \tilde{\mathbb{E}}_{\tilde{\Psi}}\big[ \big(1-h(\sigma_{0})\big)\cdots \big(1-h(\sigma_{m-1})\big)\big]
\\ & = \mathbb{E}_{\Psi}\big[\big(1- h(\sigma_{0}) \big)\cdots  \big(1-h(\sigma_{m-2})\big)\mathbb{E}\big[\big(1- h(\sigma_{m-1})\big)\,\big|\,\sigma_{0},\dots,\sigma_{m-2}   \big]\big]  \\ &\leq (1-\epsilon) \mathbb{E}_{\Psi}\big[\big(1- h(\sigma_{0})\big)\cdots  \big(1-h(\sigma_{m-2})\big)\big].
 \end{align*}
The third equality identifies the expectation as from the original statistics, and the inequality uses that the probability  $\sigma_{m-1}$ jumps from $\sigma_{m-2}$ to the support of $h$ is $\geq (1-\epsilon)$ by (i) and (ii).  Applying this inductively, we get the desired bound.

\vspace{.5cm}

\noindent (II).  By the renewal theorem, $F_{\Psi}(n)$ converges as $n\rightarrow \infty$ to the inverse of the expectation of the renewal increments  $ \tilde{n}_{m+1}-\tilde{n}_{m} $.   The expectation of the increments $ \tilde{n}_{m+1}-\tilde{n}_{m} $, $m\geq 1$ is given by  
$$\tilde{\mathbb{E}}_{\tilde{\Psi}}\big[ \tilde{n}_{m+1}-\tilde{n}_{m}  \big]= 1+\tilde{\mathbb{E}}_{\tilde{\nu}}\big[ \tilde{n}_{1} \big]= \tilde{\mathbb{E}}_{\tilde{\nu}}\Big[ \sum_{n=0}^{\tilde{n}_{1}}1 \Big]= \frac{   1 }{ \Psi_{\infty}(h) } $$   
since the distribution for the split chain following a return to the atom is $\tilde{\nu}$ by Part (1) of Proposition~\ref{IndependenceProp}, and where the second equality is by Part (1) of Proposition~\ref{BasicsOfNumII}.   The tails of the delay distribution $\tilde{\mathbb{P}}_{\tilde{\Psi}}\big[ \tilde{n}_{1}>n  \big] $ and the renewal jump distribution $\tilde{\mathbb{P}}_{\tilde{\nu}}\big[ \tilde{n}_{1}>n  \big] $ decay with order $\mathit{O}(e^{-\epsilon n})$  for large $n$ by  (I).  It follows that the renewal function $F_{\Psi}(n)$ converges exponentially with rate $\epsilon$ to the value $\Psi_{\infty}(h)$.

\end{proof}

\end{appendix}

\end{document}